\documentclass[journal,twocolumn,10pt]{IEEEtran}

\usepackage[normalem]{ulem} % enable strikethrough with \sout{}

\usepackage{epsfig,cite}
\usepackage{graphicx}
\usepackage{color}
\usepackage{amssymb,amsmath}
\usepackage{booktabs}

\usepackage{mathrsfs} % extra math fonts
\definecolor{blue}{rgb}{0,0,1}
\definecolor{darkgreen}{rgb}{0,.5,0}
\definecolor{darkred}{rgb}{.75,0,0}
\definecolor{red}{rgb}{1,0,0}

\newcommand{\bl}[1]{\textcolor{black}{#1}}
\newcommand{\dr}[1]{\textcolor{black}{#1}}

\newtheorem{theorem}{Theorem}

\newtheorem{lemma}{Lemma}

\newtheorem{definition}{Definition}
\newtheorem{cor}{Corollary}

\def\given{\:|\:}
\def\L{\mathsf{L}}
\def\H{\mathsf{H}}
\def\U{\mathsf{U}}
\def\B{\mathsf{B}}

\newcommand{\cfb}{C_{\mathrm{FB}}}
\newcommand{\ciid}{C_{\mathrm{IID}}}

\def\binent{\mathscr{H}}

\newcommand{\E}{\mathbb{E}}

\newcommand{\posta}{$\text{POST}(\alpha)$}

\begin{document}

\title{Capacity of a Simple Intercellular Signal Transduction Channel}%!PN

\author{
	Peter J. Thomas and Andrew W. Eckford,~\IEEEmembership{Senior Member, IEEE}%
	\thanks{Submitted to {\sc IEEE Transactions on Information Theory}, November 6, 2014; revised September 19, 2015 and April 5, 2016; accepted July 27, 2016.}
	\thanks{Material in this paper was presented in part at the 2013 IEEE International Symposium on Information Theory (ISIT 2013), Istanbul, Turkey.}
	\thanks{Peter J. Thomas is with the Department of Mathematics, Applied Mathematics, and Statistics, the 
	Department of Biology, the Department of Cognitive Science, and the Department of Electrical Engineering and Computer Science, Case Western Reserve University, Cleveland, Ohio, USA 44106-7058.
	Andrew W. Eckford is with the Department of Electrical Engineering and Computer Science,
	York University, 4700 Keele Street, Toronto, Ontario, Canada M3J 1P3. Emails:
	pjthomas@case.edu, aeckford@yorku.ca}
	\thanks{This work was supported by a grant from the Simons Foundation (259837 to Peter Thomas), by the Council for the International Exchange of Scholars (CIES), by the National Science Foundation (grants DMS-0720142, EF-1038677, DMS-1413770), by sabbatical support from Case Western Reserve University, and by the Natural Sciences and Engineering Research Council. 	}
	\thanks{Software used to generate some of the results in this paper can be downloaded from GitHub: http://github.com/andreweckford/CapacityOfSignalTransduction}%
}

\maketitle

\begin{abstract}
We model \bl{biochemical signal transduction, based on a ligand-receptor binding mechanism, as a discrete-time finite-state Markov channel, which we call the BIND channel.} 
We show how to obtain the capacity of this channel, \bl{for the case of binary output, binary channel state, and arbitrary finite input alphabets.} 
We show that the capacity-achieving input distribution is IID. Further, 
we show that feedback does not increase the capacity of this channel. 
We show how the capacity of the discrete-time channel  approaches the capacity of Kabanov's Poisson channel, in the limit of short time steps and rapid ligand release. 
\end{abstract}
%\tableofcontents

\section{Introduction}
\label{sec:introduction}
\subsection{Overview}

\IEEEPARstart{R}{esearch} at the intersection of biology and information theory stretches back almost to Shannon's founding papers, with notable work by Yockey \cite{yockey1958a,yockey1958b}, Attneave \cite{Attneave1954}, Barlow \cite{barlow1961}, and Berger \cite{berger1971}. After long remaining on the margins of the wider information theory community, this research area finds itself newly in the limelight due to the convergence of two recent trends. 
First, quantitative biologists increasingly apply information-theoretic methods to the analysis of high throughput, individually resolved laboratory data \cite{BrennanCheongLevchenko2012Science,LevchenkoNemenman2014CurrOpinBiotech};
second, ``mainstream'' information theorists increasingly explore biological applications (e.g., \cite{ein11b,coleman1,coleman2,son12,eck07}),
obtaining results in fields such as molecular biology and neuroscience. %\footnote{Work in the 2014 Hawaii statement here too?}  
These two trends have developed alongside increasing interest in the 
mathematical and conceptual foundations of biology, as well as interest in biologically-inspired communication systems \cite{EinolghozatiSardariFekri2012ISIT,FarsadGuoEckford2013PLoSOne}, further accentuating the importance of information theory in biology.

The current paper focuses on 
communication systems that employ chemical principles, broadly known as {\em molecular communication} \cite{hiy05}. Recent work in molecular communication can be divided into two categories.
In the first category, work has focused on the engineering possibilities:
to exploit molecular communication for specialized applications, such as nanoscale networking
\cite{hiy05,par09}. In this direction, information-theoretic work has focused on
the ultimate capacity of these channels, regardless of biological mechanisms (e.g., \cite{eck07,son12}).
In the second category, work has focused on analyzing the biological machinery of
molecular communication (particularly ligand-receptor systems), both to describe the components of a possible communication system
\cite{nak07} and to describe their capacity \cite{ata07,NIPS2003_NS03,ein11,ein11b,LiMoserGuo2014IEEE}. Our
paper, which builds on work presented in \cite{NIPS2003_NS03}, fits into the second category, and many tools in the information-theoretic literature can be used to
solve problems of this type. Related work is also found in \cite{ein11}, where capacity-achieving input distributions were 
found for a simplified ``ideal'' receptor; that paper also discusses but does not solve the capacity
for the channel model we use.

Our contribution in this paper is to prove several important properties of capacity
for a two-state signal transduction channel, \bl{which we refer to as the ``Binding IN Discrete time'' (BIND) channel}, as found in the {\em Dictyostelium} model organism \cite{Wang+Rappel+Kerr+Levine:2007:PRE,Ueda+Shibata:2007:BPJ} as well as in models of neural communication systems taking into account refractoriness or synaptic dynamics \cite{DegerHeliasCardanobileAtayRotter2010PRE}, \bl{\cite{DrosteLindner2014BiolCyb}}.
\bl{The BIND channel, introduced formally in \S \ref{sec:Capacity}, is a discrete-time analog of a ubiquitous biochemical signal-transduction mechanism, described in \S \ref{ssec:bio}.}
\bl{We show that the capacity-achieving input distribution of the BIND channel, which is a discrete-time Markov chain model, is IID,} with all the probability
weight on the minimum and maximum possible ligand concentrations. Further,
we show that feedback does not increase the capacity of the \bl{BIND} channel.
Finally, given an IID input distribution, we give a simple closed-form expression
for the mutual information, which can be maximized to find capacity.
In addition to the capacity results, we discuss the mutual information of the \bl{BIND} channel
when the channel inputs are Markov distributed, and we \bl{compare} our capacity results to earlier known
results on the capacity of Poisson counting channels. \bl{We focus on the capacity of a single
receptor, leaving the problem of multiple receptors to future work.}

\bl{Indecomposable discrete time finite state channels, of which the BIND channel is an example, have been studied extensively \cite{BlackwellBreimanThomasian1958AnnalsMathStat}.  Although the capacity is unknown for the general case, many special cases have been examined, some related
%\footnote{Better: some with \emph{potential} relevance...} 
to biological signaling.  The trapdoor channel, introduced by Blackwell \cite{Blackwell1961chapter}, has been generalized as a model for communication mediated by diffusion of chemical signals, 
%\cite{PermuterCuff_vanRoy_Weissman2007IEEE_ISIT}.  The 
feedback capacity and zero-error feedback capacity of which  
has been solved \cite{PermuterCuff_vanRoy_Weissman2007IEEE_ISIT}. Channels with internal states provide models for systems with memory effects, intersymbol interference, or both \cite[Ch.~4.6]{gal68}. In some cases, the capacity of finite state channels can be increased by feedback.  For example, feedback has been shown to increase the capacity for a class of finite state Markov channels in which the channel state transition probabilities are independent of the input
(see, e.g., \cite{AsnaniPermuterWeissman2014IEEETransIT}).
%\cite{GoldsmithVaraiya1996ieeeIT-markovchannels}.  
Finite state channel models for which feedback \emph{does not} increase capacity are 
therefore of interest.}

\bl{  
Berger, Chen, and Yin studied a general class of {\em unit output memory} (UOM) finite state channels for which feedback does not increase capacity \cite{che05,yin-unpublished}. In these models, 
the channel state and channel output are isomorphic, and the channel output is fed back to
the transmitter with unit time delay. A key feature of UOM channels is that the 
feedback-capacity-achieving input distribution has a simple form.
As we will show in \S \ref{sec:Capacity}, the BIND model falls within this class if feedback
is introduced, and we use this fact to show that the capacity and feedback capacity are the same.  
Relatedly, but distinctly, Permuter and colleagues introduced the {\em Prior Output is the STate} (POST) channels,
a class of UOM channels for which capacity and feedback capacity may be
readily evaluated  \cite{AsnaniPermuterWeissman2013IEEE_ISIT,PermuterAsnaniWeissman2014IEEETransIT}, again
showing that capacity and feedback capacity are the same for many POST channels. We discuss the distinctions and relationship between the BIND and the  \posta~and $\text{POST}(a,b)$ channels in \S \ref{ssec:POST}.}

%\bl{(New paragraph on relation to prior work.)  Trapdoor channel introduced by Blackwell ; generalizations referred to as diffusion or chemical channel is analogous to diffusion part of process, but models a different aspect than ligand-binding.  Although introduced in 1950s, its capacity is still unknown.  Its  capacity with feedback was solved by Permuter and colleagues....
%Ising channel.  What else?  Our BIND channel focuses on the second part of signal transduction.  Once a molecule diffuses through intracellular space, it must be counted by interacting with a receptor.  In the next section we discuss the biophysics underlying the BIND channel, which distinguish it from the \posta, the POST(a,b), the general POST, the billiard, and the Ising channels.
%We discuss relationship to the POST channel in more detail in \S \ref{sec:discuss}.}

\subsection{Biological Motivation}
\label{ssec:bio}
%\footnote{Another example in which a Poisson process with exponentially distributed dead time appears in a model for biological communication is \cite{DegerHeliasCardanobileAtayRotter2010PRE}, in which the dead time models the post--action-potential refractoriness of a simple model neuron (``Poisson process with random dead-time" or PPRD).}

As some readers of the Transactions may be unfamiliar with the details of biological signal transduction, 
we devote the remainder of the introduction to an overview of such systems.

Living cells communicate with one another through a web of biochemical interactions referred to as \emph{signal transduction networks}\cite{CaiDevreotes2011SeminCellDevBiol,Gough2005SciSTKE,RobbinsFeiRiobo2012SciSignal}.  
%Just as neural networks underly the interactions of many multicellular organisms with their environments, t
These biochemical networks allow individual cells to perceive, evaluate and react to chemical stimuli \cite{FedoroffFontana2002Science,Rappel+Thomas+Levine+Loomis:2002}.  Examples include chemical signaling across the synaptic cleft 
%gap 
connecting the axon of one nerve cell to the dendrite of another \cite{Jonas1993EXS}, calcium signaling within the postsynaptic spines of a dendrite \cite{KellerFranksBartolSejnowski2008PLoSOne}, pathogen localization by migratory cells in the immune system \cite{KhairDaviesDevalia1996EurRespirJ}, growth-cone guidance during neuronal development \cite{Goodhill2003NECO}, phototransduction in the retina \cite{Yau1994phototransduction}, and gradient sensing by the social amoeba \textit{Dictyostelium discoideum} \cite{Iglesias+Deverotes:2008:JCellBio}.

%Microorganisms communicate using {\em molecular communication}, in which messages are expressed as patterns of molecules, propagating via diffusion from transmitter to receiver. 
%What can information theory say about 
%this communication?
%The physics and mathematics of Brownian motion and chemoreception are 
%well understood \cite{karatzas-book, berg77}, so it is possible
%to construct channel models and calculate information-theoretic quantities, such as capacity \cite{CoverThomas1990,NIPS2003_NS03,NIPS2006}. 
%We expect that Shannon's channel coding theorem, and other limit theorems in 
%information theory, express ultimate limits on reliable communication, not just for 
%human-engineered systems, but for naturally occurring systems as well. We can hypothesize that evolutionary pressure may, in some instances, have optimized natural molecular communication
%systems with respect to these limits.  Calculating quantities such as
%capacity may allow us to make predictions about biological systems, and explain
%biological behaviour \cite{Berger2002ShannonLect,BrennanCheongLevchenko2012Science,CheongRheeWangNemenmanLevchenko2011Science,Thomas:2011:Science}.

Signal transduction at the cellular and subcellular level typically involves a complex macromolecular apparatus comprising multiple proteins. 
For example, transmission of neural signals often depends on diffusion of neurotransmitter molecules across a narrow gap (the synaptic cleft) to receptor proteins on the postsynaptic membrane.  These neurotransmitter receptors are connected to large protein ``signaling machines" \cite{sci:Kennedy:2000} that control the downstream effects of neurotransmitter signaling, including signaling mediated by the influx of extracellular calcium ions.  
%
%% is the below statement right?
%
In general, activation of a receptor will produce {\em second messengers} within
the cell, which control its behaviour.

In this paper we are most interested in the process at the receiving end
of signal transduction, where a signaling molecule (ligand molecule) binds to a receiver 
molecule (protein) at a destination cell.
Despite the apparent complexity of this process,
a key simplifying observation is that the receptor proteins are
driven through a finite 
series of {\em states} by the presence of signaling molecules \cite{ShifmanChoiEtAlKennedy2006PNAS}.  

A two-state example, where the receptor can be either bound to the ligand (signaling) molecule or else unbound,
is shown in Figure \ref{fig:LigandReceptor}: if the receptor is unbound, an available 
ligand can bind with it, changing its state; the receptor must then go through an ``unbinding'' process,
processing the ligand and reverting to the initial state, before it can bind with another ligand.
This two-state, bound-unbound 
receptor model is appropriate for the 3'-5'-cyclic adenosine monophosphate (cAMP)
receptor in the {\em Dictyostelium} amoeba, which is used as a model
organism for studies of signal transduction \cite{Devreotes1994Neuron,KleinSunSAxeKimmelJohnsonDevreotes1988Science,NIPS2003_NS03}. This is the simplest nontrivial example of a ligand binding to a receptor, 
and forms the basis for the results in this paper.

\begin{figure}[t!] %  figure placement: here, top, bottom, or page
   \centering
   \includegraphics[width=3.25in]{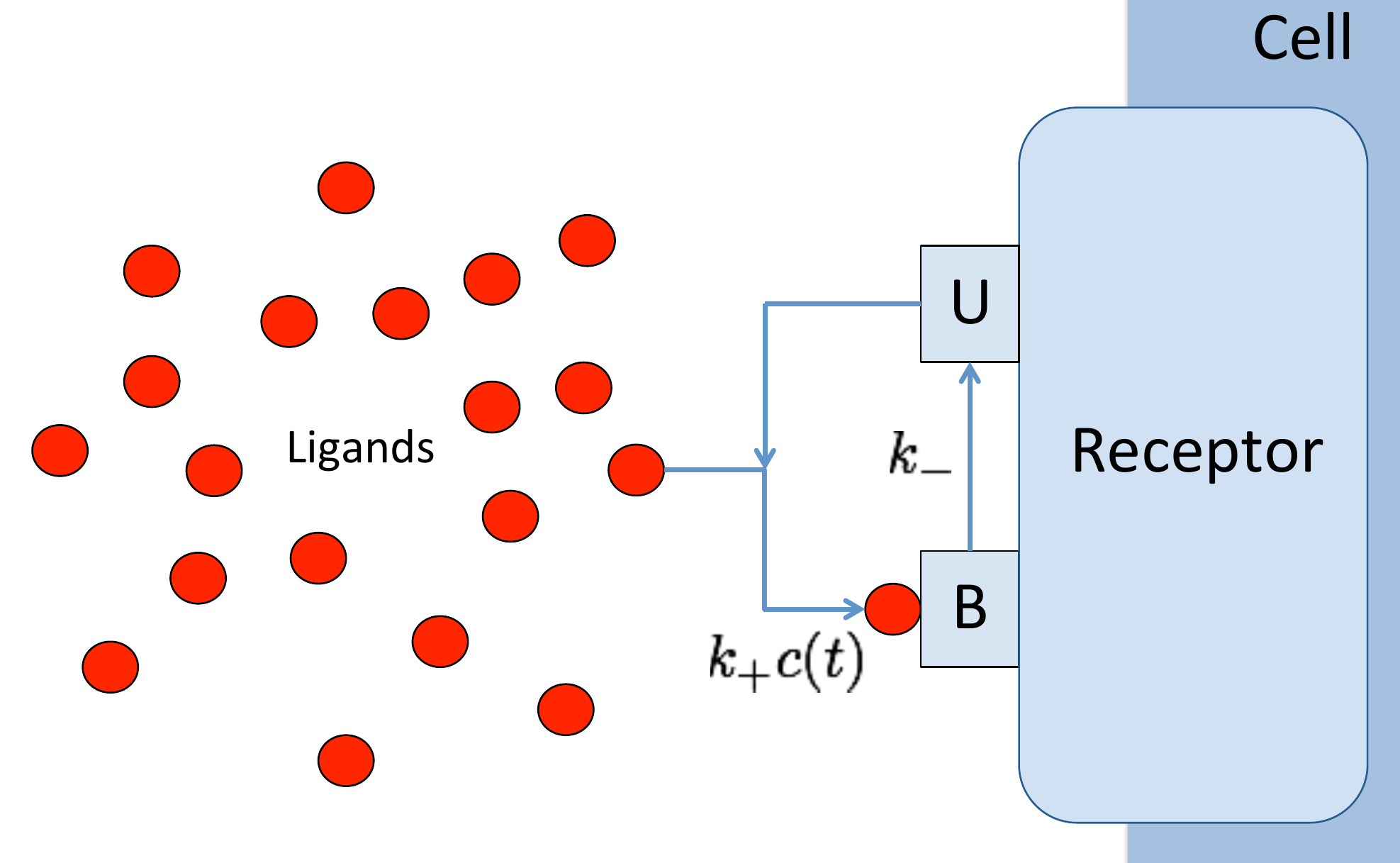} 
   \caption{An example of a two-state binding and unbinding process of the receptor. If the receptor is unbound (U), then a ligand can be absorbed by the receptor; this process occurs at rate $k_+ c(t)$, proportional to the ligand concentration $c(t)$. If the receptor is bound (B), it reverts to the unbound state
   at rate $k_-$, independent of ligand concentration. }
   \label{fig:LigandReceptor}
\end{figure}

%More generally, the receiver molecule has a finite set of states $\mathcal{S}=\{0,1,2,\cdots,N_S\}$, connected by instantaneous transition rates $q_{ij}$.  One or more of the transitions occur at a rate that depends on the local time-varying concentration $c(t)$ of a signaling molecule.   
%Viewed as a communications channel, the concentration $c(t)$ serves as the input signal, and the channel state (or some functional of the channel state) serves as the output signal.  

The complexity of these systems can be much higher.
In many instances, signal transduction molecules possess a number of sites at which ligand molecules can bind to the receiver protein.  A protein with $k$ binding sites, each either bound or unbound, can have $2^k$ distinct binding states. 

%The transition corresponding to a single binding event occurs with a rate proportional to the ligand concentration, so these transitions are sensitive to the input signal.  The reverse transition also occurs, in which one occupied binding site releases or degrades the bound ligand, and returns to the unbound, receptive state.  This transition happens at a constant rate, independent of the ligand concentration, so it is insensitive to the input signal and its occurrence conveys no information about the instantaneous signal level.  

The signal is expressed through the time-varying concentration $c(t)$ of ligand molecules, which affects the binding rate of the receptor (as in Figure \ref{fig:LigandReceptor}).
\bl{We assume throughout the paper that the binding of ligand molecules to a receptor protein obeys the familiar law of mass action \cite{GuldbergWaage1879JPraktischeChemie,vantHoffCohen1896book}, namely, that the rate of the reaction
\begin{equation}
\text{Ligand} + \text{Unbound Receptor} \to \text{Bound Receptor}
\end{equation}
proceeds at a rate proportional to the product of the concentration of the reactants, i.e. 
$$r_\text{binding}=k_+[\text{Ligand}][\text{Unbound Receptor}].$$
Here $k_+$ is the rate constant for the forward (binding) reaction, and [A] is the concentration of chemical species A, typically measured in nM ($10^{-9}$ moles per liter).  For the cyclic AMP (cAMP) molecule binding to the cAMP receptor in the \textit{Dictyostelium} amoeba, $k_+$ is on the order of $4\times 10^{-2}(\text{sec}\,\text{nM})^{-1}$.  
The reverse reaction also occurs:
\begin{equation}
\text{Bound Receptor} \to \text{Ligand} + \text{Unbound Receptor} 
\end{equation}
with rate 
$$r_\text{unbinding}=k_-[\text{Bound Receptor}].$$
Ueda \textit{et al.}~measured the distribution of binding durations of individual cAMP receptors and found the release time following binding is well approximated by an exponential waiting time distribution with rate $k_-\approx 1/\text{sec}$ \cite{UedaSakoTanakaDevreotesYanagida2001Science}.  The law of mass action thus dictates that the concentration of bound receptors $[\B]$ obeys the differential equation
\begin{equation}
\frac{d[\B]}{dt}=k_+[\text{Ligand}][\U]-k_-[\B].
\end{equation}
The signal available to the cell from its surroundings takes the form of the time varying ligand concentration, $c(t)=[\text{Ligand}]$.  This time-varying concentration serves as the input of the channel we consider.  For a given cell, the total number of bound and unbound receptors is a fixed constant, $[\U]+[\B]=[\text{Total}].$  Dividing by the total number of receptors, and setting $y(t)=[\B]/[\text{Total}]$ to be the fraction of receptors that are bound to ligand at time $t$, the law of mass action translates into a first order affine linear differential equation
\begin{equation}\label{eq:massaction3}
\frac{dy}{dt}=k_+c(t)(1-y) - k_-y.
\end{equation}
Such a differential equation is a simple example of a {\em chemical master equation} \cite{GeQian2013EncycSysBiol}.}

\bl{
We focus now on a single receptor, binding and releasing ligand independently of the other receptors.  At the single protein level mass action kinetics translates into a well established stochastic representation \cite{Higham2008SIREV}.  Let $p(t)$ be the probability that the receptor's binding site is occupied by a ligand molecule at time $t$.  Then $p(t)$ evolves according to a master equation of the same form as \eqref{eq:massaction3}
\begin{equation}\label{eq:cts-time-2state}
\frac{dp}{dt}=k_+c(t)(1-p) - k_-p.
\end{equation}
This system may naturally be viewed as a communications channel in which the input is the time varying concentration $c(t)$, and the output is the receptor state (bound or unbound).}

\bl{ 
The BIND channel, introduced in \S \ref{sec:Capacity}, is a discrete time analog of this system.  Both the continuous and discrete time versions of the ligand-binding channel share an important asymmetry.  When the receptor is in the unbound state, its transition rate is sensitive to the input signal (the ligand concentration).  When the receptor is in the bound state, it {\em cannot bind a second ligand molecule until releasing the one it has already bound}: the channel must leave the bound state before becoming sensitive to the input again.  This asymmetry reflects the different roles of the signal in reactions (1) and (2).  In the forward reaction (binding: reaction 1) the ligand is a \emph{reactant}, and the rate of reaction is proportional to the ligand's concentration.  In the backward reaction (unbinding: reaction 2) the ligand is instead a \emph{product}.  The reaction rate is a function of the reactant concentrations, not the product concentrations, so the unbinding reaction proceeds at an instantaneous rate that does not depend on the input signal concentration.  This asymmetry occurs naturally in any model of ligand-mediated biochemical signal transduction, but is absent from other binary channel models with memory, such as the trapdoor, Ising, Glenn-Elliott, or \posta~channels \cite{BergerBonomi1990IEEE,PermuterCuff_vanRoy_Weissman2007IEEE_ISIT,AsnaniPermuterWeissman2013IEEE_ISIT,Blackwell1961chapter,MushkinBarDavid1989IEEETransIT}. Thus ligand-receptor binding presents a novel, and intrinsically biological, type of communications channel.}

Basic mechanisms %principle 
of signal transduction have been known for decades \cite{SpringerGoyAdler1979Nature,HunterCooper1985AnnRevBioc}. However, recent technological advances have dramatically increased the ability to manipulate and measure the signals entering and leaving signal transduction networks at the molecular level.  These advances create an opportunity for quantitative understanding of molecular communication. For example, microfluidics combined with cell-by-cell single track measurements have been used to estimate the mutual information between a chemical gradient and the motile response of the \textit{Dictyostelium} amoeba \cite{FullerChenAdlerGroismanLevineRappelLoomis2010PNAS,HuChenLevineRappel2011JStatPhys}.  Single molecule fluorescence methods have allowed visualizing the binding and unbinding of signaling molecules to single receptors in real time \cite{UedaSakoTanakaDevreotesYanagida2001Science}.  High throughput measurements have led to sufficiently precise  capacity estimates, for a  cancer-related signaling network, to extract information about the network topology \cite{CheongRheeWangNemenmanLevchenko2011Science}.  Optogenetics methods have created a new paradigm for manipulating molecular communication devices using applied light sources \cite{Deisseroth2010NatureMethods_Optogenetics}.  
%And then there is the promise of nanotechnology -- say something else here? References such as \cite{JianrongYuqingNongyueXiaohuaSijiao2004BiotechAdv} or \cite{MannixEtAlIngber2007NatureNanotech}?}  
Thus, our results come at an opportune time for biological researchers, both in terms of their
analytical capabilities and their interest in exploiting information theory.  \bl{At the same time, development of novel communications models based on biological systems has been identified as an important  growth area by the information theory community.}\footnote{\bl{2014 Report of the IEEE Information Theory Society Committee on New Directions \cite{ITSOC_new_directions_report_2014}.}}

%The rate of transition from the unbound state to the bound state depends on the local concentration of the signaling molecule or ligand, while the rate of transition from the bound to the unbound state is independent of the ligand concentration.  We obtain the capacity of a discrete time realization of this channel, under the assumption that the ligand concentration remains bounded, and we show that the capacity is not increased by feedback.  We also analyze the behavior of the capacity expression in the small-time-step limit, and show that we can recover the formulas due to Kabanov \cite{Kabanov1978} and Davis \cite{Davis1980ieeeIT} for the capacity of the Poisson channel as a special case of our results.  

%Dictyostelium as an example of a two-state binding process

\subsection{Capacity problem for a general point process channel}
\label{ssec:generalcapacityproblem}
We now introduce a \bl{general description} of the continuous time signal transduction channel with arbitrary (bounded) scalar input and discrete output.  
Finite state Markov processes conditional on an input process provide models of signal transduction and communication in a variety of biological systems,  as detailed in the preceding section. 
%including chemosensation via ligand-receptor interaction [check for redundancy of these references!  Where do they belong? \cite{Wang+Rappel+Kerr+Levine:2007:PRE,Ueda+Shibata:2007:BPJ}\footnote{These papers consider quantifying the level of noise in signaling pathways.}; and dynamics of ion channels sensitive to signals carried by voltage, neurotransmitter concentration, or light \cite{Colquhoun+Hawkes:1983chapter,KellerFranksBartolSejnowski2008PLoSOne,NikolicLoizuDegenaarToumazou2010IntegrBiol}.  ]
Typically, a single ion channel, or receptor, is in one of $n$ states.  The states form a finite directed graph $\mathcal{Y}$ with $n$ vertices, with edges connecting states that intercommunicate through a conformational or chemical change, or ligand-binding/unbinding event.  The receptor performs a continuous-time random walk on the graph, with one or more transition rates being influenced by the external input signal, $X(t)$.  The input signal can be the concentration of a diffusing signaling molecule for a ligand-gated receptor; it can be the transmembrane electrical potential for a voltage-gated receptor \cite{SchmidtThomas2014JMN}.  

There is a rich literature on the use of master equations for representing stochastic chemical reactions \cite{Higham2008SIREV} and algorithms for generating sample trajectories \cite{Gillespie1977,AndersonKurtz2010Chapter,AndersonErmentroutThomas2015JCNS}.
In the master equation representation of a signal transduction channel, the instantaneous transition rate matrix $\mathbf{Q} = [q_{jk}]$ depends on the external input $X(t)$. The probability, $p_k$, that the channel is in state \bl{$Y(t)=k$ for some $k\in\mathcal{Y}$} evolves according to 
\begin{equation}\label{eq:cts-time-general}
\frac{dp_k}{dt}=\sum_{j=1}^n p_j(t)q_{jk}(X(t))
\end{equation}
where for  $(j\ne k)$, $q_{jk}\ge 0$ is the input-dependent rate at which the receptor transitions from state $j$ to state $k$, and $q_{jj}=-\sum_{k,k\ne j}q_{jk}$.
Taking $\{X(t)\}_{t=0}^T$ as the input, and the receptor state $Y(t)\in\mathcal{Y}$ as the output, gives a channel model, the capacity of which is of general interest.  

\bl{We emphasize that equation \eqref{eq:cts-time-2state} corresponds to \eqref{eq:cts-time-general} in the case $n=2$.  Let $\mathcal{Y}=\{\U,\B\}$ be the state graph and let $p_\U$ and $p_\B$ be the probability of the receptor being in the unbound and bound states, respectively.  Set 
\begin{equation}
q_{\U\B}(c(t))=k_+ c(t),\hspace{1cm}q_{\B\U}=k_-,
\end{equation}
and identify $X(t)=c(t)$ as the input signal.  The correspondence follows, since the equations \begin{align}
\frac{dp_{\B}}{dt}&=-p_\B q_{\B\U}+p_\U q_{\U\B}(X(t))\\
\frac{dp_{\U}}{dt}&=p_\B q_{\B\U}-p_\U q_{\U\B}(X(t))
\end{align}
are equivalent to equation \eqref{eq:cts-time-2state}, given $p_\B(0)+p_\U(0)\equiv 1$.
}

The  majority of biological signal transduction systems operate without regulation by a fast clock, i.e.~they operate as continuous-time stochastic systems.  Nevertheless, discrete time channel models arise as approximations to continuous time systems by fixing a small time step.  While most of our analysis falls in the discrete time framework, we discuss the relation to continuous time systems further in \S\ref{sec:continuoustime}.

A related classical Poisson channel was solved by Kabanov \cite{Davis1980ieeeIT,Kabanov1978}. 
%(RELOCATE THIS SECTION?)}
%
In the %rapid-unbinding limit, 
\bl{limit as $k_- \rightarrow \infty$},
in which 
transition from the bound state back to the unbound state is instantaneous, the ligand-binding channel becomes
a simple counting process, with the input encoded in the time varying intensity.  This situation
is exactly the one considered in Kabanov's analysis of the capacity of a Poisson channel, under a max/min intensity constraint \cite{Davis1980ieeeIT,Kabanov1978}.  For the Poisson channel, the capacity may be achieved by setting the input to be a two-valued random process fluctuating between the maximum and minimum intensities.  If the intensity is restricted to lie in the interval $[1,1+c]$, the capacity is \cite{Kabanov1978}
\begin{equation}\label{eq:Kabanov}
C_{\mbox{Kab}}(c)=\frac{(c+1)^{1+1/c}}{e}-\left(1+\frac{1}{c} \right)\ln(c+1).
\end{equation}
%
%\red{As shown in (SECTION?),} 
\bl{As shown by Wyner \cite{Wyner1988IEEE_IT_I,Wyner1988IEEE_IT_II},} 
Kabanov's formula may be obtained \bl{(nonrigorously)} by restricting the input to a two-state  discrete time process with \bl{IID} input $X(t)$ taking the values $X_\text{lo}=1$ and $X_\text{hi}=1+c$. 
%, with transitions $X_{lo}\to X_{hi}$ happening with probability $r$, and transitions $X_{hi}\to X_{lo}$ with probability $s$, per time step  \cite{Wyner1988IEEE_IT}.  Maximizing the mutual information with respect to $r$ and $s$, and taking the limit of small time steps, yields  (\ref{eq:Kabanov}). For completeness, we recapitulate this derivation in \S \ref{sec:Kabanov}.
In addition, Kabanov proved that the capacity of the Poisson channel cannot be increased by allowing feedback.

Kabanov's approach, focusing on instantaneous unbinding and restricted intensity, is not directly applicable to molecular signal transduction. However,
our long-term goal is to obtain expressions analogous to \eqref{eq:Kabanov} for the continuous-time systems \eqref{eq:cts-time-2state} and \eqref{eq:cts-time-general}.
As a first step, we  restrict attention to a discrete time analog of the two-state system \eqref{eq:cts-time-2state}. As we show in  \S\ref{ssec:ReductionToKabanov}, our channel model can be seen as a natural generalization of Kabanov's counting process channel model.

In the next section we \bl{define and analyze the BIND channel}, a two-state signal-transduction channel model in discrete time \bl{based on ligand-mediated biochemical signal transduction,} and we \bl{rigorously find} its capacity.

%
%
%Recent work on molecular communication can be divided into two categories.
%In the first category, work has focused on the engineering possibilities:
%to exploit molecular communication for specialized applications, such as nanoscale networking
%\cite{hiy05,par09}. In this direction, information-theoretic work has focused on
%the ultimate capacity of these channels, regardless of biological mechanisms (e.g., \cite{eck07,son12}).
%In the second category, work has focused on analyzing the biological machinery of
%molecular communication (particularly ligand-receptor systems), both to describe the components of a possible communication system
%\cite{nak07} and to describe their capacity \cite{ata07,NIPS2003_NS03,ein11,ein11b}. Our
%paper, which builds on work presented in \cite{NIPS2003_NS03}, fits into this category, and many tools in the information-theoretic literature can be used to
%solve problems of this type. Related work is also found in \cite{ein11}, where capacity-achieving input distributions were 
%found for a simplified ``ideal'' receptor; that paper also discusses but does not solve  the capacity
%for the channel model we use.

\section{Capacity of the \bl{BIND channel, a} discrete intercellular signal transduction channel}
\label{sec:Capacity}

In this section, we \bl{introduce the BIND channel and} prove our main capacity results.
A roadmap for these results is given as follows:
\begin{enumerate}
	\item In (\ref{eqn:ClosedFormMI}), 
	we give a closed-form expression for the mutual information \bl{of the BIND channel} if the inputs are IID. 
	This expression can be 
	maximized to find capacity with input distribution constrained to be IID (the {\em IID capacity}, 
	see (\ref{eqn:CIID})), 
	but IID capacity is not available in closed form.
	\item In Theorem \ref{thm:maxmin}, we show that
	the IID capacity \bl{of the BIND channel}
	is achieved when only the minimum and maximum possible ligand concentrations
	are used, and no intermediate concentration. 
	(This theorem is given first as it simplifies the proof of the main result.)
	\item In Theorem \ref{thm:main}, 
	we then show that capacity  \bl{of the BIND channel}  is achieved by
	the IID input distribution, with inputs only on minimum and maximum ligand concentration. 
	We do so by 
	showing that the feedback capacity of the channel is satisfied by an IID
	input distribution, relying on the important results on feedback capacity
	from \cite{yin03,che05}.
	\item Finally, in
	Corollary \ref{cor:capacity},
	combining Theorems \ref{thm:maxmin} and \ref{thm:main} with Equation
	(\ref{eqn:CIID}),
	IID capacity is given by (\ref{eqn:CIIDLH}). An illustration is given
	in Figure \ref{fig:mutualInfoPlot}.
%	\item In \S \ref{sec:CapMaxParameters}, we show how capacity
%	behaves as a function of the system parameters, and give the parameters
%	that maximize capacity.
%	\item Finally, in \S \ref{sec:Counterexample}, we consider
%	whether capacity and feedback capacity are equal for signal transduction channels
%	with more than two receptor 
%	states. In general they are not, which we illustrate with a counterexample.
%	As a corollary,
%	since feedback capacity upper bounds Shannon capacity,
%	the Shannon capacity must be achieved by an IID input distribution with inputs 
%	only on $x = c_\L$ and $x = c_\H$. 
\end{enumerate}

\subsection{Discrete-input, discrete-time model}
\label{ssec:mathmodel}

\bl{The BIND channel is} a discrete-time, two-state %finite-state 
Markov channel representation of the signal reception process, 
an example of which may be found in the {\em Dictyostelium} cAMP receptor.  
The channel input, channel output, and input-output relationship are described
as follows.

{\em Channel input.} 
The channel input is the local concentration of ligands at the receptor:
at the interface between the receptor and the environment,
the receptor is sensitive to the concentration of ligands, binding more frequently as 
concentration increases.
We assume that the input concentration $c_j$ is one of $m$ discrete levels, and
without loss of generality, we will assume
$c_1 \leq c_2 \leq \ldots \leq c_m$. The {\em lowest} concentration $c_1$ 
and {\em highest} 
concentration $c_m$ are especially important in our analysis; we will give them the special symbols
$c_\L := c_1$ and $c_\H := c_m$.
Thus, the input (concentration) 
alphabet is 
\begin{equation}
	\label{eqn:ConcentrationAlphabet}
	\mathcal{X} = \{ c_\L,c_2,c_3,\ldots,c_{m-1},c_\H\} .
\end{equation}
Further, let $X_0^{n-1} = [X_0, X_1, \ldots, X_{n-1}] \in \mathcal{X}^n$  represent a sequence of inputs to the receptor.%
\footnote{We offset the indices of input and output because the pair $(x_i,y_i)$ 
can jointly form a Markov chain; see the discussion in the next section. Thus, it is more natural
for the input $x_i$ to affect the output $y_{i+1}$.}

{\em Channel output.}
The channel output is the state of the receptor.\bl{\footnote{\bl{In \cite{AsnaniPermuterWeissman2013IEEE_ISIT} Permuter and colleagues discuss finite state channels in which the internal state of the receptor is identified with its output.  The BIND channel falls in this general class, although it is distinct from the specific examples discussed in \cite{AsnaniPermuterWeissman2013IEEE_ISIT,PermuterAsnaniWeissman2014IEEETransIT}.  For further discussion see \S \ref{ssec:POST}.}}}
As in Figure \ref{fig:LigandReceptor}, 
the receptor may either be in an unbound state, in which the receptor is waiting for a molecule to bind; or in a bound state, in which the receptor has captured a molecule, and cannot  capture another until the molecule is degraded or released. Thus, channel output is binary: let $\mathcal{Y} = \{\U,\B\}$ represent the output alphabet, where $\U$ represents
the unbound state %(in which the receiver is waiting for a molecule to arrive), 
and $\B$ represents the bound state. %(in which the receiver is processing a molecule). 
Further, let $Y_1^n = [Y_1, Y_2, \ldots, Y_n] \in \mathcal{Y}^n$ represent a sequence of receptor states.
(Note the offset of index compared with $X_0^{n-1}$, clarified in the diagram below.)

{\em Input-output relationship.}
\bl{Figure \ref{fig:state_transition_figure} illustrates the state transitions of the BIND model.} 
\begin{figure}[t] %  figure placement: here, top, bottom, or page
   \centering
   \includegraphics[width=3.25in]{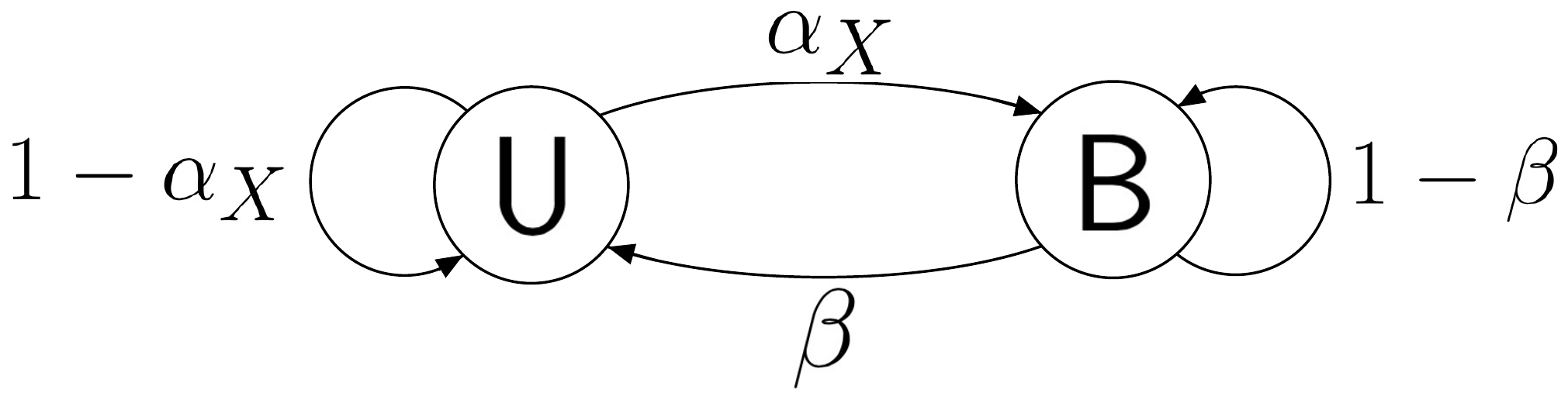} 
   \caption{\bl{State transition diagram for the BIND model.  Channel state $Y\in\{\U,\B\}$ determines sensitivity to input ($\U$: Unbound receptor, sensitive to input.  $\B$: Bound receptor, insensitive to input).  Ligand concentration $c_X$ is the input, with $c_\L\le c_X\le c_\H$. ($c_\L$ = lowest allowed input concentration. $c_\H$ = highest input concentration.)  Input binding probability $\alpha_X=k_+c_X\Delta t$ depends on the input concentration.   Unbinding probability $\beta=k_-\Delta t$ is independent of input concentration.}}
   \label{fig:state_transition_figure}
\end{figure}
The dependencies of the transition probabilities can be illustrated graphically as follows:
$$\begin{array}{ccccccccccc}
&X_0&&X_1&&X_2&&X_3&&X_4&\cdots\\
&\downarrow&&\downarrow&&\downarrow&&\downarrow&& \downarrow   \\
Y_0&\longrightarrow&Y_1&\longrightarrow&Y_2&\longrightarrow&Y_3&\longrightarrow&Y_4&\cdots
\end{array}$$
The state of the receptor is dependent on the 
previous input and the previous state, forming a Markov transition PMF
$p_{Y_i | X_{i-1},Y_{i-1}}(y_i \given x_{i-1},y_{i-1})$. Following the discussion in the previous section,
if $Y_{i-1} = \U$, i.e. the
receptor was previously unbound, then the distribution of $Y_i$ 
depends on the input concentration $X_{i-1}$.
However, if $Y_{i-1} = \B$,
then $Y_i$ 
is independent of $X_{i-1}$.

Thus, the Markov transition PMF $p_{Y_i | X_{i-1},Y_{i-1}}(y_i | x_{i-1},y_{i-1})$
has
$m+1$ parameters:
the $m$-dimensional vector $\alpha = [\alpha_1, \alpha_2, \ldots, \alpha_m]$ of
binding rates, where 
\begin{equation} 
	\label{eqn:AlphaDefinition}
	\alpha_j := p_{Y_i | X_{i-1},Y_{i-1}}(\B \given c_j, \U) ;
\end{equation}
and
$\beta$, the unbinding rate, independent of input signal concentration, where
\begin{equation}
	\label{eqn:BetaDefinition}
	\beta := p_{Y_i | X_{i-1},Y_{i-1}}(\U \given c_j, \B) 
\end{equation}
which is constant for all $c_j \in \mathcal{X}$.
This may also be written as a state transition probability matrix
\begin{equation}
	\label{eqn:StateTransitionProbabilityMatrix}
	\mathbf{P}_{Y|X = c_j} =
	\left[
		\begin{array}{cc}
			1 - \alpha_j & \alpha_j \\
			\beta & 1-\beta
		\end{array}
	\right] .
\end{equation}
Recalling the notation from (\ref{eqn:ConcentrationAlphabet}), 
we write $\alpha_\L$ and $\alpha_\H$ for the lowest and highest binding
rates, respectively. Thus, we can write $\alpha = [\alpha_\L, \alpha_2, \ldots, \alpha_{m-1}, \alpha_\H]$.
%\bl{The parameters $\alpha$ and $\beta$ are related to the rates $k_+$ and $k_-$ (from (\ref{eq:cts-time-2state})) as follows: given a discrete time step $\Delta t$, we have $\alpha_i =  k_+ c_i\Delta t$, and $\beta =  k_-\Delta t$. The time step $\Delta t$ has to be small enough that $\mathbf{P}_{Y|X=c_j}$ is a valid transition probability matrix.}

\bl{To relate this system to the master equation in (\ref{eq:cts-time-2state}), time is discretized into steps of length $\Delta t$. The parameters $\alpha$ and $\beta$ are then obtained from the rates $k_+$ and $k_-$ via $\alpha_i =  k_+ c_i\Delta t$, and $\beta =  k_-\Delta t$. The time step $\Delta t$ has to be small enough that $\mathbf{P}_{Y|X=c_j}$ is a valid transition probability matrix.}

From the above discussion, the sequence $Y_1^n$, 
given $X_1^n$ and initial input/output pair $X_0,Y_0$,
forms a time-inhomogeneous Markov chain with PMF
%\footnote{Should the first term be $p_{Y_1}(y_1)$ instead of $p_{Y_1|X_1}(y_1 \given x_1)$?  Do we specify an initial joint distribution $p(y_1,x_1)$ to begin the process? }
%
\begin{equation}
	\label{eqn:TimeInhomogeneousPMF}
	p_{Y_1^n|X_0^{n-1},Y_0}(y_1^n \given x_0^{n-1},y_0) = \prod_{i=1}^{n} 
	p_{Y_i | X_{i-1},Y_{i-1}}(y_i \given x_{i-1},y_{i-1}) .
\end{equation}

We give the following expressions and definitions, which will be useful in the remainder of this section.
For an IID input distribution $p_X(x)$,
since there are $m$ possible values for $x$, we will express $p_X(x)$ as a vector $p$,
with elements 
\begin{align}
	\label{eqn:VectorIID1}
	p &= [p_1,p_2,\ldots,p_m] \\
	\label{eqn:VectorIID2}
	&= [p_\L,p_2,p_3,\ldots,p_{m-1},p_\H] .
\end{align}
For the IID input distribution vector $p$, let $\bar{\alpha}_p$ represent
the average binding probability, given by
\begin{equation}
	\label{eqn:AlphaBar}
	\bar{\alpha}_p = \sum_{j=1}^m \alpha_j p_j .
\end{equation}
Finally, we give a condition on the parameters that will be used in many of our results:
\begin{definition}[Strictly Ordered Parameters]
	\label{defn:StrictlyOrdered}
	The parameters $\alpha$ and $\beta$ are said to be {\em strictly ordered} 
	if they satisfy
	\begin{equation}
		0 < \alpha_\L < \alpha_2 < \alpha_3 < \ldots < \alpha_{m-1} < \alpha_\H < 1 
	\end{equation}
	and
	\begin{equation}
		0 < \beta < 1 .
	\end{equation}
\end{definition}

\subsection{Mutual information and capacity under IID inputs}

Let $C$ represent the Shannon capacity of the system; as \bl{the BIND channel} is a channel with memory,
capacity is defined by
\begin{equation}
	\label{eqn:Capacity}
	C = \lim_{n \rightarrow \infty} \max_{p_{X_0^{n-1}}(x_0^{n-1})}  \frac{1}{n} I(X_0^{n-1};Y_1^n) ,
\end{equation}
Let $\ciid$ represent the capacity from (\ref{eqn:Capacity}) where $p_{X_0^{n-1}}(x_0^{n-1})$ 
is constrained to be IID, i.e., we can write
$p_{X_0^{n-1}}(x_0^{n-1}) = \prod_{i=0}^{n-1} p_X(x_i)$.

If the input distribution is IID, then $Y_1^n$ forms a time-homogeneous Markov chain \cite{che05}.
To see this, again assuming for convenience that $y_0$ is given, we start with
\begin{align}
	\nonumber\lefteqn{p_{Y_1^n,X_0^{n-1}|Y_0}(y_1^n, x_0^{n-1}\given y_0)}&\\ 
	&= p_{Y_1^n|X_0^{n-1},Y_0}(y_1^n \given x_0^{n-1},y_0) p_{X_0^{n-1}}(x_0^{n-1}) \\
	&= p_{Y_1^n|X_0^{n-1},Y_0}(y_1^n \given x_0^{n-1},y_0) \prod_{i=0}^{n-1} p_X(x_i) \\
	\label{eqn:YStationaryMarkov}
	&= 
	\prod_{i=1}^n p_{Y_i|X_{i-1},Y_{i-1}}(y_i \given x_{i-1}, y_{i-1})p_X(x_{i-1}) ,
\end{align}
where (\ref{eqn:YStationaryMarkov}) follows from (\ref{eqn:TimeInhomogeneousPMF}).
%
%\begin{equation}
%	p_{Y^n|X^n}(y^n \given x^n) = 
%	p_{Y_1|X_1}(y_1 \given x_1) \prod_{i=2}^n p_{Y_i|X_i,Y_{i-1}}(y_i \given x_i, y_{i-1}) .
%\end{equation} % warning: different expression given in (6)
%
%where $p_{Y_i|X_i,Y_{i-1}}(y_i \given x_i, y_{i-1})$ is given by the transition probability matrices
%(\ref{eqn:LowConcentrationOutput}).
%
Continue by letting 
%%
%\begin{equation}
%	p_{Y_1}(y_1) = \sum_{x_1}p_{Y_1|X_1}(y_1 \given x_1)p_X(x_1)
%\end{equation}
%%
%and 
%%
\begin{align}
	\nonumber\lefteqn{p_{Y_i | Y_{i-1}}(y_i \given y_{i-1})}&\\ 
	&= \sum_{x_{i-1}} p_{Y_i|X_{i-1},Y_{i-1}}(y_i \given x_{i-1}, y_{i-1})p_X(x_{i-1}).
\end{align}
Finally, marginalizing over $X_0^{n-1}$,
\begin{eqnarray}
	\nonumber\lefteqn{p_{Y_1^n | Y_0}(y_1^n \given y_0)} & & \\ 
	& = & \sum_{x_0^{n-1}} p_{Y_1^n,X_0^{n-1}|Y_0}(y_1^n, x_0^{n-1}\given y_0) \\
	& = & \prod_{i=1}^n \sum_{x_{i-1}} p_{Y_i|X_{i-1},Y_{i-1}}(y_i \given x_{i-1}, y_{i-1})p_X(x_{i-1}) \\
	\label{eqn:YStationaryMarkov2}
	& = & \prod_{i=1}^n p_{Y_i | Y_{i-1}}(y_i \given y_{i-1}) ,
\end{eqnarray}
which is the distribution of a time-homogeneous Markov chain. 
%
%We now give the transition probabilities and stationary probabilities in 
%(\ref{eqn:YStationaryMarkov2}). 
%We will express the iid input distribution
%$p_X(x_i)$ as the vector $p$, as in (\ref{eqn:VectorIID1})-(\ref{eqn:VectorIID2}).
If $y_{i-1} = \U$, 
\begin{align}
	p_{Y_i | Y_{i-1}}(\B \given \U)
	&= \sum_{x_i} p_{Y_i|X_{i-1},Y_{i-1}}(\B \given x_{i-1}, \U)p_X(x_{i-1}) \\
	&= \sum_{j=1}^m \alpha_j p_j \\
	\label{eqn:Transitions1}
	&= \bar{\alpha}_p ,
\end{align}
with $p_{Y_i | Y_{i-1}}(\U \given \U) = 1-\bar{\alpha}_p$. If $y_{i-1} = \B$, 
\begin{align}
	p_{Y_i | Y_{i-1}}(\U \given \B)
	&= \sum_{x_{i-1}} p_{Y_i|X_{i-1},Y_{i-1}}(\U \given x_{i-1}, \B)p_X(x_{i-1}) \\
	&= \sum_{j=1}^m \beta p_j \\
	\label{eqn:Transitions2}
	&= \beta ,
\end{align}
with $p_{Y_i | Y_{i-1}}(\B \given \B) = 1 - \beta$.
The transition probability matrix for $Y$ is
given by
\begin{equation}
	\label{eqn:YTransitionProbabilityMatrix}
	\mathbf{P}_Y =
	\left[
		\begin{array}{cc}
			1 - \bar{\alpha}_p  & \bar{\alpha}_p \\
			\beta & 1-\beta
		\end{array}
	\right] .
\end{equation}

Suppose the parameters are strictly ordered (Definition 1). Then $Y$ has
a stationary distribution, by inspection of (\ref{eqn:YTransitionProbabilityMatrix}).
The stationary probability of state $\U$ is given by
\begin{equation}
	\label{eqn:StationaryUnbound}
	p_Y(\U) = 
	\frac{1}{1+\bar{\alpha}_p/\beta} ,
\end{equation}
and $p_Y(\B) = 1-p_Y(\U)$.

When $Y_1^n$ is a time-homogeneous Markov chain, we can write the mutual information rate as
\begin{align}
	\mathcal{I}(X;Y) &= \lim_{n \rightarrow \infty} \frac{1}{n} I(X_0^{n-1};Y_1^n) \\
	\label{eqn:iidMutualInformationRate}	
	&= H(Y_i \given Y_{i-1}) - H(Y_i \given X_{i-1},Y_{i-1}) 
\end{align}
for any $i \in \{1,2,\ldots,n\}$. Let 
\begin{equation}
	\label{eqn:BinaryEntropyFunction}
	\binent(p) = - p \log p - (1-p) \log (1-p)
\end{equation}
represent the binary
entropy function.
Dealing with each term on the right hand side of (\ref{eqn:iidMutualInformationRate})
individually,
\begin{align}
	\nonumber\lefteqn{H(Y_i \given Y_{i-1})}&\\
	&= p_{Y}(\U) H(Y_i \given Y_{i-1}=\U) + p_Y(\B) H(Y_i \given Y_{i-1}=\B) \\
	&= p_Y(\U)\binent(\bar{\alpha}_p) + p_Y(\B) \binent(\beta) 
\end{align}
which follows from (\ref{eqn:Transitions1})-(\ref{eqn:Transitions2}); and
\begin{align}
	\nonumber\lefteqn{H(Y_i \given X_{i-1},Y_{i-1})}&\\
	=& \sum_{x_{i-1}} p_X(x_{i-1}) 
	p_{Y}(\U) H(Y_i \given X_{i-1}=x_{i-1},Y_{i-1}=\U) \nonumber \\ 
	& +  \sum_{x_{i-1}} p_X(x_{i-1}) p_Y(\B) H(Y_i \given X_{i-1}=x_{i-1},Y_{i-1}=\B)  \\
	=&\ p_Y(\U) \sum_{j=1}^m p_j \binent(\alpha_j) 
		+ p_Y(\B) \binent(\beta).
\end{align}
%
%which follows from (\ref{eqn:AlphaDefinition})-(\ref{eqn:BetaDefinition}).
Then the mutual information rate is given by
\begin{align}
	\mathcal{I}(X;Y) 
	&= H(Y_i \given Y_{i-1}) - H(Y_i \given X_{i-1},Y_{i-1}) \\
	&= p_Y(\U) \left( \binent(\bar{\alpha}_p) - \sum_{j=1}^m p_j \binent(\alpha_j) \right) \\
	\label{eqn:ClosedFormMI}
	&= \frac{\binent(\bar{\alpha}_p) 
		- \sum_{j=1}^m p_j \binent(\alpha_j)}
		{1 + \bar{\alpha}_p/\beta} .
\end{align}
Finally, $\ciid$ is given by
\begin{equation}
	\label{eqn:CIID}
	\ciid = \max_{p} \frac{\binent(\bar{\alpha}_p) 
		- \sum_{j=1}^m p_j \binent(\alpha_j)}
		{1 + \bar{\alpha}_p/\beta} .
\end{equation}

\subsection{$\ciid$ is achieved with all probability mass on $x=\L,\H$}
 
Here we show that the $\ciid$-achieving input distribution $p_X(x)$ uses only
the extreme values of concentration: $\L$ and $\H$.
The result is stated as follows.
\begin{theorem}
	\label{thm:maxmin}
	Let $p^* = [p_1^*, p_2^*, \ldots, p_m^*]$ represent an IID distribution that
	maximizes (\ref{eqn:CIID}). 
	If the parameters are strictly ordered (see Definition \ref{defn:StrictlyOrdered}), 
	then $p_2^* = p_3^* = \ldots = p_{m-1}^* = 0$.
\end{theorem}

\begin{proof}
	The proof proceeds by
	contradiction.			
	Assume the theorem is false: that $p_i^* > 0$ for at least one index $i$ in 
	$\{2,3,\ldots,m-1\}$. 
	Let $u$ represent the smallest index in $\{2,3,\ldots,m-1\}$ such that
	$p_u^* > 0$. From the initial assumption, $u$ must exist, and $\alpha_u$ is the corresponding
	binding probability.
	
	Since $\alpha_1 < \alpha_u < \alpha_m$, 
	there exist constants $\pi_1$ and $\pi_m$ such that 
	\begin{align}
		\label{eqn:InBetween}
		0 < \pi_1,\pi_m &< 1 \\
		\pi_1 + \pi_m &= 1 \\
		\label{eqn:AlphaUDef}
		\pi_1 \alpha_1 + \pi_m \alpha_m &= \alpha_u .
	\end{align}
	
	Let $q = [q_1,\ldots,q_m]$ represent a distribution constructed as follows:
	\begin{align}
		\label{eqn:QDef1}
		q_1 &= p_1^* + p_u^*\pi_1 \\
		q_m &= p_m^* + p_u^*\pi_m \\
		\label{eqn:QDef1a}
		q_u &= 0 \\
		\label{eqn:QDef2}
		q_j &= p_j^* \:\:\:\: \forall j \neq \{1,u,m\} .
	\end{align}
	Note that $q$ is constructed so that $\bar{\alpha}_q = \bar{\alpha}_{p^*}$ (see (\ref{eqn:AlphaBar})).
	
	From (\ref{eqn:ClosedFormMI}),
	the mutual information under distribution $p^*$ is
	\begin{equation}
		\label{eqn:MIUnderPStar}
		\mathcal{I}_{p^*}(X;Y) =
		\frac{
			\binent(\bar{\alpha}_{p^*}) - \sum_{i=1}^m p_i^* \binent(\alpha_i)
		}
		{
			1 + \bar{\alpha}_{p^*} / \beta
		} ,
	\end{equation}
	and under distribution $q$ it is (recalling $\bar{\alpha}_q = \bar{\alpha}_{p^*}$)
	\begin{equation}
		\label{eqn:MIUnderQ}
		\mathcal{I}_{q}(X;Y) =
		\frac{
			\binent(\bar{\alpha}_{p^*}) - \sum_{i=1}^m q_i \binent(\alpha_i)
		}
		{
			1 + \bar{\alpha}_{p^*} / \beta
		} .
	\end{equation}
	Equations (\ref{eqn:MIUnderPStar})-(\ref{eqn:MIUnderQ}) differ only in the term 
	under summation.
	We can write
	\begin{align}
		\nonumber 
		\lefteqn{\sum_{i=1}^m p_i^* \binent(\alpha_i) 
			- \sum_{i=1}^m q_i \binent(\alpha_i) } & \\
		\nonumber
		=&
		\Big( p_1^* \binent(\alpha_1) 
			+ p_u^* \binent(\alpha_u) + p_m^* \binent(\alpha_m) \Big)\\
		\label{eqn:ConcaveDifference1} & - \Big( q_1 \binent(\alpha_1) 
			+ q_u \binent(\alpha_u) + q_m \binent(\alpha_m) \Big) \\
		\label{eqn:ConcaveDifference2}
		=& p_u^* \bigg( \binent(\alpha_u) 
			- \Big( \pi_1 \binent(\alpha_1) + \pi_m \binent(\alpha_m) \Big) \bigg)\\
		\label{eqn:ConcaveDifference}
		=& p_u^* \bigg( \binent(\pi_1 \alpha_1 + \pi_m \alpha_m) 
			- \Big( \pi_1 \binent(\alpha_1) + \pi_m \binent(\alpha_m) \Big) \bigg) ,
	\end{align}
	where (\ref{eqn:ConcaveDifference1}) follows from (\ref{eqn:QDef2}),
	(\ref{eqn:ConcaveDifference2}) follows from
	(\ref{eqn:QDef1})-(\ref{eqn:QDef1a}), and 
	(\ref{eqn:ConcaveDifference}) follows from (\ref{eqn:AlphaUDef}).

	$\binent$ is strictly concave, $0 < \pi_1,\pi_m < 1$ (from (\ref{eqn:InBetween})), and
	$\alpha_1 < \alpha_m$ (by assumption), so
	(\ref{eqn:ConcaveDifference}) is always positive.
	This implies that
	\begin{equation}
		\mathcal{I}_{p^*}(X;Y) < \mathcal{I}_q(X;Y) .
	\end{equation}
	However, $p^*$ is the maximizing IID input distribution (by definition), which is a
	contradiction. The theorem follows.
\end{proof}

\subsection{Capacity and feedback capacity are achieved by an IID input distribution}

%We begin by defining {\em feedback capacity} $\cfb$ (i.e., capacity where the transmitter 
%has causal knowledge of the channel outputs), which is central to our capacity
%result in this section.
%%$\cfb$ is defined using directed information:

The directed information \cite{Massey1990ieee-conf-IT} between vectors $X_0^{n-1}$ and $Y_1^n$, written $I(X_0^{n-1} \rightarrow Y_1^n)$, is
given by 
\begin{equation}
	\label{eqn:DirectedInfo}
	I(X_0^{n-1} \rightarrow Y_1^n) = \sum_{i=1}^n I(X_0^{i-1};Y_i \given Y_1^{i-1}) .
\end{equation}
The per-symbol directed information rate is given by
\begin{equation}
	\lim_{n \rightarrow \infty} \frac{1}{n} I(X_0^{n-1} \rightarrow Y_1^n) .
\end{equation}

\bl{We use Kramer's double-bar notation for causal-conditional distributions \cite{Kramer2003IEEE_Trans_IT}. 
In the form we require in this paper,
\begin{equation}
	\label{eqn:CausalConditional}
	p(x_0^{n-1} || y_1^{n-1}) = \prod_{k=0}^{n-1} 
	p_{X_k|X_0^{k-1},Y_1^k}(x_k \given x_0^{k-1}, y_1^k) ,
\end{equation}
where vectors $x_0^{-1}$ and $y_1^0$ are null.}
Let $\mathcal{P}$ represent the set of causal-conditional feedback input distributions, i.e., 
%\marginpar{Why ``0,1'' subscripts?}
$p_{X_0^{n-1}|Y_1^{n-1}}(x_0^{n-1}\given y_1^{n-1}) \in \mathcal{P}$ if and only if
$p_{X_0^{n-1}|Y_1^{n-1}}(x_0^{n-1}\given y_1^{n-1}) = p(x_0^{n-1} || y_1^{n-1})$.

In our channel, $Y_1^n$ forms both the channel output and the channel state; therefore,
the feedback received by the transmitter is the channel state.
\bl{Following \cite{che05}, in finite state channels where the channel state is the channel output,
and where the transmitter receives this output (causally) as feedback,} the feedback capacity $\cfb$ is given by
\begin{equation}
	\label{eqn:FeedbackCapacity}
	\cfb = \max_{p_{X_0^{n-1}|Y_1^{n-1}}(x_0^{n-1}\given y_1^{n-1}) \in \mathcal{P}}
	\left(
		\lim_{n \rightarrow \infty} \frac{1}{n} I(X_0^{n-1} \rightarrow Y_1^n) 
	\right).
\end{equation}

Our capacity result is stated as follows.
\begin{theorem}
	\label{thm:main}
	%For the intercellular signal transduction channel
	%described in this paper, 
	If 
	the parameters are strictly ordered (see Definition \ref{defn:StrictlyOrdered}), then
	\begin{equation}
		\cfb = C = \ciid .
	\end{equation}
\end{theorem}
%

%The remainder of this section is dedicated to the proof of Theorem \ref{thm:main}.
\bl{The roadmap to the proof is as follows.
We give several lemmas prior to proving the main result, involving subsets of $\mathcal{P}$:}
\begin{itemize}
	\item \bl{Let $\mathcal{P}^*$ 
	%(see (\ref{eqn:CausalConditional})-(\ref{eqn:FeedbackCapacity})) 
	represent the set of feedback input distributions that 
	can be written
	\begin{equation}
		p_{X_0^{n-1}|Y_1^{n-1}}(x_0^{n-1}\given y_1^{n-1}) = \prod_{i=0}^{n-1} 
		p_{X_i|Y_i}(x_i \given y_i) ,
	\end{equation}
	where $y_0$ is null.
	(Note that distributions in $\mathcal{P}^*$ need not be stationary:
	$p_{X_i|Y_i}(x \given y)$ can depend on $i$.)
	Then $\mathcal{P}^* \subset \mathcal{P}$ for $n>2$. }
	\item \bl{Let 
%$\mathcal{P}^{\perp} \subset \mathcal{P}^*$ represent
%the distributions that can be written 
%%
%\begin{equation}
%	p_{X^n | Y^n}(x^n \given y^n) = \prod_{k=1}^n p_{X_i}(x_i) ,
%\end{equation}
%%
%with independent, but not necessarily identically distributed, $x_i$.
%Further, let 
$\mathcal{P}^{**}$ represent the feedback input distributions
that can be written with stationary $p_{X_i|Y_i}(x_i \given y_i)$, i.e.,
with some time-independent distribution $p_{X|Y}$ such that 
\begin{equation}
	p_{X_0^{n-1}|Y_0^n}(x_0^{n-1} \given y_1^n) = \prod_{i=0}^{n-1} p_{X|Y}(x_i \given y_i) ,
\end{equation}
where $y_0$ is null.
($\mathcal{P}^{**}$ is used in Lemma 2.)}

\end{itemize}
\bl{It should be clear from these definitions that $\mathcal{P}^{**} \subset \mathcal{P}^* \subset \mathcal{P}$. We use an existing result to show that $\cfb$ is 
satisfied by a distribution in $\mathcal{P}^*$ (Lemma 1). We then show that, {\em if we restrict ourselves to the stationary distributions $\mathcal{P}^{**}$}, then the optimal input distribution
is IID (Lemma 2). Finally, we show that the optimal input distribution is stationary, because our system satisfies certain conditions given by Chen and Berger \cite{che05} (Lemma 3). Taking these lemmas together, the capacity-achieving input distribution must be IID.
These results are laid out in the sequel.}

We begin with the following lemma,
%found in the literature, 
stating there is at least one feedback-capacity--achieving
input distribution in $\mathcal{P}^*$.
\begin{lemma}
	Taking the maximum in (\ref{eqn:FeedbackCapacity}) 
	over $\mathcal{P}^* \subset \mathcal{P}$,
	\begin{align}
		\nonumber \lefteqn{\cfb =} & \\
		\label{eqn:Lemma1}
		& \max_{p_{X_0^{n-1}|Y_1^{n-1}}(x_0^{n-1} \given y_1^{n-1}) \in \mathcal{P}^{*}} 
		\left(
			\lim_{n \rightarrow \infty} \frac{1}{n} I(X_0^{n-1} \rightarrow Y_1^n)
		\right) .
	\end{align}
\end{lemma}
\begin{proof}
	The lemma follows from \cite[Thm. 1]{yin03}.
\end{proof}
%
%Thus, there is a feedback-capacity-achieving input distribution in $\mathcal{P}^*$.

% NOTATION: Be careful about initial states
If the feedback-capacity--achieving input distribution is
in $\mathcal{P}^*$, then $Y_1^n$ is a Markov chain (the reader may check; 
see also \cite{yin03,che05}).
That is,
\begin{equation}
	p_{Y_i|Y_1^{i-1}}(y_i \given y_1^{i-1}) = p_{Y_i|Y_{i-1}}(y_i \given y_{i-1}) .
\end{equation}
Using the following shorthand notation:
\begin{eqnarray}
	%p_{U}^{(i)} & := & p_{Y_i}(U) \\
	p_{j|\B}^{(i)} & := & p_{X_i |Y_i}(c_j \given \B) \\
	p_{j|\U}^{(i)} & := & p_{X_i | Y_i}(c_j \given \U) ,
	%\bar{\alpha}^{(i)} & := & \alpha_\H (1-p_{\L|\U}^{(i)}) + \alpha_\L p_{\L|\U}^{(i)} ,
\end{eqnarray}
where the superscripts represent the time index,
the transition probability $p_{Y_i|Y_{i-1}}(y_i \given y_{i-1})$
may be represented as a matrix $\mathbf{P}_Y^{(i)}$, where
\begin{equation}
	\label{eqn:PYMatrix}
	\mathbf{P}_Y^{(i)} 
	=
	\left[
		\begin{array}{cc}
			1-  \sum_{j=1}^m \alpha_j p_{j|\U}^{(i)}
				& \sum_{j=1}^m \alpha_j p_{j|\U}^{(i)} \\
			\beta & 1-\beta
		\end{array}
	\right] 
\end{equation}
%
%with the first row and column corresponding to $U$, and the second row and column corresponding to $B$ 
(cf. (\ref{eqn:YTransitionProbabilityMatrix}), where
the input distribution is IID).  % I don't really like this because I got rid of all the other matrices.

%For now we assume that the Markov chain $Y$ is regular.
%If so,
%then the stationary distribution is given by
%%
%\begin{equation}
%	\label{eqn:PLBStationary}
%	p_{B} = \frac{\alpha_H (1-p_{L|U}) + \alpha_L p_{L|U}}
%		{\alpha_H (1-p_{L|U}) + \alpha_L p_{L|U} + \beta} .
%\end{equation}

% notation problem: we're not looking for iid here
%

%We now consider stationary distributions.

Then:
\begin{lemma}
	\label{lem:Independent}
	Suppose the parameters are strictly ordered (Definition 1). 
	Taking the maximum in (\ref{eqn:FeedbackCapacity}) 
	over $\mathcal{P}^{**} \subset \mathcal{P}^* \subset \mathcal{P}$,
		\begin{align}
			\nonumber\lefteqn{\ciid =} & \\
			\label{eqn:Lemma2}
			& \max_{p_{X_0^{n-1}|Y_1^{n-1}}(x_0^{n-1} \given y_1^{n-1}) \in \mathcal{P}^{**}} 
			\left(
				\lim_{n \rightarrow \infty} \frac{1}{n} I(X_0^{n-1} \rightarrow Y_1^n)
			\right) .
		\end{align}
\end{lemma}
%
%%%%%%%%%% FIX THE PROOF ... not sure what needs fixing, this may be an old comment
%
\begin{proof}
	We start by showing that
	$I(X_0^{i-1}; Y_i \given Y_1^{i-1})$ is independent of $p_{L|B}^{(k)}$ for all $k$.
	There is a feedback-capacity--achieving input distribution in $\mathcal{P}^*$ (from Lemma 1).
	Using this input distribution,
	\begin{eqnarray}
		\lefteqn{I(X_0^{i-1}; Y_i \given Y_1^{i-1})} & & \nonumber \\ 
		& = & H(Y_i \given Y_1^{i-1}) - H(Y_i \given Y_{i-1}, X_0^{i-1}) \\
		\label{eqn:DirectedTermSimplified}
		& = & H(Y_i \given Y_{i-1}) - H(Y_i \given Y_{i-1}, X_{i-1}) .
	\end{eqnarray}
	where (\ref{eqn:DirectedTermSimplified}) follows since (by definition)
	$Y_i$ is conditionally independent of $X_0^{i-2}$ given $(Y_{i-1},X_{i-1})$,
	and since (from the parameters being strictly ordered) 
	$Y_1^i$ is a time-homogeneous, first-order Markov chain.
	Expanding (\ref{eqn:DirectedTermSimplified}),
	\begin{eqnarray}
		\nonumber
		\lefteqn{I(X_0^{i-1}; Y_i \given Y_1^{i-1}) = } & & \\ 
		%& = & \sum_{y_i} \sum_{y_{i-1}} \sum_{x_i}
		%p_{Y_i, Y_{i-1}, X_i}(y_i, y_{i-1}, x_i)
		%\log \frac{p_{Y_i | Y_{i-1}, X_i}(y_i \given y_{i-1}, x_i)}
		%	{p_{Y_i | Y_{i-1}}(y_i \given y_{i-1})}\\
		\label{eqn:DirectedInfoSums}
		& & \sum_{y_{i-1}} p_{Y_{i-1}}(y_{i-1}) \sum_{x_{i-1}} p_{X_{i-1}|Y_{i-1}}(x_{i-1} \given y_{i-1}) \\
		\nonumber
		& & \cdot \sum_{y_i} p_{Y_i | Y_{i-1}, X_{i-1}}(y_i | y_{i-1}, x_{i-1}) \\
		\nonumber
		& & \cdot \log \frac{p_{Y_i | Y_{i-1}, X_{i-1}}(y_i | y_{i-1}, x_{i-1})}
			{p_{Y_i | Y_{i-1}}(y_i | y_{i-1})} .
	\end{eqnarray}
	From (\ref{eqn:PYMatrix}), 	
	$p_{Y_{i-1}}(y_{i-1})$ is calculated from parameters in 
	$\mathbf{P}_Y^{(i)}$ and the initial state,
	so $p_{Y_{i-1}}(y_{i-1})$ is independent of $p_{j|\B}^{(k)}$ for all $j$ and $k$.
	Further, everything under the last sum (over $y_i$) is independent of $p_{j|\B}^{(k)}$, from 
	(\ref{eqn:PYMatrix}) and the definition of $p_{Y_i | Y_{i-1}, X_i}(y_i \given y_{i-1}, x_i)$.
	There remains the term $p_{X_i|Y_{i-1}}(x_i \given y_{i-1})$, which is dependent on
	$p_{j|\B}^{(i-1)}$ when $y_{i-1} = \B$.
	However, if $y_{i-1} = \B$, then 
	\begin{eqnarray}
		\nonumber \lefteqn{\sum_{y_i} p_{Y_i | Y_{i-1}, X_{i-1}}(y_i \given \B, x_{i-1})
			} & & \\
		\nonumber
		 \lefteqn{\cdot \log \frac{p_{Y_i | Y_{i-1}, X_{i-1}}(y_i \given \B, x_{i-1})}
			{p_{Y_i | Y_{i-1}}(y_i \given \B)}} & & \\
		\label{eqn:InnerSumZero}
		& = & \sum_{y_i} p_{Y_i | Y_{i-1}}(y_i \given \B)
			\log \frac{p_{Y_i | Y_{i-1}}(y_i \given \B)}
			{p_{Y_i | Y_{i-1}}(y_i \given \B)} \\
			& = &  \sum_{y_i} p_{Y_i | Y_{i-1}}(y_i \given \B) \log 1 \\
		& = & 0 ,
	\end{eqnarray}
	where (\ref{eqn:InnerSumZero}) follows since $y_i$ is independent of $x_i$ in state $\B$.
	Thus, the entire expression is independent of $p_{L|\B}^{(k)}$ for all $k$.
	Moreover, from (\ref{eqn:DirectedInfo}), directed information
	is independent of $p_{j|\B}^{(k)}$ for all $k$.
	
	To prove (\ref{eqn:Lemma2}), distributions in $\mathcal{P}^{**}$ have
	$p_{j|\U}^{(1)} = p_{j|\U}^{(2)} = \ldots$, and 
	$p_{j|\B}^{(1)} = p_{j|\B}^{(2)} = \ldots$. 
	Since $I(X_0^{i-1}; Y_i \given Y_1^{i-1})$ is independent of $p_{j|\B}^{(k)}$ for all $k$ (by the 
	preceding argument),
	we may set $p_{j|\B}^{(k)} = p_{j|\U}^{(k)}$ for all $j$ and $k$, 
	without changing $I(X_0^{i-1}; Y_i \given Y_1^{i-1})$. Thus, inside $\mathcal{P}^{**}$,
	there exists a maximizing input 
	distribution that is independent for each channel use. By the definition
	of $\mathcal{P}^{**}$, that maximizing input distribution is IID, and there cannot exist
	an IID input distribution outside of $\mathcal{P}^{**}$. %Thus, (\ref{eqn:Lemma2}) follows.
\end{proof}
%
%Thus, if we restrict ourselves to stationary input distributions, 
%the feedback capacity is $\ciid$. 

Finally, we must show that $\cfb$ is itself achieved by
a distribution in $\mathcal{P}^{**}$.
To do so, we rely on \cite[Thm. 4]{che05}, which shows that this is the case, 
as long as several technical conditions are satisfied.
Stating the conditions and proving that they hold for this channel requires restatement of 
definitions from \cite{che05}, so we give this result in 
Appendix \ref{app:stationary} as Lemma 3. 
%(The proof of Lemma 3 is simplified by Theorem \ref{thm:maxmin}).

We can now return to the proof of Theorem \ref{thm:main}, where we relate these 
results to the Shannon capacity $C$.

\begin{proof}
	From Lemma 1, $\cfb$ is satisfied by an input distribution in $\mathcal{P}^*$.
	From Lemma 2, if we restrict ourselves to the stationary 
	input distributions $\mathcal{P}^{**}$ (where $\mathcal{P}^{**} \subset \mathcal{P}^*$), 
	then the feedback capacity is $\ciid$. From Lemma 3, the conditions of \cite[Thm. 4]{che05} are
	satisfied, which implies that there is a 
	feedback-capacity--achieving input distribution in $\mathcal{P}^{**}$.
	Therefore,
	\begin{equation}
		\label{eqn:FBequalsIID}
		\cfb = \ciid .
	\end{equation}
	For general channels,
	\begin{equation}
		\label{eqn:CapacityInequality}
		\cfb \geq C \geq \ciid ,
	\end{equation}
	because the set $\mathcal{P}$ includes the set of input distributions without feedback,
	and because the set of input distributions without feedback includes the IID input distributions.
	The theorem follows from (\ref{eqn:FBequalsIID}) 
	and (\ref{eqn:CapacityInequality}).
\end{proof}

From Theorems \ref{thm:maxmin} and \ref{thm:main}, and Equation (\ref{eqn:CIID}),
we have the following.
\begin{cor}
	\label{cor:capacity}
	Capacity $C$ of the discrete channel model given in (\ref{eqn:TimeInhomogeneousPMF})
	is given by
	\begin{equation}
			\label{eqn:CIIDLH}
			C = \max_{p_\H} \frac{\binent(p_\L \alpha_\L + p_\H \alpha_\H) 
			- p_\L \binent(\alpha_\L) - p_\H \binent(\alpha_\H)}
			{1 + (p_\L \alpha_\L + p_\H \alpha_\H)/\beta} ,
	\end{equation}
	where it is sufficient to maximize over $p_\H$, since $p_\L = 1-p_\H$.
\end{cor}

This result has an intuitively appealing form: the mutual information rate appearing in \eqref{eqn:CIID} and \eqref{eqn:CIIDLH} is the product of the binary channel MI rate with transition probabilities $\{\alpha_\L,\alpha_\H\}$, and the fraction of time the channel is in the sensitive (unbound) state. 

In Figure \ref{fig:mutualInfoPlot}, we illustrate the behaviour of the maximizing value of $p_\H$: in this
figure, the mutual information in (\ref{eqn:ClosedFormMI}) is plotted for 
$\alpha_\L = 0.1$, $\beta = 0.9$, and various values of $\alpha_\H$; in the input distribution,
all $p_j$ are equal to zero except $p_\L$ and $p_\H$. The maximum values
on each mutual information curve are illustrated.
In Figure \ref{fig:capacityPlot}, we give a contour plot of capacity for values of $\alpha_\H$ and $\alpha_\L$,
where $\beta = 0.9$.

\begin{figure}[t!]
	\begin{center}
	\includegraphics[width=3.25in]{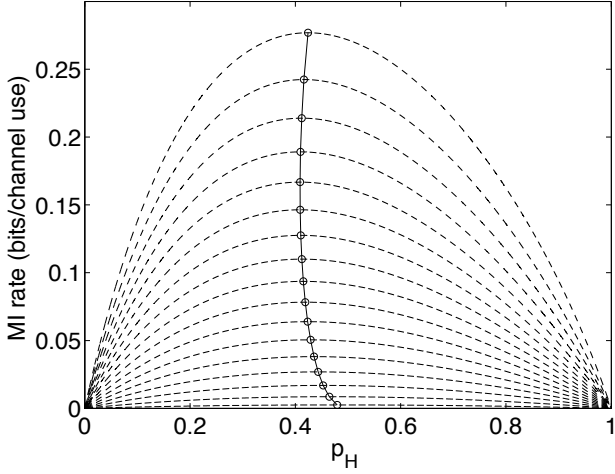}
	\end{center}
	\caption{\label{fig:mutualInfoPlot} Information maximizing values of 
	$p_\H$, with $\alpha_\L = 0.1$ and $\beta = 0.9$. Each 
	dashed curve corresponds to a particular
	value of $\alpha_\H$: from the bottom, $\alpha_\H = 0.15$; each higher 
	curve increases $\alpha_\H$ by
	0.05, up to $\alpha_\H = 0.95$ in the topmost curve.
	The maxima are circled and connected with a solid line.}
\end{figure}

\begin{figure}[t!]
	\begin{center}
	\includegraphics[width=3.25in]{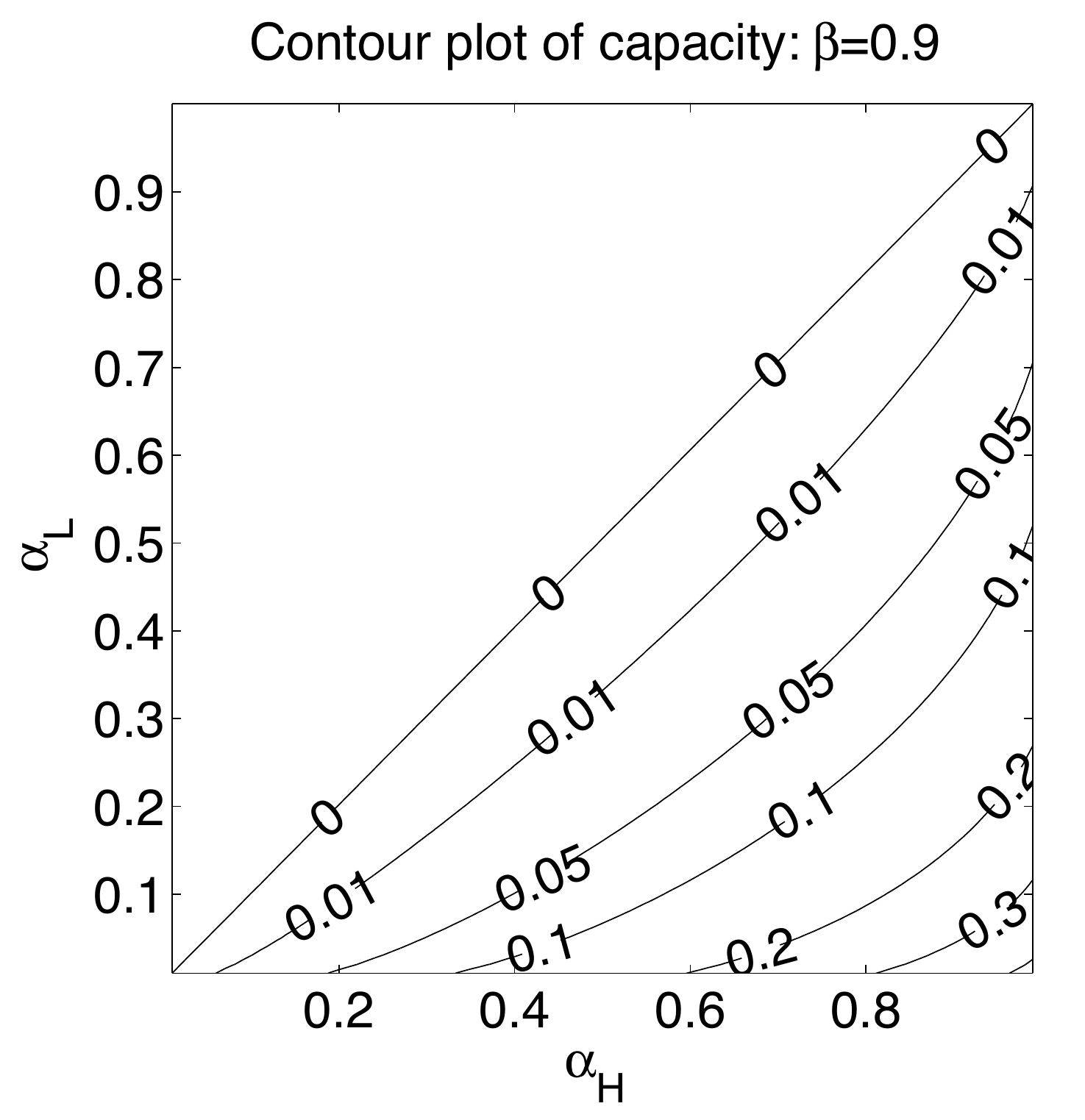}
	\end{center}
	\caption{\label{fig:capacityPlot} Contour plot of capacity with respect to $\alpha_\L$ and $\alpha_\H$,
	fixing $\beta = 0.9$. Note that $\alpha_\L > \alpha_\H$ in the upper left triangle, so
	capacity here is undefined.}
\end{figure}

\dr{With equation (\ref{eqn:CIIDLH}) we have rigorously solved the capacity for the discrete time BIND channel.  The mutual information, and hence the capacity, depend on the parameters $\alpha_\L, \alpha_\H,$ and $\beta$.  In Appendix \ref{sec:CapMaxParameters} we show that the capacity is an increasing function of $\alpha_\H$ and $\beta$, and a decreasing function of $\alpha_\L$, and that the capacity is finite for all $0\le \alpha_\L\le \alpha_\H\le 1$ and $0\le\beta\le 1$.}

%\subsection{Feedback Can Increase Capacity in a General Signal Transduction Network Channel}
%\label{sec:Counterexample}

%In this paper, we demonstrated that the cAMP\footnote{I would rather call our model the 2-state signal transduction channel model, rather than the cAMP signal transduction model.  For two reasons.  First, any biologists familiar with cAMP know that the overall signal transduction process in \textit{Dictyostelium} involves a lot more than just the single receptor molecule, and many of the molecules involved can have more than two states.  So while the simple binding/unbinding analogy is accurate for the first link in this chain, it will seem naive to people with biological background to assert that cAMP signal transduction \textit{per se} can be reduced to a two state process.  Also, cAMP is an important signaling molecule not only in the \textit{Dictyostelium} amoeba, but also in synaptic learning and memory processes, in regulation of cardiac tissue, and many other pathways.  And not all the cAMP receptors in all these pathways are necessarily two-state-only receptors.  Second, there are lots of other signaling processes that are driven by simple binding/unbinding, and we don't want to appear to be limiting the applicability of our channel model \emph{only} to channels involving cAMP.  So I think it is confusing to call our channel ``the cAMP channel".} discrete-time signal transduction channel
%has $\cfb = C$, i.e. feedback doesn't increase capacity.
To close this section,
we may wonder if it is true that $\cfb = C$
in general signal transduction models. 
The answer is no: \bl{Permuter \textit{et al.}~give an example of an $n$-state POST channel for which $\cfb > C$ \cite{AsnaniPermuterWeissman2013IEEE_ISIT,PermuterAsnaniWeissman2014IEEETransIT}.}
%\bl{Permuter's work on the POST channel \cite{permuter-arXiv} shows that there exist channels of a form similar the BIND channel, in which $\cfb > C$.
\bl{We may also ask under what conditions $\cfb = C$:
the existence of at most one sensitive transition (such as the $U \rightarrow B$ transition
in our example) is a sufficient condition for $\cfb = C$, but the necessary conditions are 
presently unknown.
}

\section{Markov inputs}
\label{sec:MarkovInputs}

Although we found in the previous section that \dr{the discrete time BIND channel has} a capacity-achieving input distribution 
that is IID,
physical concentration does not behave like an IID random variable: concentrations, either 
low $\L$ or high $\H$,
can persist for long periods of time. 
\dr{In order to increase the applicability of our analysis to biological or 
bioengineered systems, }
\bl{we can model these persistent input concentrations using a two-state $\{\L,\H\}$ Markov chain,
which we analyze in this section. Though this generalizes the IID input process to an input process with memory, we use Markov chain \dr{inputs} for simplicity, as they allow us to gain insight into the 
behaviour of the system in the presence of correlated input processes. Finite-state Markov chains may not capture the complete dynamics of the diffusion process in full generality, and we leave more general analysis to future work.}

In this section, we analyze the capacity of the discrete time channel when the
channel inputs are Markov, though we restrict ourselves to binary Markov inputs ($\L$ and $\H$) for simplicity.

\subsection{Mathematical model with Markov inputs}

Assume the sequence $X^n$  forms a Markov chain with
two parameters, $r$ (the $\L$-to-$\H$ transition probability) and 
$s$ (the $\H$-to-$\L$ transition probability),
giving a transition probability matrix of
\begin{eqnarray}
	\mathbf{P}_X =
	\left[
		\begin{array}{cc}
			1-r & r \\ s & 1-s 
		\end{array}
	\right] ,
\end{eqnarray}
with entries for $\L$ on the first row and column, and $\H$ on the second row and column.

%The channel state $Y_n$ makes transitions according to one of two matrices, depending on the value of $X_{n}$:
%\begin{eqnarray}
%\mathbf{P}_{Y|X=\L}=
%	\left[
%		\begin{array}{cc}
%			1-\alpha_\L & \alpha_\L \\ \beta & 1-\beta 
%		\end{array}
%	\right] ,
%	\hspace{1cm}
%\mathbf{P}_{Y|X=\H}=
%	\left[
%		\begin{array}{cc}
%			1-\alpha_\H & \alpha_\H \\ \beta & 1-\beta 
%		\end{array}
%	\right]. 
%\end{eqnarray}  
%
%

The joint sequence $Z^n$ forms a four-state Markov chain
with states $\{\L\U, \L\B, \H\U, \H\B\}$, with transition probability matrix
given in equation (\ref{eq:4-state-chain}).
(See Figure \ref{fig:Capacity2StateMarkov-diagram}.)
%
%\newcounter{MYtempeqncnt}
\begin{figure*}[t!]
% ensure that we have normalsize text
\normalsize
% Store the current equation number.
%\setcounter{MYtempeqncnt}{\value{equation}}
% Set the equation number to one less than the one
% desired for the first equation here.
% The value here will have to changed if equations
% are added or removed prior to the place these
% equations are referenced in the main text.
%\setcounter{equation}{84}
%
\begin{eqnarray}\label{eq:4-state-chain}
	\mathbf{P}_Z =\left[
		\begin{array}{cccc}
			(1-\alpha_\L)(1-r)& 
				\alpha_\L(1-r)& 
				(1-\alpha_\L) r&
				\alpha_\L r\\
			\beta (1-r)& 
				(1-\beta)(1-r)& 
				\beta r& 
				(1-\beta) r\\
			(1-\alpha_\H) s& 
				\alpha_\H s& 
				(1-\alpha_\H)(1-s)& 
				\alpha_\H (1-s)\\
			\beta s& 
				(1-\beta) s& 
				\beta (1-s)& 
				(1-\beta)(1-s)
			\end{array}
	\right] &
\end{eqnarray}
%
% Restore the current equation number.
%\setcounter{equation}{\value{MYtempeqncnt}}
% IEEE uses as a separator
\hrulefill
% The spacer can be tweaked to stop underfull vboxes.
\vspace*{4pt}
\end{figure*}
\begin{figure}[t!] %  figure placement: here, top, bottom, or page
   \centering
   \includegraphics[width=3.25in]{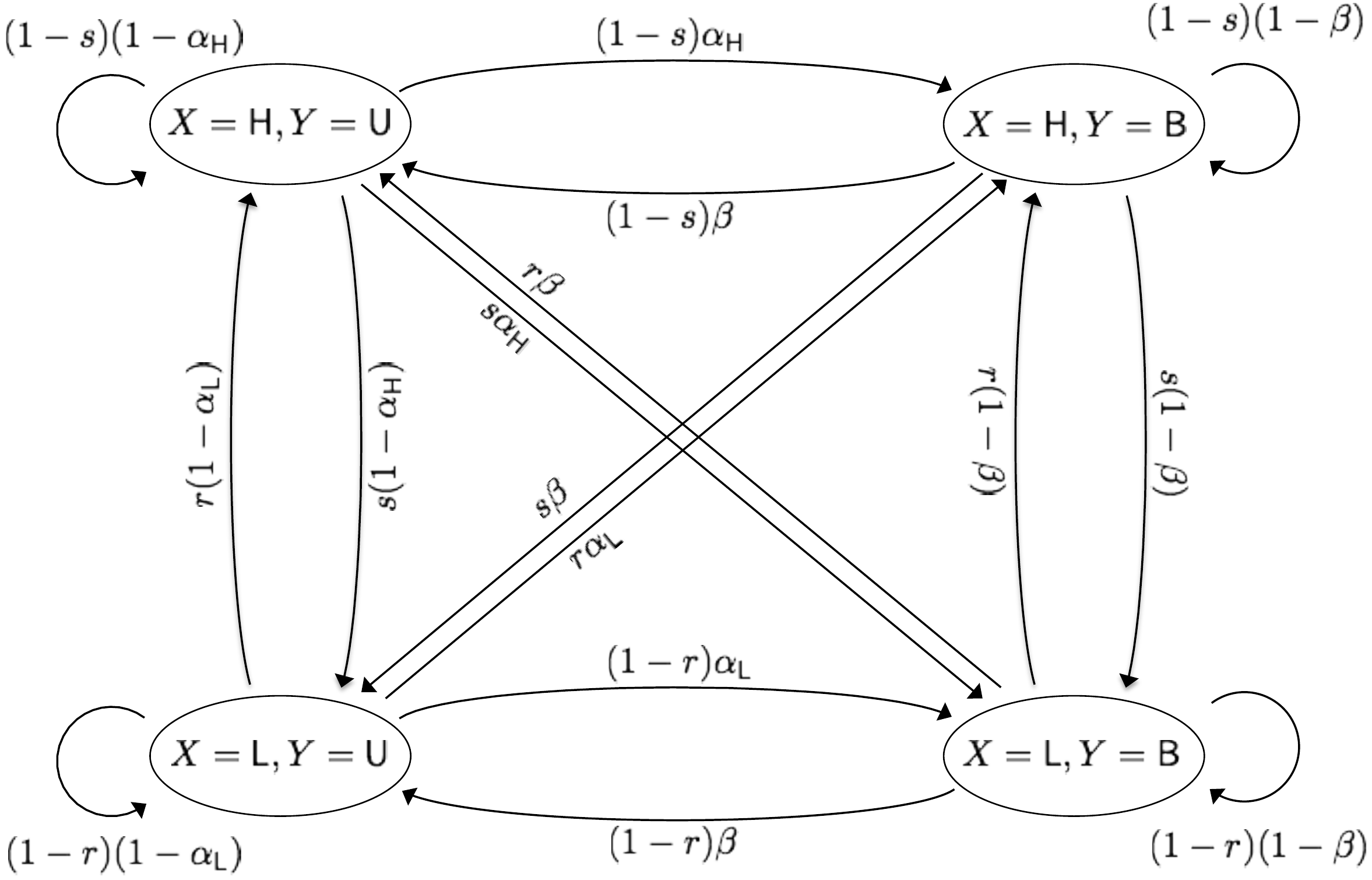} 
   \caption{Transition diagram for a 2-state channel ($Y=\U$, unbound receptor; $Y=\B$, bound receptor) driven by a 2-state input Markov process ($X=\L$, low concentration of signaling molecule; $X=\H$, high concentration of signaling molecule). Probability per time step of $X_k=\L$ to $X_{k+1}=\H$ transition is $0< r \le 1$.  Probability per time step of $X_k=\H$ to $X_{k+1}=\L$ transition is $0< s \le 1$.  Probability per time step of $Y_k=\U$ to $Y_{k+1}=\B$ transition is $\alpha_\L$ when $X_{k}=\L$, and is $\alpha_\H$ when $X_{k}=\H$; these probabilities satisfy the strictly-ordered condition (Def.~\ref{defn:StrictlyOrdered}). Probability per time step of $Y_k=\B$ to $Y_{k+1}=\U$ is $0< \beta< 1$, regardless of $X_{k}$. Compare Equation \eqref{eq:4-state-chain}. }
\label{fig:Capacity2StateMarkov-diagram}
\end{figure}
The input $X^n$ has a unique stationary distribution if $0 < r,s < 1$. The chain $Z^n$ has
a stationary distribution if $X^n$ has a stationary distribution, and the parameters
are strictly ordered. (These conditions are sufficient, but not necessary.) The
steady-state distribution on $X$ is given by
\begin{eqnarray}
p_\L=p_X(\L)=\frac{s}{r+s},\hspace{1cm}
p_\H=p_X(\H)=\frac{r}{r+s}.
\end{eqnarray}
%
%If the parameters are strictly ordered (Definition \ref{defn:StrictlyOrdered}), the joint (input, output) process, $Z_n=(X_{n},Y_n)$ is also stationary.
%The matrix $\mathbf{P}_Z$ (Eq.~\eqref{eq:4-state-chain}) defining the joint Markov process has a stationary distribution $P(X_{n},Y_n)$ given by its (normalized) eigenvector with unit eigenvalue.
%

\newpage
The stationary distribution of $Z$ is given by the (normalized) eigenvector of $\mathbf{P}_Z$ 
with unit eigenvalue.
This is given by $[p_{X,Y}(\L,\U),p_{X,Y}(\L,\B),p_{X,Y}(\H,\U),p_{X,Y}(\H,\B)]$,
where
\begin{eqnarray}\label{eq:P00}
\lefteqn{p_{X,Y}(\L,\U) =} & & \\ 
\nonumber& &\frac{1}{K}\left(\beta s (-r-s+\alpha_\H (r+s-1)+\beta (r+s-1))\right)\\
\lefteqn{p_{X,Y}(\L,\B) =} & & \\ 
\nonumber & &\frac{1}{K}\left(s (\alpha_\H (\alpha_\L (r+s-1)-r)+\alpha_\L (\beta (r+s-1)-s))\right) \\
\lefteqn{p_{X,Y}(\H,\U) =} & & \\ 
\nonumber & & \frac{1}{K}\left(\beta r (-r-s+\alpha_\L (r+s-1)+\beta (r+s-1))\right)\\
\label{eq:Z} \lefteqn{p_{X,Y}(\H,\B) =} & & \\ 
\nonumber & &   \frac{1}{K}\left(r (\alpha_\H (-r+\alpha_\L (r+s-1)+\beta (r+s-1))-\alpha_\L s)\right) ,
\end{eqnarray}
where $K$ is the normalization constant, to ensure the probabilities sum to 1.
The expressions (\ref{eq:P00}-\ref{eq:Z}) may be simplified by introducing the notation
\begin{eqnarray}
\bar{\alpha}&=&\frac{r\alpha_\H + s\alpha_\L}{r+s}\\
\lambda&=&1-r-s\\
\mu&=&\frac{\lambda}{1-\lambda}.
\end{eqnarray}
The quantity $\bar{\alpha}$ is the mean value of $\alpha$ under the equilibrium distribution for $X$ (cf. (\ref{eqn:AlphaBar}) for IID inputs); $ \lambda $ is the second eigenvalue of the matrix $\mathbf{P}_{X}$; and $0\le \mu < \infty$ is a monotonically increasing function of $\lambda$.
With this notation, the stationary distribution of the joint process satisfies
\begin{eqnarray}\label{eq:P00-subs}
p_{X,Y}(\L,\U)&=&\frac{1}{K}\beta s (1+(\alpha_\H +\beta) \mu)\\
p_{X,Y}(\L,\B) &=&\frac{1}{K} s (\bar{\alpha}+\alpha_\L(\alpha_\H + \beta)\mu) \\
p_{X,Y}(\H,\U) &=& \frac{1}{K}\beta r(1 + (\alpha_\L +\beta)\mu)\\
p_{X,Y}(\H,\B)&=&  \frac{1}{K}r(\bar{\alpha}+\alpha_\H(\alpha_\L + \beta)\mu) ,
\label{eq:P11-subs}
\end{eqnarray}
with normalization constant
\begin{equation}
K = \lambda(\alpha_\L+\beta)(\alpha_\H+\beta)+(r+s)(\bar{\alpha}+\beta)\label{eq:Z-subs}.
\end{equation}

From $p_{X,Y}(X,Y)$ one may obtain the stationary marginal distribution
$p_Y(y)$.
Define
\begin{equation}
	\Delta = \frac{s-r}{s+r}
\end{equation}
to represent the relative difference in  probabilities between the low-to-high and high-to-low transitions.
Then
\begin{eqnarray}
p_Y(\U)&=& \label{eq:P_Y(0)}
\left(1+\mu(\bar{\alpha}+\beta + \Delta (\alpha_\H-\alpha_\L) )   \right)/K'
\\
p_Y(\B)&=&  \label{eq:P_Y(1)}
(\bar{\alpha}+\mu(\bar{\alpha}\beta+\alpha_\H\alpha_\L))/K' ,
\end{eqnarray}
with normalization constant
\begin{equation}
	K'=1+\bar{\alpha}+\mu(\bar{\alpha}+\bar{\alpha}\beta+\beta+\alpha_\H\alpha_\L + \Delta(\alpha_\H-\alpha_\L)) .
\end{equation}   
%
%However, if $X^n$ is a Markov chain, $Y^n$ by itself is not generally Markov.

\subsection{Capacity estimates for Markov inputs}
\label{ssec:MarkovCapacity}

To estimate capacity for Markov inputs, we need the entropy rates for $X$, $Y$, and $Z$.
Since $X$ and $Z$ are stationary Markov processes, their entropy rates are
available in closed form.
The entropy rate of $X$ is given as a function of $r$ and $s$ by 
\begin{align}
	\mathcal{H}(X) &= \lim_{n \rightarrow \infty} \frac{1}{n} H(X^n)\\
	&= H(X_i \given X_{i-1})\\
	\label{eqn:InputEntropyRate}
	&=
	p_\L\left(r\log\frac{1}{r} + (1-r)\log\frac{1}{1-r}\right)\\ 
	& \:\: \nonumber+
	p_\H\left(s\log\frac{1}{s} + (1-s)\log\frac{1}{1-s}\right).
\end{align}
The joint entropy rate $\mathcal{H}(Z) = \mathcal{H}(X,Y)$ can be calculated directly from (\ref{eq:P00-subs})-(\ref{eq:Z-subs}).
Let
\begin{equation}
	\label{eqn:pi}
	\pi_{xy}=p_{X,Y}(x,y)
\end{equation}
denote the stationary density of the Markov process, i.e.~the four terms
given in \eqref{eq:P00-subs}-\eqref{eq:P11-subs}. 
Denote the joint transition probabilities from the matrix $\mathbf{P}_Z$ as
\begin{align}
	T_{xy \rightarrow x'y'} &= p_{X_i,Y_i|X_{i-1},Y_{i-1}}(x^\prime,y^\prime \given x,y) \\
	\label{eqn:M}
	&= p_{Y_{i}|X_{i-1},Y_{i-1}}(y^\prime \given x, y) 
	p_{X_{i}|X_{i-1}}(x^\prime \given x) .
\end{align}
Further let 
\begin{equation}
	\label{eqn:PhiEquation}
	\phi(p) = - p \log p .
\end{equation}
Then the joint entropy rate is
\begin{equation}
\label{eqn:JointEntropyRate}
\mathcal{H}(X,Y)=\sum_{x=\L}^{\H}\sum_{y=\U}^{\B}\pi_{xy}\left(\sum_{x'=\L}^{\H}\sum_{y'=\U}^{\B}\phi\left(  T_{xy\rightarrow x'y'} \right)\right).
\end{equation}

However, the output process $Y$ is not a Markov process in general, and
its entropy rate is \dr{not available in closed form}.
To bound the entropy rate of  $Y$, we use the fact that this rate $\mathcal{H}(Y)$ is bounded above and below by entropies conditioned on a finite number of previous channel states.  From standard inequalities (\cite{CoverThomas1990}, Theorem 4.4.1) we have, for each $n$,
\begin{align}
\nonumber\lefteqn{H(Y_n|X_0,Y_0,\cdots,Y_{n-1})} & \\
\label{eq:CoverThomasInequalities}
&\le \mathcal{H}(Y)\le H(Y_n|Y_0,\cdots,Y_{n-1})
\end{align}
and
\begin{align}
\nonumber\lefteqn{\lim_{n\to\infty}H(Y_n|X_0,Y_0,\cdots,Y_{n-1})} & \\
\label{eqn:LimitInequalities}
&=\mathcal{H}(Y)=\lim_{n\to\infty}H(Y_n|Y_0,\cdots,Y_{n-1}).
\end{align}
%
%To avoid ambiguity, in this section $\mathcal{X}$ and $\mathcal{Y}$ will refer to the input and output sequences, and $H[\mathcal{X}], H[\mathcal{Y}]$ to the entropy rates of these sequences, respectively.

Using the inequalities (\ref{eq:CoverThomasInequalities}), we have
\begin{align}
	\mathcal{I}(X;Y) &= \mathcal{H}(X) + \mathcal{H}(Y) - \mathcal{H}(X,Y)\\
	\label{eqn:UpperInfoRate}
	&\leq \mathcal{H}(X) - \mathcal{H}(X,Y) +  H(Y_n|Y_0,\cdots,Y_{n-1}) \\
	&=: \mathcal{I}_n^+(X;Y)
\end{align}
and
\begin{align}
	\mathcal{I}(X;Y) &\geq 
		\mathcal{H}(X) - \mathcal{H}(X,Y) +  H(Y_n|X_0,Y_0,\cdots,Y_{n-1})\\
\label{eqn:LowerInfoRate}
		&=: \mathcal{I}_n^-(X;Y) .
\end{align}
The required bounds on $\mathcal{H}(Y)$ are derived below.

\subsubsection{Upper Bounds on $\mathcal{H}(Y)$ and $\mathcal{I}(X;Y)$}

First consider the one-step conditional entropy of the $Y$ sequence,
\begin{eqnarray}
\nonumber\lefteqn{H(Y_1|Y_0)} & & \\
\label{eq:hy2y1-a}
&=&\sum_{y_0}p_Y(y_0)
\binent\left(p_{Y_1|Y_0}(\B \given y_0) \right) \\
\label{eq:hy2y1-b}
&=&\sum_{y_0} p_Y(y_0) 
\binent\left( \sum_{x_0,x_1}  \frac{\pi_{x_0y_0} T_{x_0 y_0 \rightarrow x_1 \B} }{\pi_{\L y_0}+\pi_{\H y_0}}\right) ,
\end{eqnarray}
where $\binent$ is the binary entropy function, and $\pi_{xy}$ and $T_{x_0 y_0 \rightarrow x_1 \B}$ are the steady-state probability (resp. transition probability) of the 
$(X,Y)$ Markov chain, defined in (\ref{eqn:pi}) (resp. (\ref{eqn:M})).
At the same time, the mutual information rate is bounded above by the entropy rate of the input \eqref{eqn:InputEntropyRate}.  Thus, from (\ref{eqn:InputEntropyRate}), (\ref{eqn:JointEntropyRate}), and (\ref{eq:hy2y1-b}), the first upper bound $\mathcal{I}_1^+(X;Y)$ is given by
\begin{align}
	\nonumber
	\lefteqn{\mathcal{I}_1^+(X;Y) =} & \\
	\nonumber
	&p_\L\left(r\log\frac{1}{r} + (1-r)\log\frac{1}{1-r}\right)\\
	\nonumber 
	& + p_\H\left(s\log\frac{1}{s} + (1-s)\log\frac{1}{1-s}\right) \\
	\nonumber
	& +\left\{\left[ \:-\: \sum_{x=\L}^{\H}\sum_{y=\U}^{\B}\pi_{xy}
	\left(\sum_{x'=\L}^{\H}\sum_{y'=\U}^{\B}\phi\left(  T_{xy\rightarrow x'y'} \right)\right)\right.\right.\\
	\label{eqn:InfoRateUpper1}
	&  \left.\left.\:+\: \sum_{y_0} p_Y(y_0) 
\binent\left( \sum_{x_0,x_1} \frac{T_{x_0 y_0 \rightarrow x_1 \B} \pi_{x_0y_0}}{\pi_{\L y_0}+\pi_{\H y_0}}\right) \right] \wedge\: 0 \right\},
\end{align}
where $\{A\wedge B\}$ represents the lesser of $A$ and $B$.
The bound $\mathcal{I}_1^+(X;Y)$ is illustrated in Figure \ref{fig:UB1}.

\begin{figure}[t!] %  figure placement: here, top, bottom, or page
   \centering
   \includegraphics[width=3in]{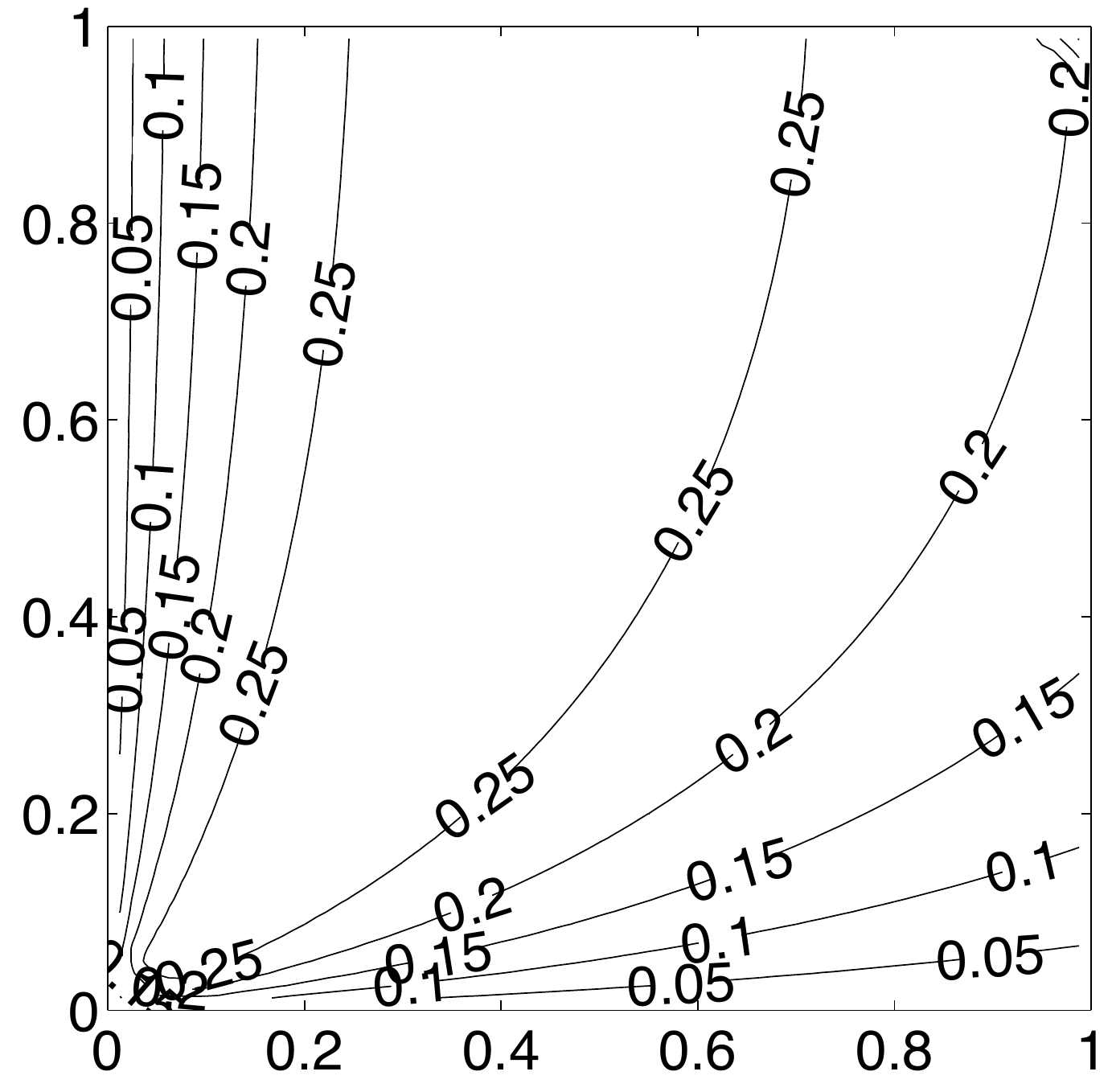} 
   \caption{Mutual information upper bound given by 
   %$\text{min}(H(\mathcal{X}),H(\mathcal{X})-H(\mathcal{X},\mathcal{Y})+H(Y_1|Y_0))$
   (\ref{eqn:InfoRateUpper1}), for $\alpha_\L=0.1, \alpha_\H=0.9,\beta=0.5$. Horizontal axis: $r$; vertical axis: $s$. }
   \label{fig:UB1}
\end{figure}

Next consider the two-step entropy, $H(Y_2|Y_0,Y_1)$. We calculate this entropy explicitly as follows:
\begin{equation}
	H(Y_2|Y_0,Y_1) = \sum_{y_0,y_1} p_{Y_0,Y_1}(y_0,y_1)
	\binent\left(p_{Y_2 | Y_0,Y_1}(\B \given y_0,y_1) \right) ,
\end{equation}
where
\begin{align}
	p_{Y_0,Y_1}(y_0,y_1) 
		&= \sum_{x_0,x_1} \pi_{x_0y_0}T_{x_0 y_0\to x_1 y_1}\\
	p_{Y_2 | Y_0,Y_1}(\B \given y_0,y_1) 
		&= \sum_{x_2} p_{X_2,Y_2|Y_0,Y_1}(x_2,\B \given y_0,y_1)
\end{align}
and, writing $X_0^n$ for $(X_0,\cdots,X_n)$,
\begin{align}
	\nonumber\lefteqn{p_{X_2,Y_2|Y_0,Y_1}(x_2,\B \given y_0,y_1)} & \\
		&= \sum_{x_0,x_1} p_{X_0^2,Y_2|Y_0,Y_1}(x_0^2,\B \given y_0,y_1) \\
		&= \sum_{x_0,x_1} 
		\left( \frac{T_{x_1 y_1 \to x_2 \B} p_{X_0,Y_0,X_1,Y_1}(x_0,y_0,x_1,y_1)}
		{\sum_{x_0,x_1} p_{X_0,Y_0,X_1,Y_1}(x_0,y_0,x_1,y_1)} \right)\\
		&= \sum_{x_0,x_1} 
		\left( \frac{T_{x_1 y_1 \to x_2 \B} T_{x_0 y_0 \rightarrow x_1 y_1} \pi_{x_0 y_0}}
		{\sum_{x_0,x_1} T_{x_0 y_0 \rightarrow x_1 y_1} \pi_{x_0 y_0} } \right) .
\end{align}
%\begin{align}
%\Pr[Y_3=B|Y_1=y_1,Y_2=y_2]&=\sum_{x_3}\Pr[Y_3=B,X_3=x_3|Y_1=y_1,Y_2=y_2]\\
%\Pr[Y_3=B,X_3=x_3|Y_1=y_1,Y_2=y_2]&= \nonumber
%\sum_{x_1,x_2}\Pr[Y_3=B,X_3=x_3|Y_1=y_1,X_1=x_1,Y_2=y_2,X_2=x_2]\cdot\\ 
%& \cdot \Pr[X_1=x_1,X_2=x_2|Y_1=y_1,Y_2=y_2]\\
%&= \sum_{x_1,x_2}\Pr[Y_3=B,X_3=x_3|Y_2=y_2,X_2=x_2]\cdot\\ 
%& \cdot\left( \frac{\Pr[Y_1=y_1,X_1=x_1,Y_2=y_2,X_2=x_2]}{\sum_{x_1,x_2}\Pr[Y_1=y_1,X_1=x_1,Y_2=y_2,X_2=x_2]}  \right)\\
%&= \sum_{x_1,x_2}\Pr[Y_3=B,X_3=x_3|Y_2=y_2,X_2=x_2]\cdot\\ 
%& \cdot\left( \frac{\pi_{x_1y_1}T(x_1,y_1\to x_2,y_2)}{\sum_{x_1,x_2}\pi_{x_1y_1}T(x_1,y_1\to x_2,y_2)}  \right)
%\end{align}
%
In general, the $n^{\text{th}}$ upper bound of this form is obtained from the $n$-step upper bound of the entropy rate of the channel state.  Writing $Y_0^{n-1}$ for $(Y_0,\cdots,Y_{n-1})\in\{\U,\B\}^n$, 
%and  $X_0^{n-1}$ for $(X_0,\cdots,X_{n-1})\in\{\L,\H\}^n$, 
the $n$-step upper bound is given by a sum involving $2^n$ terms
\begin{align}
\mathcal{H}^+_n &:= H(Y_n|Y_0^{n-1})\\
\label{eq:H-upper-bound-general}
&=\sum_{y_0^{n-1}\in\{\U,\B\}^n}p_{Y_0^{n-1}}(y_0^{n-1})\binent( p_{Y_n | Y_0^{n-1}}(\B \given y_0^{n-1}) ) .
\end{align}
In appendix \ref{sec:journalpaper-appendix-sum-product} we 
briefly show how to use the sum-product algorithm to calculate the general $n$-step bound.  
Figure \ref{fig:Capacity2StateMarkovNstep_fig1} illustrates the convergence of the sequence of upper bounds $\mathcal{H}_n^+$ with a similar sequence of lower bounds (next section) for $n=2,3,4,5$.

\subsubsection{Lower Bounds on $\mathcal{H}(Y)$ and $\mathcal{I}(X;Y)$}

In a similar fashion,  we can formulate a lower bound on $\mathcal{H}(Y)$ involving $n$ prior states of $Y$ and the initial state of $X$, namely
\begin{align}
\mathcal{H}^-_n &:= H(Y_n|X_0,Y_0^{n-1})\\
\nonumber
&=\sum_{x_0\in\{\L,\H\}}\sum_{y_0^{n-1}\in\{\U,\B\}^n}
p_{X_0,Y_0^{n-1}}(x_0,y_0^{n-1})\\
\label{eq:H-lower-bound-general}& \:\:\cdot
\binent(p_{Y_n | X_0,Y_0^{n-1}}(\B \given x_0,y_0^{n-1})) .
\end{align}
Moreover, we also have the trivial lower bound on the mutual information rate $\mathcal{I}(X;Y)\ge 0$. 

Again, appendix \ref{sec:journalpaper-appendix-sum-product} briefly shows how to perform this calculation using the sum-product algorithm. 
%
% This doesn't seem right!
%
%The lower bound given by one-step conditioning, i.e.~$\mathcal{H}_1^-=H[Y_2|X_1,Y_1],$ is identically zero. Nontrivial  bounds on the mutual information are obtained beginning with two-step conditioning, namely
%$\mathcal{H}_2^-=H[Y_3|X_1,Y_1,Y_2]$.
Figure \ref{fig:Capacity2StateMarkovNstep_fig1} illustrates the convergence of $\mathcal{H}^\pm_n$ in the interior of the region $0<r,s<1$, for  $n=2,3,4,5$.
\begin{figure}[htbp] %  figure placement: here, top, bottom, or page
   \centering
   \includegraphics[width=3.25in]{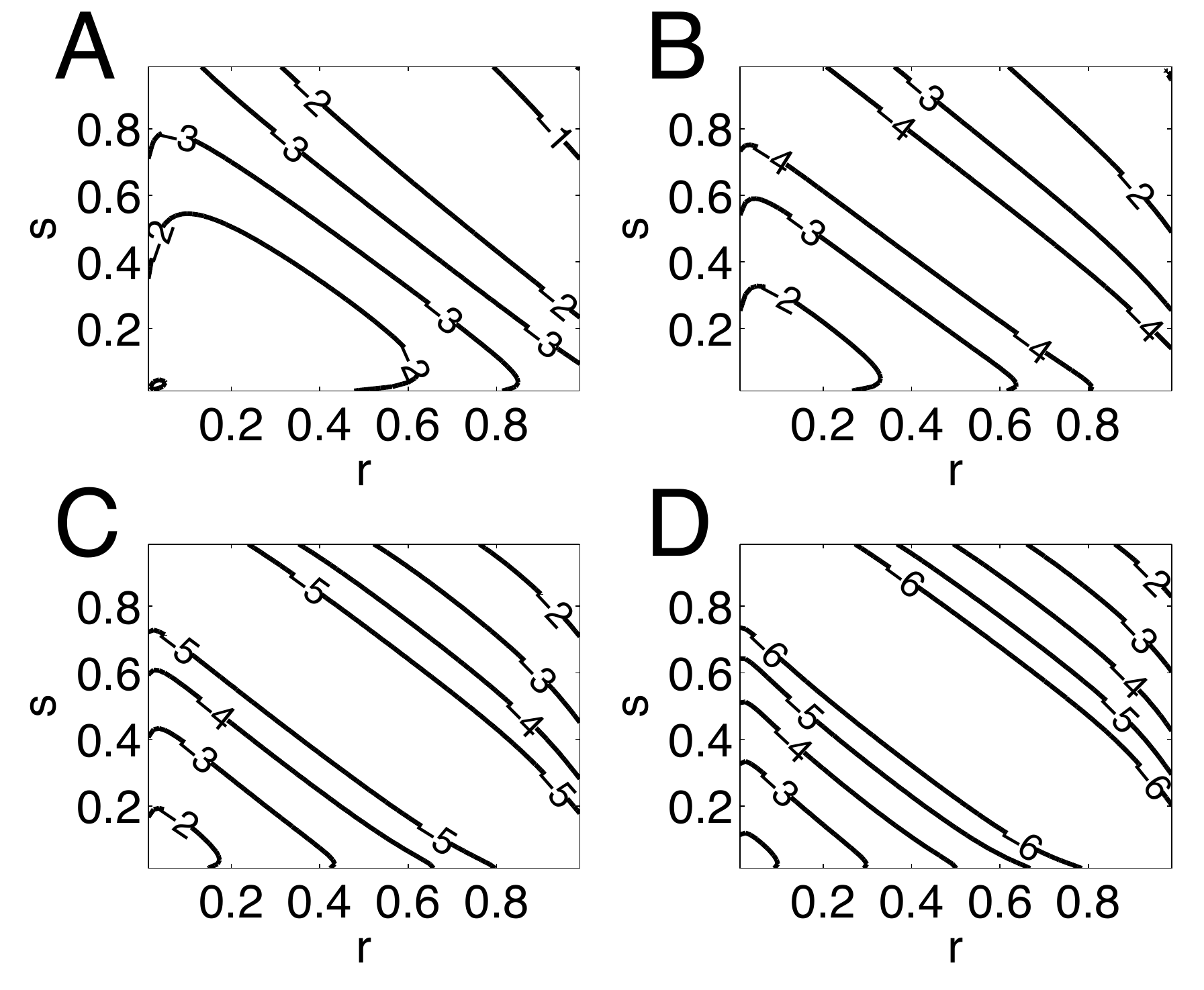} 
   \caption{Difference between upper and lower bounds for the mutual information rate $\mathcal{I}(X;Y)$ of the 2-state Markov channel, as a function of switching parameters $0<r<1$ and $0<s<1$. The upper and lower bounds are given by (\ref{eqn:UpperInfoRate}) and (\ref{eqn:LowerInfoRate}), respectively. Each panel shows $\log_{10}(\mathcal{I}^+_n-\mathcal{I}^-_n)$ for $n=2$ (panel \textbf{A}), $n=3$ (panel \textbf{B}), $n=4$ (panel \textbf{C}), $n=5$ (panel \textbf{D}). Parameter values are $\alpha_\L=0.1, \beta=0.5, \alpha_\H=0.9$.  Each increase in the depth of conditioning decreases the gap between the upper and lower bounds by roughly an order of magnitude.}
   \label{fig:Capacity2StateMarkovNstep_fig1}
   % figure source: projects/bioinfo/code/pjthomas/matlab/capacity_binding/Capacity2StateMarkovNstep_fig.m
\end{figure}
The upper and lower bounds obtained by conditioning $Y$ to a depth of five steps constrains the mutual information to within less than 1\% for input switching rates satisfying $|r+s-1|\lesssim 0.9$, or roughly all but 1\% of the  $(r,s)$ plane, for the parameters  ($\alpha_\L=0.1, \beta=0.5, \alpha_\H=0.9$) illustrated in Fig.~\ref{fig:Capacity2StateMarkovNstep_fig1}. \bl{(Elsewhere, the bounds can be obtained to greater depth using the same procedure.)}  We confirmed this result using Monte Carlo sampling to obtain empirical mutual information rates.  \begin{figure}[htbp] %  figure placement: here, top, bottom, or page
   \centering
   \includegraphics[width=3in]{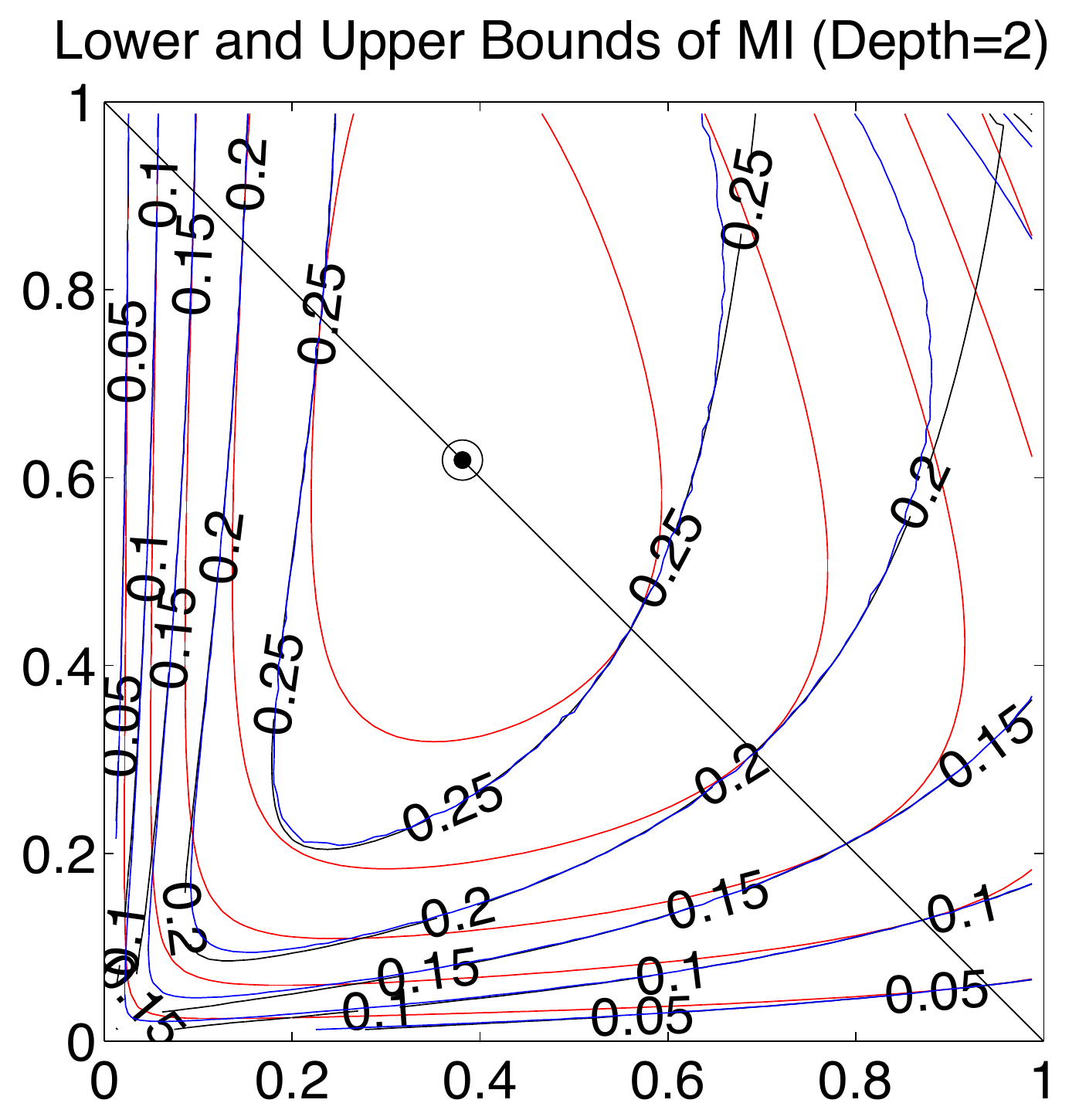}
   \includegraphics[width=3in]{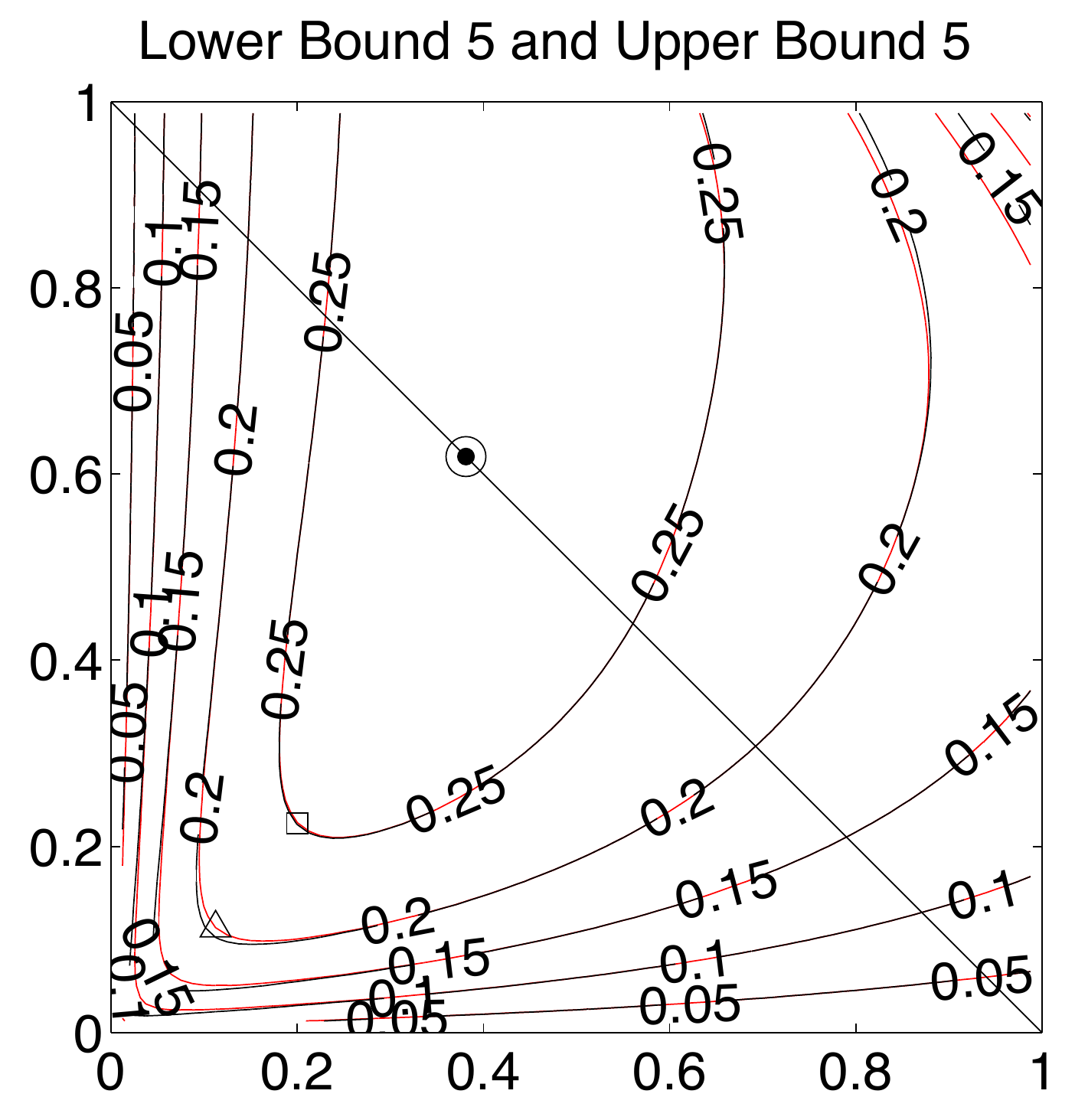} 
   \caption{Mutual information of the two-state channel driven by a two-state Markov input process. The channel parameters are  $\alpha_\L=0.1, \beta=0.5, \alpha_\H=0.9$.  The input switching rates are $r$ (low-to-high) and $s$ (high-to-low). \textbf{Top:} Three sets of curves show the upper bound $\mathcal{I}_2^+(X;Y)$ (black), lower bound $\mathcal{I}_2^-(X;Y)$ (red), and Monte Carlo estimate (blue). \bl{The diagonal line represents $r+s=1$, along which the Markov input process reduces to IID; the capacity-achieving input distribution occurs along this line.} \dr{Contour labels refer to the upper bound (black curve).   The UB, LB and MC curves coincide when $r+s=1$.}
    \textbf{Bottom:} Two sets of curves show the \dr{upper bound (black curves) and lower bound (red curves)} $\mathcal{I}_5^\pm(X;Y)$, which are indistinguishable over most of the $(r,s)$ plane. Horizontal axis: $r$; vertical axis: $s$. \dr{Circled dot indicates capacity-achieving IID input switching rates $r_\text{opt}\approx0.381, s_\text{opt}\approx 0.619$; at this point $\mathcal{I}_5^\pm(X;Y)\approx 0.279$. Square: $(r,s)\approx(0.200,0.225)$.   Triangle: $(r,s)\approx(0.112,0.112)$.  See text for details.}}
   \label{fig:MI_upper_lower_5deep}
\end{figure}

\dr{Figure \ref{fig:MI_upper_lower_5deep}, which shows the mutual information surface for parameters $\alpha_\L=0.1,\beta=0,5,\alpha_\H=0.9$, as a function of low-to-high switching rate $r$, and high-to-low switching rate $s$,  provides several insights.  First, the upper panel shows reasonable agreement between the upper bound, the lower bound, and direct Monte Carlo sampling, even when the UB and LB are only calculated to a depth of two levels of conditioning (the top panel plots $\mathcal{I}_2^\pm(X;Y)$ together with the Monte Carlo estimate).
Second, consistent with Figure \ref{fig:Capacity2StateMarkovNstep_fig1}, conditioning  five steps deep gives indistinguishable upper and lower bounds for all but a small portion of the $(r,s)$ plane.  Third, the channel can endure a significant departure from the idealized IID input case with only a modest loss of efficiency.  In the lower panel the $\text{MI}=0.25$ contour represents a roughly 10\% decrement in the mutual information rate relative to the capacity ($\approx 0.279$ bits per time step, for these parameters).  This contour extends to $r+s$ values as low as $r+s\approx 0.425$.  Finally, by introducing memory into the channel (deviating from the line $r+s=1$) we gradually change the optimal strategy for deploying high- \textit{versus} low-concentration input signals. To see this, note that the closest point to the origin along the $\text{MI}=0.25$ contour is $(r_\square,s_\square)\approx(0.200,0.225)$, marked $\square$ on the bottom panel of Figure \ref{fig:MI_upper_lower_5deep}.  When the input is IID, the optimal low- \textit{versus} high-input frequency is biased towards low-concentration inputs ($s_\text{opt}/r_\text{opt}\approx 0.619/0.381\approx1.62$).
As the sum of the switching rates decreases, this bias is gradually reduced;
for low switching rates (close to the origin of the $(r,s)$ plane) the optimal ratio approaches unity.  For example, at the point $\square$ the low-input--frequency to high-input--frequency ratio is $s_\square/r_\square\approx .225/.200\approx 1.12$.  At the next contour ($\mathcal{I}_5^\pm(X;Y)\approx 0.20$) the closest point to the origin, marked $\triangle$, is $(r_\triangle,s_\triangle)\approx(0.112,0.112)$.  At this point the low-input--frequency to high-input--frequency ratio is approximately unity.  Our analysis of the two-state discrete time BIND channel, with input constrained to a two-state Markov process, suggests that we could expect to see different signaling strategies employed in specific biological channels, depending on the persistence times of diffusion-mediated signals in those channels.}

\section{Continuous-time \dr{limits of the discrete time} channel}
\label{sec:continuoustime}

\dr{
The BIND channel arises from an underlying physical system -- ligand molecules binding to a receptor protein -- that operates in continuous rather than discrete time.  The \emph{per timestep transition probabilities} $\alpha_{\L/\H}$ and $\beta$ derive from \emph{continuous time transition rates} $k_+$ and $k_-$ in the sense that $\alpha_{\L/\H}=k_+c_{\L/\H}\,\Delta t+o(\Delta t)$ and $\beta=k_-\,\Delta t+o(\Delta t)$ (cf.~\S I-B).   Rigorous analysis of point process channels in continuous time requires additional probability theoretic techniques beyond the scope of the present paper 
(for results in this direction, see\cite{Kabanov1978,Bremaud1981,Verdu1999PoissonTalk,SundaresanVerdu2006ieeeIT}).
%(Kabanov 1978; Bremaud 1981; Verdu 1999; Sundaresan and Verdu 2006).
Nevertheless it is of interest to study how the mutual information and  capacity of the discrete time BIND channel behave in the limit of small time steps.}

\dr{In this section we \dr{therefore} consider the capacity \dr{of the discrete time BIND channel} in two limiting cases.  
In \S IV-A we \dr{evaluate} the limiting behavior of the discrete time mutual information rate in the limit of short time steps, and \dr{its supremum} with respect to parameters.  While this approach does not provide a rigorous proof of \dr{a continuous-time} capacity formula, \dr{the limiting form of the mutual information per time step takes} an intuitively appealing form, namely} the product of the mutual information rate of a counting process when the channel is in the receptive or unbound state, multiplied by the fraction of time it is in that state under stationary input conditions.  In \S IV-B we again consider the short time-step limit of the mutual information, but do so while fixing the per time-step release  probability to be unity.  \bl{Although again not a rigorous proof of capacity, in this case it is interesting to note that} the continuous time channel without \bl{an insensitive or bound state} gives the same \dr{limiting} mutual  information rate and capacity expression as Kabanov's Poisson channel (Wyner 1988a, Wyner 1988b).  
\dr{This limit provides an important consistency check on the discrete time BIND model, and indicates its connection to existing point process models.}

%\input{journalpaper-continuous-old-section-A.tex}
% This section was removed because we couldn't quickly find a way to show that H(X|Y) goes to zero faster than H(X) as (r,s)*epsilon goes to zero.

\subsection{Derivation of a capacity expression for the 2-state signal transduction channel}
\label{ssec:short-time-analysis} % based on short-time-analysis2.tex

We start with the expression for the discrete time mutual information rate \eqref{eqn:CIIDLH}.
Assuming the input distribution is IID, the
mutual information per discrete time step is given by
\begin{equation}
	\label{eqn:MIperdiscretestep}
	I(X;Y) = % note we are not dividing by time step yet
	\frac{\binent(\alpha_\H p_\H + \alpha_\L p_\L) 
		- p_\H \binent(\alpha_\H) - p_\L \binent(\alpha_\L)}
	{1 + (p_\H \alpha_\H + p_\L \alpha_\L)/\beta} ,
\end{equation}
where
\begin{equation}
	\binent(p) = - p \log p - (1-p) \log (1-p) .
\end{equation}
The IID capacity is obtained by maximizing over the set $p_\H\in[0,1]$ with $p_\L=1-p_\H$.  For convenience, we use $x$ to represent  $p_\H$ in the rest of this section.

The \dr{discrete time channel} model  assumes that the probability of transition \emph{per time step} is $\alpha_\H,\alpha_\L, $ or $\beta$, depending on the state of the input and the state of the channel.  The IID input approximation assumes the input can flicker back and forth arbitrarily fast, so that successive time steps are uncorrelated.  For the following calculation, we will assume that the input can remain IID even \dr{in the limit of vanishing time step}.  To represent discretization \dr{with an arbitrary time step,} we set
\begin{eqnarray*}
\alpha^*_\H&=&\epsilon\alpha_\H\\
\alpha^*_\L&=&\epsilon\alpha_\L\\
\beta^*&=&\epsilon\beta
\end{eqnarray*}
where $\epsilon > 0$ is the size of the time step.  The case $\epsilon=1$ corresponds to the discrete time model considered \dr{in \S\ref{sec:Capacity}}.  The fixed constants $\alpha_{\H/\L}$ and $\beta$ now represent transition rates per unit time, rather than probabilities per time step.

Figure \ref{fig:MI-rate-asymptotic} shows the per-time-step mutual information, as a function of $0\le x \le 1$, for $\alpha^*_\L=0.1\epsilon$, $\alpha^*_\H=0.9\epsilon$, and $\beta^*=0.5\epsilon$, for $\epsilon$ ranging from 1 to $10^{-4}$.  
\begin{figure}[htbp] %  figure placement: here, top, bottom, or page
   \centering
   \includegraphics[width=2.75in]{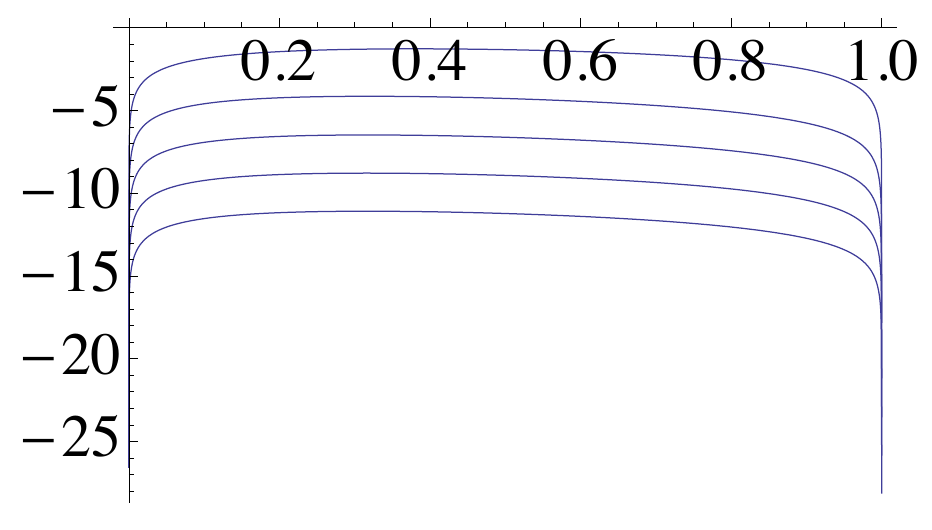} 
   \includegraphics[width=2.75in]{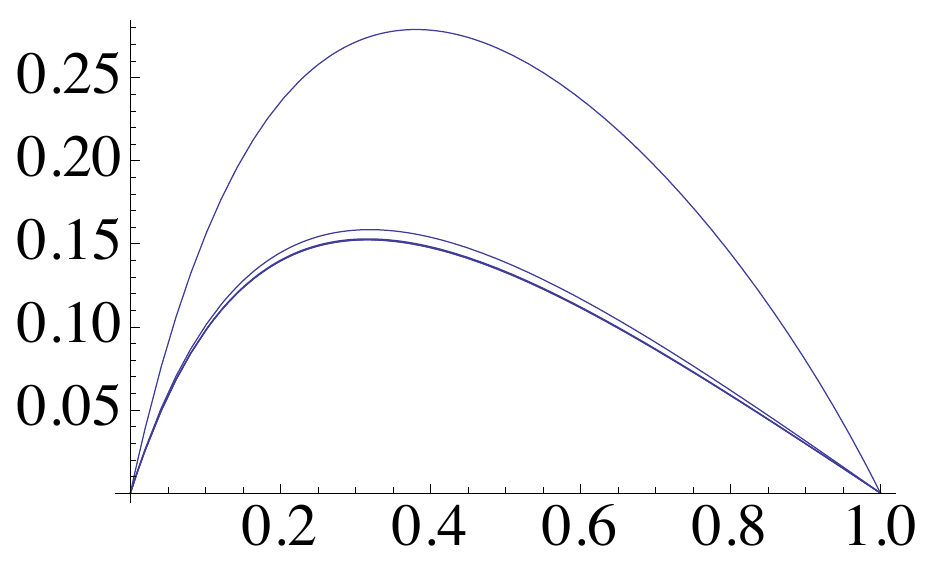}
   \caption{Convergence of mutual information rate as $\epsilon\to 0$.  \textbf{Left:} Log of the mutual information per time step, $\log(I(X;Y))$, plotted as a function of $x\in[0,1]$, for $\alpha_\L^*=0.1\epsilon$, $\alpha_\H^*=0.9\epsilon$, and $\beta^*=0.5\epsilon$.  The curves stacked from top to bottom have $\epsilon=1, 0.1, 0.01, 0.001$ and $0.0001$, respectively.
   \textbf{Right:} Mutual information \emph{rate}, $I/\epsilon$, for the same parameters.  The top curve has $\epsilon=1$; the next has $\epsilon=0.1$; the rest are indistinguishable.  
The curves suggest that as $\epsilon\to 0$, the mutual information rate $I/\epsilon$ converges to a finite quantity, and that the location of the maximum converges to a given value $x_\text{opt}$ as well.  
% code in projects/bioinfo/papers/pjthomas/new-capacity-binding/Capacity2StateMarkov_scaling_time.m
}
   \label{fig:MI-rate-asymptotic}
\end{figure}
The curves suggest that, as expected, the optimal value of $x$ lies in the interior of the interval $(0,1)$.  Moreover, the optimal value appears to converge to a given value $x_\text{opt}$ as $\epsilon\to 0$, \dr{distinct from the optimal value when $\epsilon=1$.}

We define the mutual information rate for a given $\epsilon>0$, $\mathcal{I}_\epsilon$, to be $I/\epsilon$, and we study how this quantity scales as $\epsilon\to 0^+$. 
%We will use the following notation and the asymptotic expansions:
%\begin{eqnarray}
%\bar{\alpha}(x)&\defn&x \alpha_\H + (1-x)\alpha_\L\\
%\log(1+u)&=&u-\frac{u^2}{2}+O(u^3),\text{ as }u\to 0
%\end{eqnarray}
%
In the remainder of this section we do the following: 
\begin{enumerate}
\item We study the information rate $\mathcal{I}_\epsilon(x)$ for small $\epsilon$ and show that it has a unique maximum in the interior of the interval $0\le x \le 1$, where $x=p_\H$ is the probability that the input signal is in the ``high'' concentration state.
\item By taking the limit as $\epsilon\to 0^+$, and optimizing over $x$,  we obtain an expression \dr{for the capacity of the discrete time channel in the continuous time limit}, as a function of the binding and unbinding rates $\alpha_\H$, $\alpha_\L$, and $\beta$.  
\end{enumerate}

In the following section \S \ref{ssec:ReductionToKabanov}, we further show that by taking the limit of the capacity for the continuous time channel, as the unbinding rate $\beta\to \infty$, we recover Kabanov's expression for the capacity for the Poisson channel.

\subsubsection{Critical point of the information rate $\mathcal{I}_\epsilon$ for small $\epsilon >0$}
\label{sssec:criticalpoint}
First we study the behavior of the optimal value of $x$ in the limit of small $\epsilon$.  Assuming an interior maximum for $I$, we set the derivative of the RHS of Equation (\ref{eqn:MutualInfo}) equal to zero to obtain the necessary and sufficient condition in equation (\ref{eq:interiormax}).
\begin{figure*}[!t]
% ensure that we have normalsize text
\normalsize
% Store the current equation number.
%\setcounter{MYtempeqncnt}{\value{equation}}
% Set the equation number to one less than the one
% desired for the first equation here.
% The value here will have to changed if equations
% are added or removed prior to the place these
% equations are referenced in the main text.
\begin{eqnarray}
\nonumber 0&=&
-\alpha_\H \beta \epsilon  \log \left(\frac{1}{-\alpha_\H x \epsilon +\alpha_\L
   (x-1) \epsilon +1}\right)
   +\alpha_\L \beta \epsilon  \log \left(\frac{1}{-\alpha_\H x
   \epsilon +\alpha_\L (x-1) \epsilon +1}\right)\\
   \nonumber &&+\beta \epsilon  (\alpha_\H-\alpha_\L) \log
   \left(\frac{1}{\alpha_\H x \epsilon -\alpha_\L x \epsilon +\alpha_\L \epsilon
   }\right)
   -\alpha_\H \epsilon  (\alpha_\L+\beta) \log \left(\frac{1}{\alpha_\H \epsilon
   }\right)\\
   \nonumber &&+\alpha_\L \epsilon  (\alpha_\H+\beta) \log \left(\frac{1}{\alpha_\L \epsilon
   }\right)-\alpha_\H \log \left(\frac{1}{-\alpha_\H x \epsilon +\alpha_\L (x-1)
   \epsilon +1}\right)\\
   \nonumber &&+\alpha_\L \log \left(\frac{1}{-\alpha_\H x \epsilon +\alpha_\L
   (x-1) \epsilon +1}\right)
   -\alpha_\L \log \left(\frac{1}{1-\alpha_\H \epsilon}\right)
   +\alpha_\H \alpha_\L \epsilon  \log \left(\frac{1}{1-\alpha_\H \epsilon
   }\right)\\
   \nonumber &&+\alpha_\H \log \left(\frac{1}{1-\alpha_\L \epsilon }\right)
   -\alpha_\H
   \alpha_\L \epsilon  \log \left(\frac{1}{1-\alpha_\L \epsilon }\right)
   -\beta \log
   \left(\frac{1}{1-\alpha_\H \epsilon }\right)\\
   &&+\alpha_\H \beta \epsilon  \log
   \left(\frac{1}{1-\alpha_\H \epsilon }\right)
   +\beta \log \left(\frac{1}{1-\alpha_\L
   \epsilon }\right)
  -\alpha_\L \beta \epsilon  \log \left(\frac{1}{1-\alpha_\L \epsilon
   }\right). \label{eq:interiormax}
\end{eqnarray}
% Restore the current equation number.
%\setcounter{equation}{\value{MYtempeqncnt}}
% IEEE uses as a separator
\hrulefill
% The spacer can be tweaked to stop underfull vboxes.
\vspace*{4pt}
\end{figure*}
In Appendix \ref{app:continuous} we show that this condition leads to an interior maximum at a unique value of $x$ as $\epsilon\to 0$. 

\subsubsection{Implicit expression for $x_\text{opt}$ in the limit $\epsilon\to 0$, and an expression for the \dr{limiting capacity of the discrete time channel in the small $\epsilon$ limit}}

Define the continuous time information rate, as a function of the fraction  of time the input is in the higher state ($x=p_\H$), as
\begin{align}
\mathcal{I}&=\lim_{\epsilon\to 0^+} \left(\frac{1}{\epsilon}I(X;Y)\right).
%\\
%&=\lim_{\epsilon\to 0^+} \left(\frac{1}{\epsilon}
%\frac{\binent(\epsilon\alpha_\H x + \epsilon \alpha_\L (1-x)) 
%		- x \binent(\epsilon\alpha_\H) - (1-x) \binent(\epsilon\alpha_\L)}
%	{1 + (x \alpha_\H + (1-x) \alpha_\L)/\beta} \right).
\end{align}
with $I(X;Y)$ given in (\ref{eqn:MIperdiscretestep}).
From the preceding section, we know this expression converges to a finite value, and moreover $\mathcal{I}$ has a unique maximum in the range $0\le x \le 1$.
Let $x_\text{opt}$ denote the optimal value of the high-input probability.  It is straightforward to show that %\footnote{See \texttt{projects/bioinfo/code/pjthomas/mathematica/MIRate.nb}.}
\begin{align}\label{eq:continuoustimecapacity0}
\lefteqn{\mathcal{I}(x)=} & \\
& -\left(\frac{\beta}{\beta+\bar{\alpha}} \right)\Big(  \bar{\alpha}\log(\bar{\alpha}) - (x\alpha_\H\log\alpha_\H + (1-x)\alpha_\L\log\alpha_\L)  \Big)\nonumber \end{align}
where, as above, $\bar{\alpha}(x)=x\alpha_\H + (1-x)\alpha_\L$ is the average value of $\alpha$ given $x$.
Thus the mutual information rate is given by the product of the fraction of time the channel is in the receptive state,
$$f(x)\equiv\left(\frac{\beta}{\beta+\bar{\alpha}(x)} \right)$$
and the mutual information rate conditional on the channel being in the receptive state,
$$g(x)\equiv -( \bar{\alpha}(x)\log\bar{\alpha}(x) - (x\alpha_\H\log\alpha_\H + (1-x)\alpha_\L\log\alpha_\L)).$$
Although the optimal value of the high input probability $x$ is not available explicitly, we can obtain a useful implicit expression for the \dr{small $\epsilon$ limiting capacity of the discrete time channel}.  Setting $\mathcal{I}'=f'g+f g'=0$, and noting that  $f'<0,$ %\footnote{$f'(x)=-(\alpha_{\H}-\alpha_{\L})\beta/(\bar{\alpha}(x)+\beta)^2$ and  $g'(x) = \alpha_\H\log(\alpha_\H)-\alpha_\L\log(\alpha_\L)-(\alpha_\H-\alpha_\L)(1+\log(\bar{\alpha}(x)))$.} 
we have $g(x_\text{opt})=f(x_\text{opt})g'(x_\text{opt})/f'(x_\text{opt})$, from which 
\begin{eqnarray}
\nonumber\lefteqn{\mathcal{I}(x_\text{opt})} & & \\
&=& \frac{g'(x_\text{opt})}{f'(x_\text{opt})}f(x_\text{opt})^2\\
\nonumber
&=&\left(\frac{\beta}{\alpha_\H-\alpha_\L}\right)
\Big( \alpha_\H\log\alpha_\H - \alpha_\L\log\alpha_\L\\ 
& & -\: (\alpha_\H-\alpha_\L)(1+\log\bar{\alpha}(x_\text{opt})) \Big). \label{eq:continuoustimecapacity}
\end{eqnarray}
In the continuous time setting, we have an ambiguity associated with the choice of the time unit.  Provided the low binding rate is not identically zero, we can choose time units with respect to which the low binding rate $\alpha_\L\equiv 1$.  Let $\alpha_\H=1+c$ in the same units.  Thus, from \eqref{eq:continuoustimecapacity}, the capacity is given by 
\begin{equation}\label{eq:continuoustimecapacity2}
\mathcal{I}(x_\text{opt})=\beta\left(\frac{1+c}{c}\log(1+c) - 1-\log(1+x_\text{opt} c) \right).
\end{equation}
Although $x_\text{opt}$ is not known explicitly, it lies in the interior of the unit interval, so we have the upper and lower bounds
\begin{eqnarray}
\nonumber
& \beta\left(\frac{1+c}{c}\log(1+c) - 1-\log(1+ c) \right) &\\
& \le \mathcal{I}(x_\text{opt}) \le &  \\
\nonumber
& \beta\left(\frac{1+c}{c}\log(1+c) - 1 \right) . &
\end{eqnarray}

Kabanov obtained the capacity of the Poisson channel with signal intensity bounded above by a constant $c$, and unit background intensity.  The ``high'' and ``low'' rates of the incoming process (combining signal and noise) were $1+c$ and $1$, respectively.   In the limit as the unbinding rate grows without bound, we expect that our channel should be equivalent to Kabanov's Poisson channel.  Note that because the optimal sending input distribution depends on the channel parameters, including $\beta$, the expression (\ref{eq:continuoustimecapacity2}) does not necessarily diverge as $\beta\to\infty$.  In the next section we recover Kabanov's capacity formula in this limit.

\subsection{Reduction of the 2-State Signal-Transduction Channel to Kabanov's Poisson Channel}
\label{ssec:ReductionToKabanov}

\bl{In the introduction, we discussed Kabanov's capacity \cite{Kabanov1978}, which assumes $k_- \rightarrow \infty$ (i.e., the $\B \rightarrow \U$ transition
is immediate and instantaneous). In this section, we come full circle by showing that Kabanov's capacity formula
emerges when we take the limit as $\epsilon \rightarrow 0$ and $k_- \rightarrow \infty$ \dr{of the discrete time BIND channel}.\footnote{\bl{Note that $k_- \rightarrow \infty$ 
is not realistic in a biological system. Kabanov's result is normally applied to optical detection systems.}}}

%It is of interest to consider the case of a ``perfect'' receiver that has an unbinding rate much faster than any other rate in the system.  As a thought experiment, l
Let us suppose that with a discrete time step $\epsilon$, we may set the unbinding \emph{rate} $k_- =1/\epsilon$, so that the unbinding \emph{probability} is \bl{fixed at} $\beta=\epsilon k_- =1$; thus, $k_-$ is set 
to the highest possible rate such that $\beta$ is a valid probability. Further, recall that $\alpha_\H = \epsilon k_+ c_\H$ and $\alpha_\L = \epsilon k_+ c_\L$, and let $\bar{k}_x = x k_+ c_\H + (1-x) k_+ c_\L$ represent the average binding rate.

Suppose 
$\epsilon \rightarrow 0$; setting $k_- =1/\epsilon$, this means $k_- \rightarrow \infty$.
The continuous time channel capacity is still bounded, provided the low sending rate $\alpha_\L>0$,   
%Scaling $\beta$ in this way gives
%
so we can write
\begin{eqnarray}
\nonumber \lefteqn{ \mathcal{I}_\epsilon(x) =} & & \\ 
\nonumber & & \frac{1}{\epsilon}\left(\frac{1}{1+\epsilon\bar{k}_x}\right)
\Big(\binent(\epsilon\bar{k}_x)\\
& & \:\:-\:x\binent(\epsilon k_+ c_\H)-(1-x)\binent(\epsilon k_+ c_\L) \Big)
\\
\nonumber \lefteqn{\lim_{\epsilon\to 0}\mathcal{I}_\epsilon(x)\equiv\mathcal{I}_0(x)} &&\\
\nonumber & &  =x k_+ c_\H \log k_+ c_\H \\
& & \:\:+\: (1-x)k_+ c_\L \log k_+ c_\L - \bar{k}_x \log \bar{k}_x\\
\nonumber \lefteqn{\frac{d}{dx}\mathcal{I}_0(x) =} & & \\ 
\nonumber & & k_+ c_\H \log k_+ c_\H - k_+ c_\L \log k_+ c_\L\\
& & \:\:-\:( k_+ c_\H -k_+ c_\L)(1+\log\bar{k}_x)
\end{eqnarray}
and $\frac{d}{dx}\mathcal{I}_0(x) = 0$ when
%
%\begin{eqnarray}
%\lefteqn{\log\bar{k}_x =}& & \\&&\left(\frac{1}{k_+ c_\H - k_+ c_\L}\right)\\
%&& \cdot\left(k_+ c_\H \log\left(\frac{k_+ c_\H}{e} \right)- k_+ c_\L \log\left(\frac{k_+ c_\L}{e} \right) \right).
%\end{eqnarray}
\begin{align}
\log\bar{k}_x =& \left(\frac{1}{k_+ c_\H - k_+ c_\L}\right)\cdot
\\ &\cdot\left(k_+ c_\H \log\left(\frac{k_+ c_\H}{e} \right)- k_+ c_\L \log\left(\frac{k_+ c_\L}{e} \right) \right).\nonumber
\end{align}
Recall that $x_\text{opt}$ denotes the optimal probability of the high-concentration signal. Since $\bar{\alpha}=\alpha_\L+x_\text{opt}(\alpha_\H-\alpha_\L)$, we have
\begin{equation}
x_\text{opt}=\frac{\bar{k}_x-k_+ c_\L}{k_+ c_\H - k_+ c_\L}
\end{equation}
and
\begin{equation}
(1-x_\text{opt})=\frac{k_+ c_\H -\bar{k}_x}{k_+ c_\H-k_+ c_\L}.
\end{equation}
Consequently the capacity, $\mathcal{I}_0^*=\mathcal{I}_0(x_\text{opt})$, reduces to 
%\begin{eqnarray}
%\lefteqn{\mathcal{I}^*_0} & & \nonumber \\
%&=&\bar{k}_{x_\text{opt}}-\left(\frac{k_+ c_\H k_+ c_\L}{k_+ c_\H-k_+ c_\L} \right)\log\left(\frac{k_+ c_\H}{k_+ c_\L} \right)\\
%\nonumber&=&\exp\left[\left(\frac{1}{c_\H-c_\L}\right)\left(c_\H \log\left(\frac{k_+ c_\H}{e} \right)-c_\L \log\left(\frac{k_+ c_\L}{e} \right) \right) \right]\\
%& & \:\:-\:\left(\frac{k_+ c_\H c_\L}{c_\H- c_\L} \right)\log\left(\frac{c_\H}{c_\L} \right).
%\end{eqnarray}
%
\begin{align}
\mathcal{I}^*_0 
&=\bar{k}_{x_\text{opt}}-\left(\frac{k_+ c_\H k_+ c_\L}{k_+ c_\H-k_+ c_\L} \right)\log\left(\frac{k_+ c_\H}{k_+ c_\L} \right)\\
\nonumber&=\exp\left[\left(\frac{1}{c_\H-c_\L}\right)\left(c_\H \log\left(\frac{k_+ c_\H}{e} \right)-c_\L \log\left(\frac{k_+ c_\L}{e} \right) \right) \right]\\
&  \:\:-\:\left(\frac{k_+ c_\H c_\L}{c_\H- c_\L} \right)\log\left(\frac{c_\H}{c_\L} \right).
\end{align}
Now let $\lambda = k_+ c_\L$ and $c = k_+ c_\H - k_+ c_\L$.
Since $k_- \rightarrow \infty$, this corresponds to a Poisson channel, alternating
between rates $\lambda$ and $\lambda + c$; these are Kabanov's parameters.
Substituting into the above equation,
%
%Consider an optical detection channel for which the source of the signal is a Poisson process of modulated intensity with range from zero to an upper limit $c>0$.  Suppose the noise in the channel comes from a background Poisson source with constant rate $\lambda>0$.  The peak of the detection rate (false positive + true positive) is then $\lambda+c$.  If the detector \bl{is insensitive} following each detection that is exponentially distributed with parameter $\beta$, then the corresponding channel is formally identical to our ligand binding/unbinding channel with binding rates $\alpha_\L=\lambda$ (low rate), $\alpha_\H=\lambda+c$ (high rate), and unbinding rate $\beta$.  Taking the  $\beta\to\infty$ limit as above, we can rewrite the capacity as 
\begin{eqnarray}
\nonumber\lefteqn{\mathcal{I}_o^*(c,\lambda)}&&\\
&=&
\lambda+x_\text{opt}c - \left(\frac{\lambda(\lambda+c)}{c}\right)\log\left( \frac{\lambda+c}{\lambda} \right)  \\
\nonumber &=&\exp\left[ \frac{\lambda+c}{c}\log\left(\frac{\lambda+c}{e}\right)-\frac{\lambda}{c}\log\left( \frac{\lambda}{e} \right) \right]\\
& & \:\:-\:\frac{\lambda}{c}(\lambda+c)\log\left( \frac{\lambda+c}{\lambda} \right)\\
\nonumber &=&\exp\left[\left(1+\frac{\lambda}{c}\right)\log(\lambda+c)-\frac{\lambda}{c}\log\lambda -1 \right]\\
& & \:\:-\:\lambda\left(1+\frac{\lambda}{c}\right)\log\left(1+\frac{c}{\lambda}\right)\\
\nonumber &=&\exp\left[\log\left[\frac{(\lambda+c)^{(1+\lambda/c)}}{\lambda^{(\lambda/c)}e}\right]\right]\\
& & \:\:-\:\lambda\left(1+\frac{\lambda}{c}\right)\log\left(1+\frac{c}{\lambda}\right)\\
\nonumber&=&\lambda\left[ \frac{1}{e}\left(1+\frac{c}{\lambda}\right)^{(1+\lambda/c)}
-\left(1+\frac{\lambda}{c}\right)\log\left(1+\frac{c}{\lambda}\right) \right].\\
\label{eq:davis-capacity}
\end{eqnarray}
Setting $\lambda=1$ in (\ref{eq:davis-capacity}) %, which is equivalent to choosing new units for time, 
yields
\begin{align}
%\label{Kabanov}
\label{eq:kabanov-capacity}
%C_{\text{Kabanov}}(c) 
\mathcal{I}_o^*(c,1)
&= \frac{1}{e}\left( (c+1)^{1+c^{-1}}\right)-\left(1+\frac{1}{c} \right)\log(c+1)\\
&= C_{\text{Kabanov}}(c) .
\end{align}

%Verd\'{u} describes Kabanov's capacity result in his paper ``Poisson communication Theory" \cite{Verdu1999PoissonTalk}.
%The capacity for the Poisson channel with background noise at rate $\lambda=1$ and signal bounded by intensity $c$ was first obtained by Kabanov \cite{Kabanov1978}, and provides the foundation for 
%information theoretic analysis of channels with Poisson signals (reviewed in \cite{Verdu1999PoissonTalk}).  

The analysis of the Kabanov/Poisson channel has been elaborated in numerous ways.  In \cite{Davis1980ieeeIT}, Davis gives the following formula for the capacity of the Poisson channel with noise rate $\lambda$ and signal rate bounded by $c$, namely
$$C_{\text{Davis}}(c,\lambda)=\frac{1}{e}(\lambda+c)\left(1+\frac{c}{\lambda}\right)^{\lambda/c}-\lambda\left(1+\frac{\lambda}{c}\right)\log\left(1+\frac{c}{\lambda}\right)$$
which is identical to (\ref{eq:davis-capacity}), see Equation (4b) in \cite{Davis1980ieeeIT} and also Equation (5) in \cite{Frey1991IEEE}.

We emphasize that Kabanov did much more than derive the formula.  He proved in \cite{Kabanov1978} that (\ref{eq:kabanov-capacity})  is the capacity for the Poisson channel and also that the capacity cannot be increased via feedback.  While our rigorous proofs are restricted to the discrete time case of the ligand binding/unbinding channel,  the consistency of the limiting \dr{(vanishing time step)} expressions with Kabanov's formula suggests that the analogy is sound.

\section{Discussion}
\label{sec:discuss}

\subsection{\bl{The POST channel, the BIND channel, and finite state channels}}
\label{ssec:POST}

\bl{The POST channel \cite{AsnaniPermuterWeissman2013IEEE_ISIT,PermuterAsnaniWeissman2014IEEETransIT} the trapdoor or chemical channel \cite{Blackwell1961chapter,PermuterCuff_vanRoy_Weissman2007IEEE_ISIT}
and our BIND channel are examples of finite state Markov channels, a broad class of channels which are essential for understanding 
signal transduction systems.
In this section we compare the POST
and BIND channel models, and show that they are not reducible to each other, while putting both channels in the wider context of finite state channels.}

\bl{Finite state channels have a long history in information theory \cite{BlackwellBreimanThomasian1958AnnalsMathStat}.  For instance, Blackwell discusses them in his 1961 book chapter \cite{Blackwell1961chapter}, and introduces the trapdoor channel as a simple, but still unsolved, example.
(Permuter and colleagues obtained the feedback capacity for the trapdoor channel by formulating and solving an equivalent dynamic programming problem \cite{PermuterCuff_vanRoyWeissman2006arXiv,PermuterCuff_vanRoy_Weissman2007IEEE_ISIT}.)
Capacity of finite state channels has long been an interesting, and difficult, problem for
information theorists (see, e.g., \cite{GoldsmithVaraiya1996ieeeIT-markovchannels}).
Important recent results were provided by Chen, Yin, and Berger \cite{che05,yin-unpublished} for
the class of unit output memory (UOM) finite state channels, where the channel output and channel state are identical, and where the channel output (i.e., state) is provided as feedback to the transmitter with unit delay (see \cite[Fig. 2]{che05}). It should be clear that the BIND channel is a UOM channel, as we used some of these results in \S \ref{sec:Capacity}.}

\bl{The Past Output is the STate (POST) channel, also a UOM channel, was introduced by Permuter, Asnani and Weissman  \cite{AsnaniPermuterWeissman2013IEEE_ISIT,PermuterAsnaniWeissman2014IEEETransIT}.
Two specific channel models, \posta~and $\text{POST}(a,b)$, were analyzed; the state transition probabilities in these models were carefully selected to be symmetric, in the sense that the channel architecture is invariant under simultaneous relabeling of the binary inputs and outputs. This symmetry
allows the authors to establish that feedback does not increase the capacity of the
channels they study. The BIND model, which is derived from the physiology of biological signal transduction, does not have this symmetry; this makes it both biologically relevant, and distinct from   the \posta~and $\text{POST}(a,b)$ channels.}

\bl{To illustrate the difference,
Table \ref{tab:POSTchannel} shows the transition probabilities for the \posta, $\text{POST}(a,b)$, and BIND channels.  
\begin{table}[htbp]
   \centering
   \begin{tabular}{@{} cccccc @{}} % Column formatting, @{} suppresses leading/trailing space
      \toprule
      $Y_{i-1}$ & $X_i$ & $Y_{i}$  &  \posta & $\text{POST}(a,b)$ &
      $\text{BIND}(\alpha_\L,\alpha_\H,\beta)$ \\ \hline
      0&0&0&1&$a$ & $1-\alpha_\L$\\
      0&0&1&0&$1-a$ & $\alpha_L$\\
      0&1&0&$\alpha$&$1-b$ & $1-\alpha_\H$\\
      0&1&1&$1-\alpha$&$b$ & $\alpha_\H$\\
      \hline
      1&0&0&$1-\alpha$&$b$ & $\beta$\\
      1&0&1&$\alpha$ &$1-b$ & $1-\beta$\\
      1&1&0&0&$1-a$&$\beta$\\
      1&1&1&1&$a$&$1-\beta$\\
      \bottomrule
   \end{tabular}
   \caption{Transition probabilities for the \posta, $\text{POST}(a,b)$, and BIND channels. Given the previous channel state, $Y_{i-1}$, and the next input, $X_i$, the table gives the probability of the next channel state, $Y_i$, under different channel models.}
   \label{tab:POSTchannel}
\end{table}
The symmetry of the \posta~channel is clear from the table: under simultaneous relabeling of the binary input and output states ($0\leftrightarrows 1$) the probabilities in the \posta~column remain unchanged.  The asymmetry of the BIND channel is similarly clear; since the two channel states exhibit entirely different behaviours, no relabeling of the states and inputs can recover the POST channel, except in the trivial case where $\alpha_\L = \alpha_\H = \beta$.}

\bl{Although the BIND and \posta/ $\text{POST}(a,b)$ channels are fundamentally different, they share the property that their capacities are not increased by feedback.  However, this property arises through distinct mechanisms.  The label-exchange symmetry of the POST channels guarantees that an optimal input strategy exists that is agnostic about the channel output, even when feedback information is available.  In contrast, the BIND channel has one input-sensitive and one input-insensitive state.  As established through our application of Chen and Berger's conditions, knowing when the channel is in the insensitive bound state does not change the optimal input strategy.}

\subsection{\bl{Biological Significance}}
\bl{Advances in high-throughput technologies that can measure the responses of populations of cells to chemical signals at the individual cell level have made possible the quantitative application of information theory by experimental biologists and biophysicsts.  Examples include information theoretic analysis of experiments measuring the encoding of visual information in the H1 neuron of the fly \cite{Ruyter-van-SteveninckLewenStrongKoberleBialek1997Science,spikesbook1999}, the encoding of gradient direction in the movement of the \textit{Dictyostelium} ameoba \cite{FullerChenAdlerGroismanLevineRappelLoomis2010PNAS,nips:ThomasEtAl:2004}, and the encoding of tumor necrosis factor (TNF) signal intensity in the response of the nuclear factor kappa B (NF-$\kappa$b) and activating transcription factor-2 (ATF-2)  pathways \cite{CheongRheeWangNemenmanLevchenko2011Science,Thomas:2011:Science}.  In these experiments, information theory does not so much provide a \emph{prediction} that can be confirmed or falsified by the experimental outcome; rather it provides a \emph{framework} that allows the experimenter to meaningfully ask ``how much information does this biological pathway carry"?  }

\bl{Signaling \textit{via} diffusible ligand molecules is a ubiquitous mechanism for communication between living cells. In this paper we have formulated and solved a novel discrete time finite state channel -- the BIND channel -- that captures the ligand-receptor binding/unbinding process present in the simplest type of signal transduction mechanism.  Although the introduction and analysis of the channel is the main contribution of the paper, it is natural to ask how our results compare with known properties of ligand-receptor--based signaling systems.  We offer two observations.}

\bl{In Theorem \ref{thm:maxmin} we show that, given a range of possible input concentrations, the optimal use of the channel concentrates the input signal on the extreme values, $x=\L$ and $x=\H$.   This conclusion directly contradicts the common assumption that input signals are ``small''.  The latter assumption has been made in order to approximate biochemical signaling systems with linear time-invariant systems (see e.g.~\cite{ClausznitzerEndres2011BMCSysBiol,PierobonAkyildiz2010IEEE_JSAC}), which are easier to analyze than systems that function at the extremes of their operating range.  The prediction that biological pathways should tend to use binary (alternately large or small, rather than graded) signals is confirmed in many biological systems.  For example, neurotransmitter release in central nervous system synapses is all-or-none, with large transient changes in concentration rather than smoothly graded changes.  The social amoeba \textit{Dictyostelium discoideum} signals in sharply concentrated waves separated by very low signal concentrations \cite{Kessin:2001}.  Our BIND channel (originally motivated by \textit{Dictyostelium's} cyclic AMP receptor) is consistent with this behavior.}

\bl{Our Theorem \ref{thm:main} establishes that feedback does not increase the capacity of the BIND channel.  The \textit{Dictyostelium} amoeba uses cAMP to orient towards other conspecific cells during aggregation of the colony; each amoeba responds to the received cAMP signal by secreting its own discharge of cAMP, which serves to relay the aggregation signal to other amoebas further from the aggregation center.  However, the identity of the cell from which a particular cAMP molecule originated is unknown to the amoeba receiving that molecule.  We are not aware of any mechanism by which the amoeba can regulate its pattern of cAMP secretion taking into account the state of the receptor(s) on other cells.  That is, the \textit{Dictyostelium} amoeba does not, to our knowledge, use feedback to enhance signaling via the cAMP receptor.
However, biological systems are diverse, and the BIND channel reflects only the simplest form of ligand-receptor pathway.  Some signaling systems with more elaborate pathway structure are known to use bidirectional signaling \cite{KullanderKlein2002NatRevMolecCellBiol}, which could be interpreted as a form of feedback.  In \S \ref{ssec:general} we provide an example of a ligand-receptor channel with two binding sites, for which feedback \emph{would} appear to increase the capacity.  Clearly, more elaborate channel models provide fertile ground for further investigation.} 

\subsection{Towards the Capacity of General Signal Transduction Channels} \label{ssec:general}

In this paper, we calculated the capacity of a simple signal transduction channel,
related to the cAMP receptor in {\em Dictyostelium}, and derived many useful properties of
mutual information. 
Our contribution is one of a rapidly growing body of work applying information theory
to biological communication problems. 
Indeed, a natural open problem suggested by our work is to extend Kabanov's continuous time Poisson channel to a family of channels defined by continuous time Markov chains on finite graphs.
Here we consider some features of this generalized problem.

One may consider the input signal to a general ``signal-transduction"  continuous-time Markov channel 
as any physical or biochemical process that varies the transition rate intensities between the vertices of the graph, with the output signal comprising either the transitions themselves or a related counting process on one or more vertices.  Viewed in this way,  the Kabanov-Poisson channel comprises a ``graph'' with a single vertex, with a single counting process instead of a multicomponent marked point process.

Analysis of the capacity for a general $n$-state signal-transduction channel, such as described by \eqref{eq:cts-time-general}, remains an interesting open problem.  In this paper, we considered the case $n=2$, in a sense the simplest generalization of the Poisson channel.  For our two-state signal-transduction channel, the mutual information rates in both the discrete time setting  \eqref{eqn:CIIDLH} and in the continuous time setting \eqref{eq:continuoustimecapacity0} decompose into the product of an information rate conditional on occupying a ``sensitive" state, and the fraction of time the system occupies that state.  However, as we already stated in the introduction, many higher-order Markov models are available for different kinds of receptors, so the 
generalized problem is of significant practical interest.

First, a simple extension of our results in Section \ref{sec:Capacity} gives the mutual information
of a general $n$-state receptor under IID inputs. For receptor states $i$ and $j$, $1 \leq i,j \leq n$, and input concentration $x$, taking discrete levels in $1 \leq x \leq m$,
let $\alpha_{i,j,x}$ represent the transition probability from state $i$ to state $j$ under
input concentration $x$. Let $p$ represent the $m$-dimensional 
vector containing the IID input distribution.
Let $\bar{\alpha}_{i,j} = \sum_{x=1}^m \alpha_{i,j,x} p_x$ represent the average transition
probability from $i$ to $j$. Under an IID input distribution,  
the sequence of receptor states $Y$ forms a regular Markov chain with transition probability matrix
$\mathbf{P}_Y = [\bar{\alpha}_{i,j}]$. If $p_Y$ is the stationary distribution on the receptor states,
given by the normalized eigenvector of $\mathbf{P}_Y$ with eigenvalue 1, 
and recalling $\phi(\cdot)$ from (\ref{eqn:PhiEquation}),
then the mutual information under IID inputs is given by
\begin{equation}\label{eq:MI_IID_general}
	I(X;Y) = \sum_{i = 1}^n p_{Y,i} \sum_{j=1}^n 
	\left(\phi(\bar{\alpha}_{i,j}) - \sum_{x = 1}^m p_x \phi(\alpha_{i,j,x}) \right) .
\end{equation}

However, it is clear that for a Markov channel taking the form of an arbitrary network, it is not generally true that $C_{\text{IID}}=C_{\text{FB}}$, as the following example illustrates. 

Consider a channel with three states arranged in  a chain
\begin{equation}
\begin{array}{ccccc}
&(\alpha_\L\text{ or }\alpha_\H) & & (\epsilon\text{ or }1-\beta) & \\
1 & \rightleftharpoons & 2 & \rightleftharpoons & 3\\
&\beta & & \epsilon &
\end{array},
\end{equation}
where $0<\alpha_\L<\alpha_\H<1$, $0< \beta < 1$, and $0<\epsilon<1-\beta$.  The $1\to 2$ and $2\to 3$ transition probabilities depend on the input (assumed binary for this example) in the same manner as in \S \ref{sec:MarkovInputs}.  That is, we have $\alpha_{1,2,\L}=\alpha_\L$ and $\alpha_{1,2,\H}=\alpha_\H$, and $\alpha_{2,3,\L}=\epsilon$ while $\alpha_{2,3,\H}=1-\beta$.  The other transitions are insensitive to the input value $x$, i.e.~$\alpha_{2,1,x}=\beta$ and $\alpha_{3,2,x}=\epsilon$ independently of $x$.  Hence the transitions out of state 3 do not carry information about the input.  Given the input probabilities $p_\H+p_\L=1$, the transition matrix of the channel state for IID input is
\begin{align}
\nonumber\lefteqn{T_{\text{IID}}=}&\\
& \left(\begin{array}{ccc}1-\bar{\alpha} & \bar{\alpha} & 0 \\
\beta & 1-\beta-(p_\L\epsilon +p_\H(1-\beta)) & p_\L\epsilon +p_\H(1-\beta)\\
0 & \epsilon & 1-\epsilon
\end{array}
\right),
\end{align}
where $\bar{\alpha}=p_\L\alpha_\L + p_\H\alpha_\H$,
and the stationary distribution is 
\begin{align}
p_{Y,1}&=\frac{\beta}{Z_{\text{IID}}},\\
p_{Y,2}&=\frac{\bar{\alpha}}{Z_{\text{IID}}},\\
p_{Y,3}&=\frac{p_\H(1-\beta)\bar{\alpha}}{\epsilon Z_{\text{IID}}},\\
Z_{\text{IID}}&=\bar{\alpha}+\beta+\frac{p_\H(1-\beta)\bar{\alpha}}{\epsilon}
\end{align}
\bl{From \eqref{eq:MI_IID_general} we obtain the mutual information for the three-state channel with IID inputs:
\begin{align}
\nonumber\lefteqn{I_3^{\text{IID}} =} & \\
\nonumber& \frac{1}{Z_\text{IID}}\Big(\beta( \binent(\bar{\alpha})-\left[ p_\H\binent(\alpha_\H)+p_\L\binent(\alpha_\L) \right] )\\
&+\bar{\alpha}(\binent(p_\H(1-\beta)+p_\H\epsilon)-\left[ p_\H\binent(1-\beta)+p_\L\binent(\epsilon) \right] ) \Big).
\end{align}
}%
The mutual information for a given input distribution is reduced, compared to that of the two-state channel, because the channel gets trapped in the long-lived, insensitive state 3, thus reducing the fraction of time spent in the sensitive state 1.

In case the sender is informed of the state of the channel, the sender may arrange to send input $x=\L$ whenever the channel is in state 2, thus reducing the rate at which the channel enters the trap in state 3.  In case the sender adopts this strategy, the transition matrix for the channel state becomes
\begin{equation}
T_{\text{FB}}=\left(\begin{array}{ccc}1-\bar{\alpha} & \bar{\alpha} & 0 \\
\beta & 1-\beta-\epsilon & \epsilon \\
0 & \epsilon & 1-\epsilon
\end{array}
\right),
\end{equation}
and the stationary distribution is 
\begin{equation}
p_{Y,1}=\frac{\beta}{Z_{\text{FB}}},\quad p_{Y,2}=\frac{\bar{\alpha}}{Z_{\text{FB}}},\quad p_{Y,3}=\frac{\bar{\alpha}}{Z_{\text{FB}}},\quad Z_\text{FB}=2\bar{\alpha}+\beta.
\end{equation}
The capacity under this feedback scheme is 
\bl{
\begin{equation}
I_3^\text{FB}=\frac{1}{Z_\text{FB}}(\beta( \binent(\bar{\alpha})-\left[ p_\H\binent(\alpha_\H)+p_\L\binent(\alpha_\L) \right] )).
\end{equation}
}
To compare the mutual information for any choice of input probabilities $p_{\L/\H}$ and parameters $\alpha_{\L/\H},\beta,\epsilon$, consider the ratio of the mutual information under the IID inputs versus the feedback scheme:
\bl{\begin{align}
\lefteqn{\frac{I_3^\text{IID}}{I_3^\text{FB}}=
\frac{\epsilon(2\bar{\alpha}+\beta)}{\epsilon(\bar{\alpha}+\beta)+p_\H(1-\beta)\bar{\alpha}} } & \\
&
\nonumber \cdot\Bigg(1+\\
&\left(\frac{\bar{\alpha}}{\beta}\right)\frac{\binent(p_\H(1-\beta)+p_\H\epsilon)-\left[ p_\H\binent(1-\beta)+p_\L\binent(\epsilon) \right] }{ \binent(\bar{\alpha})-\left[ p_\H\binent(\alpha_\H)+p_\L\binent(\alpha_\L) \right] }  \Bigg)\\
&\to 0,\hspace{1cm}\text{ as }\epsilon\to 0^+.
\end{align}}
That is, the ratio of the mutual information under the IID inputs versus inputs informed by the channel state can be made arbitrarily small, by taking the slow transition rates $\epsilon$ sufficiently small.  This suggests that the feedback capacity and the IID capacity cannot be equal for this simple example. 

The question of the regular capacity for this channel, and channels with arbitrary state graphs, remains an interesting problem for future work.

\appendix

\subsection{Stationary distributions achieve feedback capacity}
\label{app:stationary}

%Here we use the same indexing of the input/output process $(X^n,Y^n)$ as
%in Section \ref{sec:Capacity}.

\begin{figure*}[h!]
% ensure that we have normalsize text
\normalsize
% Store the current equation number.
%\setcounter{MYtempeqncnt}{\value{equation}}
% Set the equation number to one less than the one
% desired for the first equation here.
% The value here will have to changed if equations
% are added or removed prior to the place these
% equations are referenced in the main text.
%\setcounter{equation}{84}
%
%
\begin{equation}
	\label{eqn:QMatrix}
	\mathbf{Q}_j = 
	\left[ 
		\begin{array}{cccc}
			p_{Y_i|X_{i-1},Y_{i-1}}(\B \given c_1, j) &  
			p_{Y_i|X_{i-1},Y_{i-1}}(\B \given c_2, j) & 
			\ldots &
			p_{Y_i|X_{i-1},Y_{i-1}}(\B \given c_{m}, j) \\
			p_{Y_i|X_{i-1},Y_{i-1}}(\U \given c_1, j) &  
			p_{Y_i|X_{i-1},Y_{i-1}}(\U \given c_2, j) & 
			\ldots &
			p_{Y_i|X_{i-1},Y_{i-1}}(\U \given c_{m}, j)
		\end{array}
	\right] ,
\end{equation}
%
%
% Restore the current equation number.
%\setcounter{equation}{\value{MYtempeqncnt}}
% IEEE uses as a separator
\hrulefill
% The spacer can be tweaked to stop underfull vboxes.
\vspace*{4pt}
\end{figure*}

We start with several definitions.
Assuming that the input distribution is in $\mathcal{P}^*$ (i.e., $Y_1^n$ is a Markov chain),
and recalling (\ref{eqn:StateTransitionProbabilityMatrix}),
let $\hat{\mathbf{P}} = [\hat{P}_{ij}]$ represent a $2 \times 2$ matrix, taking values in $\{0,1\}$, with elements
\begin{equation}
	\hat{P}_{ij} =
	\left\{
		\begin{array}{cl}
			1, & \min_{k \in \{1,2,\ldots,m\}} P_{Y|X=k,ij} > 0 \\
			0, & \mathrm{otherwise} ,
		\end{array}
	\right.
\end{equation}
and for positive integers 
$\ell$, let $\hat{P}_{ij}^\ell$ represent the $i,j$th element of $\hat{\mathbf{P}}^\ell$.
Further, for the $i$th diagonal element of the $\ell$th matrix power $\hat{P}_{ii}^\ell$, 
let $\mathcal{D}_i$ contain the set of integers $\ell$ such that 
$\hat{P}_{ii}^\ell \neq 0$. 
Then: 
\begin{itemize}
	\item $Y_1^n$ is strongly irreducible if, for each pair $i,j$, there exists 
an integer $h > 0$
such that $\hat{P}_{ij}^h \neq 0$; and
	\item If $Y_1^n$ is strongly irreducible, it is also strongly aperiodic if,
	for all $i$, the greatest common divisor of $\mathcal{D}_i$ is 1.
\end{itemize}
These conditions are described in terms of graphs in \cite{che05}, but our description
is equivalent.

Let $\mathbf{Q}_j$ be a $2 \times m$ matrix,
		defined as in (\ref{eqn:QMatrix}),
and let
\begin{equation}
	\label{eqn:IPQ}
	I(p,\mathbf{Q}_i) = I(X_{j-1} ;Y_j \given Y_{j-1} = i)
\end{equation}
(cf. (\ref{eqn:DirectedTermSimplified})). 
For example, 
if $i=\U$, we have
\begin{equation}
	\label{eqn:IPQ1}
	I(p,\mathbf{Q}_\U) = \mathcal{H} ( \bar{\alpha}_p ) 
		- \sum_{i=1}^{m} p_i \mathcal{H}(\alpha_i) .
\end{equation}
We will use the following corollary to Theorem \ref{thm:maxmin}.
\begin{cor}
	\label{cor:maxmin}
	Let $p^\prime = [p_1^\prime, p_2^\prime, \ldots, p_m^\prime]$ represent
	the distribution satisfying
	\begin{equation}
		p^\prime = \arg \max_p I(p,\mathbf{Q}_\U) ,
	\end{equation}
	using $I(p,\mathbf{Q}_\U)$ from (\ref{eqn:IPQ1}).
	Then $p_2^\prime = p_3^\prime = \ldots = p_{m-1}^\prime = 0$.
\end{cor}
\begin{proof}
	The quantity in 
	(\ref{eqn:IPQ1}) is equal to the numerator of (\ref{eqn:CIID}).
	To prove the theorem, we relied only on terms in the numerator, so the same
	argument applies to this corollary.
\end{proof}

\begin{lemma}
	If the parameters are strictly ordered (Definition \ref{defn:StrictlyOrdered}),
	%If $0 < \alpha_1 < \alpha_2 < \ldots < \alpha_m < 1$, and $0 < \beta < 1$,
	then the conditions of \cite[Thm. 4]{che05} are satisfied, namely:
	\begin{enumerate}
		\item $Y^n$
		is strongly irreducible and strongly aperiodic.
		\item 
		(Reiterating \cite[Defn. 6]{che05}) for $j \in \{\U,\B\}$,  for the set of possible input distributions in $\mathcal{P}^*$, 
and for all $j \in \{\U,\B\}$, 
there exists a subset $\tilde{\mathcal{P}}^*$ satisfying%
\begin{enumerate}
	\item $\{\mathbf{Q}_j p : p \in \mathcal{P}^*\} = \{\mathbf{Q}_j p : p \in \tilde{\mathcal{P}}^*\}$. 
	\item For any $q \in \{\mathbf{Q}_j p : p \in \mathcal{P}^*]\}$,
	\begin{eqnarray}
	\label{eqn:Defn6Part2}
		\left\{ \arg \max_{p:p\in\mathcal{P}^*,\mathbf{Q}_j p=q} I(p,\mathbf{Q}_j) \right\} 
		\nonumber & & \\
		\cap 
		\left\{ \arg \max_{p:p\in\tilde{\mathcal{P}}^*,\mathbf{Q}_j p=q} I(p,\mathbf{Q}_j) \right\} 
		& \neq & \emptyset
	\end{eqnarray}
	\item There exists a positive constant $\lambda$ such that
	\begin{equation}
		\frac{\partial I(p,\mathbf{Q}_j)}{\partial \ell} - \frac{\partial I(q,\mathbf{Q}_j)}{\partial \ell}
		\leq
		- \lambda || p - q ||
	\end{equation}
	for any nonidentical $p,q \in \tilde{\mathcal{P}}^*$, where $\ell$ is in the direction from 
	$q$ to $p$, and the norm is the Euclidean vector norm.
\end{enumerate}
\end{enumerate}
\end{lemma}

\begin{proof}
	To prove the first part of the lemma, if 
	the parameters are strictly ordered, then
	%$0 < \alpha_1 < \alpha_2 < \ldots < \alpha_m < 1$, and $0 < \beta < 1$, then
	$\hat{\mathbf{P}}$ is an all-one matrix, so
	$Y^n$ is strongly irreducible (with $h = 1$); further, since the positive powers of an all-one
	matrix can never have zero elements, $\mathcal{D}_i$ contains all positive integers from 1 to $n$,
	whose greatest common divisor is 1, so $Y^n$ is strongly aperiodic.
	
	To prove the second part of the lemma,
	we first show that the definition is satisfied for $\mathbf{Q}_\B$, given by
\begin{equation}
	\label{eqn:QB}
	\mathbf{Q}_\B = 
	\left[ 
		\begin{array}{cccc} 
			1-\beta & 1-\beta & \ldots & 1-\beta \\ 
			\beta & \beta & \ldots & \beta
		\end{array}
	\right] .
\end{equation}
We choose the subset 
$\tilde{\mathcal{P}}^*$ to consist of a single point $p \in \mathcal{P}^*$ (it can be any point, as all
points give the same result).
The columns of $\mathbf{Q}_\B$ are
identical, since the output is not dependent
on the input in state $\B$.
Then for every $p\in\mathcal{P}^*$, 
\begin{eqnarray}
	\mathbf{Q}_\B p & = & 
		\left[ \begin{array}{ccc} 1-\beta & \ldots & 1-\beta \\ \beta & \ldots & \beta \end{array}\right] 
		\left[ \begin{array}{c} p_{1 |\B} \\ p_{2 |\B} \\ \vdots \\ p_{m|\B} \end{array}\right] \\
	& = & \left[ \begin{array}{c} (1-\beta) \sum_{j=1}^{m} p_{j |\B}\\ 
		\beta \sum_{j=1}^{m} p_{j |\B} \end{array}\right] \\
	& = & \left[ \begin{array}{c} (1-\beta) \\ \beta \end{array}\right] .
\end{eqnarray}
This is also true of the single point in $\tilde{\mathcal{P}}^*$,
so condition (a) is satisfied.
Similarly, by inspection of (\ref{eqn:QB}), when $Y_0 = \B$, 
the output $Y_1$ is not dependent on the input $X_1$, so 
$I(p,\mathbf{Q}_\B) = 0$ for all $p \in \mathcal{P}$. 
Since all $p \in \mathcal{P}^*$ ``maximize'' $I(p,\mathbf{Q}_\B)$
and have identical values of $\mathbf{Q}p$ (including the single point in $\tilde{\mathcal{P}}^*$),
then the single point $p \in \tilde{\mathcal{P}}^*$ is always in both sets, and the intersection
(\ref{eqn:Defn6Part2})
is nonempty; so condition (b) is satisfied. There is only one point in 
$\tilde{\mathcal{P}}^*$, so there is no pair of nonidentical points, and 
condition (c) is satisfied trivially.

Now we show that the conditions are satisfied for $\mathbf{Q}_\U$, given by
\begin{equation}
	\mathbf{Q}_\U = 
	\left[ 
		\begin{array}{cccc} 
			\alpha_1 & \alpha_2 & \ldots & \alpha_m \\ 
			1-\alpha_1 & 1-\alpha_2 & \ldots & 1-\alpha_m
		\end{array}
	\right] .
\end{equation}
%
%Obviously, $\mathrm{rank}(\mathbf{Q}_\U)$ is either 1 or 2. We deal with each case in turn.
Since the parameters are strictly ordered, $\mathrm{rank}(\mathbf{Q}_\U) = 2$. (The lemma is satisfied if $\mathrm{rank}(\mathbf{Q}_\U) = 1$, by the same argument we gave above, though in this case
$\alpha_1 = \ldots = \alpha_m$ and the capacity is zero.)
Now we have
\begin{equation}
	\mathbf{Q}_\U p =
	\left[
		\begin{array}{c}
			\sum_{j=1}^m p_i \alpha_i \\
			1 - \sum_{j=1}^m p_i \alpha_i 
		\end{array}
	\right]
	=
	\left[
		\begin{array}{c}
			\bar{\alpha}_p \\
			1 - \bar{\alpha}_p 
		\end{array}
	\right].
\end{equation}
Since the parameters are strictly ordered, $\bar{\alpha}_p$ can take any value on the interval $[\alpha_1,\alpha_m]$.

Let $\tilde{\mathcal{P}}^*$ represent the set of input distributions $p$ from Corollary \ref{cor:maxmin},
with $p_2 = p_3 = \ldots = p_{m-1} = 0$. For $p \in \tilde{\mathcal{P}}^*$,
\begin{equation}
	\mathbf{Q}_\U p =
	\left[
		\begin{array}{c}
			p_1 \alpha_1 + p_m \alpha_m \\
			1 - p_1 \alpha_1 - p_m \alpha_m 
		\end{array}
	\right] ,
\end{equation}
and $p_1 \alpha_1 + p_m \alpha_m$ can take any value on the interval $[\alpha_1,\alpha_m]$.
Therefore, condition (a) is satisfied.

From Corollary \ref{cor:maxmin}, all distributions $p$ maximizing 
$I(p,\mathbf{Q}_i)$ have $p_2 = p_3 = \ldots = p_{m-1} = 0$. 
Thus, all maximizing distributions in $\mathcal{P}^*$ are also in $\tilde{\mathcal{P}}^*$, and condition
(b) is satisfied.

Finally, by the definition of the directional derivative, condition (c) is equivalent to
\begin{equation}
	\label{eqn:DirectionalDerivative1}
	(p - q) \cdot \Big( \nabla_p I(p,\mathbf{Q}_\U) - \nabla_q I(q,\mathbf{Q}_\U) \Big) < 0 ,
\end{equation}
where $\cdot$ represents vector dot product. 
%Let $\bar{\alpha}_1 = \sum_{i=1}^m p_{1,i} \alpha_i$,
%and $\bar{\alpha}_2 = \sum_{i=1}^m p_{2,i} \alpha_i$.
Inequality (\ref{eqn:DirectionalDerivative1}) reduces to
\begin{equation}
	\frac{1}{\mathrm{log_e}\: 2} (\bar{\alpha}_p - \bar{\alpha}_q) \log 
	\frac{
		\bar{\alpha}_q - \bar{\alpha}_p\bar{\alpha}_q
	}
	{
		\bar{\alpha}_p - \bar{\alpha}_p\bar{\alpha}_q
	}
	 < 0
\end{equation}
By inspection, this inequality is satisfied as long as $\bar{\alpha}_p \neq \bar{\alpha}_q$. To check
when this is satisfied
in the subset $\tilde{P}^*$, we can write
\begin{align}
	\label{eqn:ConditionCSatisfied}
	\bar{\alpha}_q - \bar{\alpha}_p &= 
	\alpha_1 q_1 + \alpha_m q_m - \alpha_1 p_1 - \alpha_m p_m \\
	&= \alpha_1 q_1  + \alpha_m (1-q_1) - \alpha_1 p_1 - \alpha_m (1-p_1)\\
	&= (q_1 - p_1) (\alpha_1 - \alpha_m) .
\end{align}
where (\ref{eqn:ConditionCSatisfied}) follows from the definition of $\tilde{P}^*$.
By assumption, $\alpha_m \neq \alpha_1$. Thus, $\bar{\alpha}_q \neq \bar{\alpha}_p$ so
long as $q_1 \neq p_1$, i.e., for any distinct points in $\tilde{P}^*$. Thus, condition (c) is 
satisfied, and the lemma follows.
\end{proof}
Closely related 
results were given in the (unfortunately unpublished) \cite{yin-unpublished}, as well as stronger results for all possible binary-input, binary-output, unit-memory Markov channels.

\subsection{Entropy Rates via the Sum-Product Algorithm}
\label{sec:journalpaper-appendix-sum-product}

%We wish to calculate the mutual information rate between an input time series $X_0^n\dfn(X_0,\cdots,X_n)$ and an output time series $Y_0^n\dfn(Y_0,\cdots,Y_n)$ generated by a stationary Markov channel. 
In this appendix, we briefly explain how to use the sum-product algorithm
\cite{KschischangFreyLoeliger2001IEEE}, both to calculate bounds on mutual information 
and to perform the 
Monte Carlo simulations that were discussed in Section \ref{sec:MarkovInputs}.
%(This appendix uses the same indexing of the input/output process as in Section \ref{sec:MarkovInputs}.)

The channel is specified by the conditional probabilities $p_{Y_{i+1}|X_i,Y_i}(y_{i+1}|x_i,y_{i})$, with a Markov input process governed by transition probabilities $p_{X_{i+1}|X_i}(x_{i+1}|x_{i})$.  
%For notational compactness, define $p(x_1|x_0)\dfn p(x_1)$ and $p(y_1|x_1,y_0)\dfn p(y_1|x_1)$.
%\marginpar{How to correct indexing here?}
%
In order to approximately calculate the mutual information rate, 
\begin{align}
\mathcal{I}(X,Y) 
%&= \lim_{n\to\infty}\frac{1}{n}\left(H(Y^n)-H(Y^n|X^n)\right)\\
&= \lim_{n\to\infty}\left(H(Y_n|Y_0^{n-1})-H(Y_n|X_0^{n-1},Y_0^{n-1})\right)
\end{align}
the second term reduces (in the case of our Markov channel) to $H(Y_n|X_0^{n-1},Y_0^{n-1})=H(Y_n|X_{n-1},Y_{n-1})$, which is available in closed form.
Thus,
we need a way to estimate the first term, $H(Y_n|Y_0^{n-1})$, which requires the
calculation of two quantities: $p_{Y_k | Y_0^{k-1}}(y_k \given y_0^{k-1})$
and $p_{Y_k | X_0,Y_0^{k-1}}(y_k \given x_0,y_0^{k-1})$,
for various values of $k$.

%From \cite{CoverThomas1990} we have the following inequalities:
%\begin{eqnarray*}
%H\left(Y_n|Y_{n-k}^{n-1}\right)\ge H\left(Y_n|Y_1^{n-1}\right)\ge H\left(Y_n|Y_{n-k}^{n-1},X_{n-k}\right)
%\end{eqnarray*}
%If we assume stationarity, then the first term is equal to $H\left(Y_k|Y_1^{k-1}\right)$, and the last term is equal to $H\left(Y_k|Y_1^{k-1},X_1\right)$.  The more terms we take, the closer our bounds on the mutual information rate become. 
%

Calculation of $p_{Y_k|Y_0^{k-1}}\left(y_k|y_0^{k-1}\right)$ can be accomplished efficiently using the sum-product algorithm.
By defining a sequence of functions $\varphi_i(x_i)$, which act as ``messages" propagating along the factor graph, one obtains a recursive algorithm:
\begin{eqnarray}
\label{eqn:SumProduct1}
\nonumber\lefteqn{\varphi_0(x_0,y_0) =}&&\\
& &p_{X_0,Y_0}(x_0,y_0)\\
\label{eqn:SumProduct2}
\nonumber\lefteqn{\varphi_i(x_i,y_i) =}&&\\
\nonumber &&\sum_{x_{i-1}}p_{Y_i|X_{i-1},Y_{i-1}}(y_i|x_{i-1},y_{i-1})p_{X_i|X_{i-1}}(x_i | x_{i-1})\\
&&\:\:\cdot\varphi_{i-1}(x_{i-1},y_{i-1}),\mbox{ for }1<i\le k \\
\label{eqn:SumProduct3}
\nonumber\lefteqn{p_{Y_0^k}\left(y_0^k\right) =}&&\\
\nonumber&&\sum_{x_{k-1}}p_{Y_k|X_{k-1},Y_{k-1}}(y_k|x_{k-1},y_{k-1})\varphi_{k-1}(x_{k-1},y_{k-1}) ,\\
\end{eqnarray} 
where the probability on the right side of (\ref{eqn:SumProduct1}) is the steady-state probability
for the Markov process $(x_i,y_i)$.
This well-known algorithm arises from the decomposition of the probability $p_{Y_0^k}(y_0^k)$ into a sum of products:
\begin{eqnarray*}
\nonumber\lefteqn{p_{Y_0^k}(y_0^k)}&&\\
&=&\sum_{x_0^k} p_{X_0^k,Y_0^k}(x_0^k,y_0^k)\\
\nonumber &=&\sum_{x_0^k}p_{X_0,Y_0}(x_0,y_0)\\
&&\cdot\Bigg[\prod_{i=1}^k 
p_{Y_i|X_{i-1},Y_{i-1}}(y_i|x_{i-1},y_{i-1})
p_{X_i|X_{i-1}}(x_i | x_{i-1})
\Bigg] ,
\end{eqnarray*}
valid for our channel driven by a Markov input source.
Finally, we obtain $p_{Y_k|Y_0^{k-1}}\left(y_k|y_0^{k-1}\right)$ by
\begin{equation}
	\label{eqn:SumProductFinal}
	p_{Y_k|Y_0^{k-1}}\left(y_k|y_0^{k-1}\right) =
	\frac{p_{Y_0^k}(y_0^k)}{\sum_{y_k}p_{Y_0^k}(y_0^k)}
\end{equation}

Calculation of $p_{Y_k|X_0,Y_0^{k-1}}\left(y_k|x_0,y_0^{k-1}\right)$
proceeds similarly, except for $k = 1$, (\ref{eqn:SumProduct2}) is replaced by
\begin{equation}
	\varphi_1(x_1,y_1) = p_{Y_1|X_{0},Y_{0}}(y_1|x_{0},y_{0})p_{X_1|X_{0}}(x_1 | x_{0})\varphi(x_{0},y_{0}),
\end{equation}
i.e., we do not sum over $x_0$. The final result in (\ref{eqn:SumProduct3}) 
is then the joint probability $p_{Y_0^k,X_0}(y_0^k,x_0)$.
Finally, we obtain
\begin{equation}
	\label{eqn:SumProductFinalX0}
	p_{Y_k|X_0,Y_0^{k-1}}\left(y_k|x_0,y_0^{k-1}\right) =
	\frac{p_{Y_0^k}(y_0^k,x_0)}{\sum_{y_k}p_{Y_0^k}(y_0^k,x_0)} .
\end{equation}

The upper and lower bounds from Section \ref{sec:MarkovInputs} are obtained
by substituting (\ref{eqn:SumProductFinal})
into (\ref{eq:H-upper-bound-general}), and (\ref{eqn:SumProductFinalX0}) into
(\ref{eq:H-lower-bound-general}), respectively.

To calculate a Monte Carlo
%calculate $p\left(y_k|y_1^{k-1}\right)$ exactly, for
%sufficiently large $k$; this can be accomplished efficiently using the sum-product algorithm.
%The sum-product algorithm\footnote{Kschischang et al \cite{KschischangFreyLoeliger2001IEEE} review the sum-product algorithm, an efficient and widely used means of estimating joint and conditional probabilities, entropies and mutual informations for discrete time, discrete state bivariate processes.} allows efficient evaluation of the probability $p(y_1^k)$ and the conditional probability $p(y_k|y_1^{k-1})$, from which we can
estimate of the information rate, we obtain an estimate of
\begin{equation}
\label{eqn:MonteCarlo}
H\left(Y_k|Y_0^{k-1}\right)=\E\left[\log\left( \frac{1}{p_{Y_k|Y_0^{k-1}}\left(y_k|y_0^{k-1}\right)}  \right) \right]
\end{equation}
for sufficiently large $k$. 
Here we
generate sample sequences $y_0^k$ with the correct distribution,
calculate
$p_{Y_k|Y_0^{k-1}}\left(y_k|y_0^{k-1}\right)$ 
using (\ref{eqn:SumProductFinal}), and take the sample mean
to obtain the term under the expectation
in (\ref{eqn:MonteCarlo}).

\subsection{Critical point of the continuous time information rate}
\label{app:continuous}
In \S\ref{sssec:criticalpoint} we consider the information rate as the time step $\epsilon$ goes to zero.  We assume the mutual information rate has an interior maximum as a function of the high-state probability $x\equiv p_\H$ \bl{(recall that mutual information is concave with respect to the input distribution)}. Here we show that this maximum is unique.  Setting the derivative of the mutual information rate \eqref{eqn:MutualInfo} equal to zero gives a necessary and sufficient condition for the maximum, Equation \eqref{eq:interiormax}.
We may simplify \eqref{eq:interiormax} by introducing $\bar{\alpha}=x \alpha_\H + (1-x)\alpha_\L$, which gives
\begin{eqnarray*}
\nonumber\lefteqn{0 =}&&\\
&&-\alpha_\H \beta \epsilon  \log \left(\frac{1}{1-\epsilon\bar{\alpha}}\right)
   +\alpha_\L \beta \epsilon  \log \left(\frac{1}{1-\epsilon\bar{\alpha}}\right)\\
   &&+\beta \epsilon  (\alpha_\H-\alpha_\L) \log
   \left(\frac{1}{\epsilon\bar{\alpha}}\right)
   -\alpha_\H \epsilon  (\alpha_\L+\beta) \log \left(\frac{1}{\alpha_\H \epsilon
   }\right)\\
   &&+\alpha_\L \epsilon  (\alpha_\H+\beta) \log \left(\frac{1}{\alpha_\L \epsilon
   }\right)-\alpha_\H \log \left(\frac{1}{1-\epsilon\bar{\alpha}}\right)\\
   &&+\alpha_\L \log \left(\frac{1}{1-\epsilon\bar{\alpha} }\right)
   -\alpha_\L \log \left(\frac{1}{1-\alpha_\H \epsilon}\right)\\
 &&  +\alpha_\H \alpha_\L \epsilon  \log \left(\frac{1}{1-\alpha_\H \epsilon
   }\right)
   +\alpha_\H \log \left(\frac{1}{1-\alpha_\L \epsilon }\right)\\
 &&  -\alpha_\H
   \alpha_\L \epsilon  \log \left(\frac{1}{1-\alpha_\L \epsilon }\right)
   -\beta \log
   \left(\frac{1}{1-\alpha_\H \epsilon }\right)\\
   &&+\alpha_\H \beta \epsilon  \log
   \left(\frac{1}{1-\alpha_\H \epsilon }\right)
   +\beta \log \left(\frac{1}{1-\alpha_\L
   \epsilon }\right)\\
&&  -\alpha_\L \beta \epsilon  \log \left(\frac{1}{1-\alpha_\L \epsilon
   }\right).
\end{eqnarray*}
Inverting,
\begin{eqnarray*}\nonumber\lefteqn{0 =}&&\\
&&\alpha_\H \beta \epsilon  \log \left(1-\epsilon\bar{\alpha}\right)
   -\alpha_\L \beta \epsilon  \log \left(1-\epsilon\bar{\alpha}\right)\\
   &&-\beta \epsilon  (\alpha_\H-\alpha_\L) \log
   \left(\epsilon\bar{\alpha}\right)
   +\alpha_\H \epsilon  (\alpha_\L+\beta) \log \left(\alpha_\H \epsilon\right)\\
   &&-\alpha_\L \epsilon  (\alpha_\H+\beta) \log \left(\alpha_\L \epsilon\right)
   +\alpha_\H \log \left(1-\epsilon\bar{\alpha}\right)\\
   &&-\alpha_\L \log \left(1-\epsilon\bar{\alpha} \right)
   +\alpha_\L \log \left(1-\alpha_\H \epsilon\right)\\
&&   -\alpha_\H \alpha_\L \epsilon  \log \left(1-\alpha_\H \epsilon\right)
   -\alpha_\H \log \left(1-\alpha_\L \epsilon \right)\\
&&   +\alpha_\H   \alpha_\L \epsilon  \log \left(1-\alpha_\L \epsilon \right)
   +\beta \log   \left(1-\alpha_\H \epsilon \right)\\
   &&-\alpha_\H \beta \epsilon  \log
   \left(1-\alpha_\H \epsilon \right)
   -\beta \log \left(1-\alpha_\L\epsilon \right)\\
&&  +\alpha_\L \beta \epsilon  \log \left(1-\alpha_\L \epsilon\right).
\end{eqnarray*}
Gathering like terms,
\begin{eqnarray*}0&=&
\alpha_\H \beta \epsilon  \log \left(1-\epsilon\bar{\alpha}\right)
   -\alpha_\L \beta \epsilon  \log \left(1-\epsilon\bar{\alpha}\right)\\
   &&+\alpha_\H \log \left(1-\epsilon\bar{\alpha}\right)-\alpha_\L \log \left(1-\epsilon\bar{\alpha} \right)\\
   &&-\beta \epsilon  (\alpha_\H-\alpha_\L) \log
   \left(\epsilon\bar{\alpha}\right)\\
   &&+\alpha_\H \epsilon  (\alpha_\L+\beta) \log \left(\alpha_\H \epsilon\right)\\
   &&-\alpha_\L \epsilon  (\alpha_\H+\beta) \log \left(\alpha_\L \epsilon\right)
   \\
   &&
   +\alpha_\L \log \left(1-\alpha_\H \epsilon\right)
   -\alpha_\H \alpha_\L \epsilon  \log \left(1-\alpha_\H \epsilon\right)\\
   &&+\beta \log   \left(1-\alpha_\H \epsilon \right)-\alpha_\H \beta \epsilon  \log
   \left(1-\alpha_\H \epsilon \right)\\
   &&-\alpha_\H \log \left(1-\alpha_\L \epsilon \right)
   +\alpha_\H   \alpha_\L \epsilon  \log \left(1-\alpha_\L \epsilon \right)\\
&&   -\beta \log \left(1-\alpha_\L\epsilon \right)
  +\alpha_\L \beta \epsilon  \log \left(1-\alpha_\L \epsilon\right)\\
%%% start again
  &=&
%%% start again
\left(  \alpha_\H \beta \epsilon  \log 
   -\alpha_\L \beta \epsilon +\alpha_\H -\alpha_\L \right)\log \left(1-\epsilon\bar{\alpha} \right)\\
&&   -\beta \epsilon  (\alpha_\H-\alpha_\L) \log
   \left(\epsilon\bar{\alpha}\right)\\
   &&+\alpha_\H \epsilon  (\alpha_\L+\beta) \log \left(\alpha_\H \epsilon\right)-\alpha_\L \epsilon  (\alpha_\H+\beta) \log \left(\alpha_\L \epsilon\right)
   \\
   &&
   +\left(\alpha_\L 
   -\alpha_\H \alpha_\L \epsilon +\beta -\alpha_\H \beta \epsilon \right) \log
   \left(1-\alpha_\H \epsilon \right)\\
   &&+\left(-\alpha_\H   +\alpha_\H   \alpha_\L \epsilon    -\beta  +\alpha_\L \beta \epsilon\right)  \log \left(1-\alpha_\L \epsilon\right).
\end{eqnarray*}
Only the terms involving $\bar{\alpha}$ depend on $x$.  In order for equality to hold as $\epsilon\to 0^+$, we require the value $x(\epsilon)$ for which we have
\begin{align}
\nonumber\lefteqn{f(x,\epsilon)}&\\
\nonumber=&\left(  \alpha_\H \beta \epsilon   
   -\alpha_\L \beta \epsilon +\alpha_\H -\alpha_\L \right)\log \left(1-\epsilon\bar{\alpha} \right)
   \\ & -\beta \epsilon  (\alpha_\H-\alpha_\L) \log
   \left(\epsilon\bar{\alpha}\right)\\ \nonumber
  =
 &-\left( \alpha_\H \epsilon  (\alpha_\L+\beta) \log \left(\alpha_\H \epsilon\right)-\alpha_\L \epsilon  (\alpha_\H+\beta) \log \left(\alpha_\L \epsilon\right)\right.   \\
   &+\left(\alpha_\L 
   -\alpha_\H \alpha_\L \epsilon +\beta -\alpha_\H \beta \epsilon \right) \log
   \left(1-\alpha_\H \epsilon \right)\nonumber\\
&   \left.+\left(-\alpha_\H   +\alpha_\H   \alpha_\L \epsilon    -\beta  +\alpha_\L \beta \epsilon\right)  \log \left(1-\alpha_\L \epsilon\right)  
    \right)\\
    =&g(\epsilon).
\end{align}
Expanding both sides in orders of $\epsilon$, and using the expansion $\log(1+u)=u-u^2/2+O(u^3), $ as $\epsilon\to 0^+$, we have:
\begin{eqnarray*}
\lefteqn{f(x,\epsilon) =} &&\\
&&\left\{(\alpha_\H-\alpha_\L)\beta\epsilon+(\alpha_\H-\alpha_\L)\right\}\left\{-\epsilon\bar{\alpha}-\frac{\epsilon^2\bar{\alpha}^2}{2}+O\left(\epsilon^3\right)  \right\}\\
&&
-\beta\epsilon(\alpha_\H-\alpha_\L)\left\{\log\epsilon + \log\bar{\alpha}\right\}
\end{eqnarray*}
That is to say, we have the regular perturbation expansion
\begin{eqnarray*}
f(x,\epsilon)&=&\epsilon\log(\epsilon) f_0(x) + \epsilon f_1(x) + \epsilon^2 f_2(x) + O(\epsilon^3) ,
\end{eqnarray*}
as $\epsilon\to 0^+$, with
\begin{eqnarray*}
f_0(x)&=&-\beta(\alpha_\H-\alpha_\L)\\
f_1(x)&=&-(\alpha_\H-\alpha_\L)\bar{\alpha}(x) - \beta(\alpha_\H-\alpha_\L)\log(\bar{\alpha}(x))\\
f_2(x)&=&-\bar{\alpha}(x)(\alpha_\H-\alpha_\L)\beta - \frac{\bar{\alpha}(x)^2}{2}(\alpha_\H-\alpha_\L).
\end{eqnarray*}
Note that $f_0$ does not, in fact, depend on $x$.
For the right hand side we have:
\begin{eqnarray*}
\lefteqn{g(\epsilon)} &&\\
&=& 
- \alpha_\H \epsilon  (\alpha_\L+\beta) \log \left(\alpha_\H \epsilon\right)
+ \alpha_\L \epsilon  (\alpha_\H+\beta) \log \left(\alpha_\L \epsilon\right)   \\
   &&
   -\left(\alpha_\L   -\alpha_\H \alpha_\L \epsilon +\beta -\alpha_\H \beta \epsilon \right) \log   \left(1-\alpha_\H \epsilon \right)\\
&&   -\left(-\alpha_\H   +\alpha_\H   \alpha_\L \epsilon    -\beta  +\alpha_\L \beta \epsilon\right)  \log \left(1-\alpha_\L \epsilon\right)  \\
   &=&
   - \alpha_\H \epsilon  (\alpha_\L+\beta) \log \left(\alpha_\H \epsilon\right)
+ \alpha_\L \epsilon  (\alpha_\H+\beta) \log \left(\alpha_\L \epsilon\right)   \\
   &&
+   \left(-\alpha_\L   +\alpha_\H \alpha_\L \epsilon -\beta +\alpha_\H \beta \epsilon \right) \log   \left(1-\alpha_\H \epsilon \right)\\
&&  + \left(\alpha_\H   -\alpha_\H   \alpha_\L \epsilon    +\beta  -\alpha_\L \beta \epsilon\right)  \log \left(1-\alpha_\L \epsilon\right)  \\
&=&
   - \alpha_\H \epsilon  (\alpha_\L+\beta) \log \left(\alpha_\H \epsilon\right)\\
&&+ \alpha_\L \epsilon  (\alpha_\H+\beta) \log \left(\alpha_\L \epsilon\right)  \\
&&  + \left(-\alpha_\L   +\alpha_\H \alpha_\L \epsilon -\beta +\alpha_\H \beta \epsilon \right) \log   \left(1-\alpha_\H \epsilon \right)\\
&& +  \left(\alpha_\H   -\alpha_\H   \alpha_\L \epsilon    +\beta  -\alpha_\L \beta \epsilon\right)  \log \left(1-\alpha_\L \epsilon\right)  \\
&=&
   - \alpha_\H   (\alpha_\L+\beta) \log \left(\alpha_\H \right)\epsilon
   - \alpha_\H   (\alpha_\L+\beta) \epsilon\log \left( \epsilon\right)\\
&&
+ \alpha_\L   (\alpha_\H+\beta) \log \left(\alpha_\L \right)\epsilon
+ \alpha_\L   (\alpha_\H+\beta) \epsilon\log \left( \epsilon\right)  \\
&&  
+ \left(-\alpha_\L   -\beta+(\alpha_\H \alpha_\L   +\alpha_\H \beta) \epsilon \right) 
\left\{-\alpha_\H \epsilon - \frac{\alpha_\H^2\epsilon^2}{2} + O(\epsilon^3) \right\} \\ % \log   \left(1-\alpha_\H \epsilon \right)
&& +  \left(\alpha_\H  +\beta  -(\alpha_\H   \alpha_\L      + \alpha_\L \beta )\epsilon\right)  
\left\{-\alpha_\L \epsilon - \frac{\alpha_\L^2\epsilon^2}{2} + O(\epsilon^3) \right\},
%\text{ as }\epsilon\to 0^+.
\\% \log \left(1-\alpha_\L \epsilon\right)  
\end{eqnarray*}
as $\epsilon \rightarrow 0^+$.
Therefore, as $\epsilon\to 0^+$, we have 
\begin{eqnarray*}
g(\epsilon)&=&\epsilon\log(\epsilon) g_0 + \epsilon g_1 + \epsilon^2 g_2 + O(\epsilon^3),\text{ as }\epsilon\to 0^+,\text{ with}\\
g_0&=& -\alpha_\H(\alpha_\L+\beta) +\alpha_\L(\alpha_\H+\beta)      \\
g_1&=&  - \alpha_\H   (\alpha_\L+\beta) \log \left(\alpha_\H \right) 
               + \alpha_\L   (\alpha_\H+\beta) \log \left(\alpha_\L \right)\\&&    
               - \alpha_\H (-\alpha_\L   -\beta) - \alpha_\L(\alpha_\H  +\beta)\\
g_2&=&  -\alpha_\H(\alpha_\H \alpha_\L   +\alpha_\H \beta)
               +\alpha_\L(\alpha_\H   \alpha_\L      + \alpha_\L \beta ) \\&&
               +(-\alpha_\L-\beta)\left(-\frac{\alpha_\H^2}{2}\right)
               +(\alpha_\H+\beta)\left(-\frac{\alpha_\L^2}{2}\right).
\end{eqnarray*}
Comparing the terms of order $\epsilon\log\epsilon$, we see that $f_0=-\beta(\alpha_\H-\alpha_\L)=g_0$ holds independently of $x$.

Moving to the $O(\epsilon)$ terms, we require $x\in(0,1)$ for which $f_1(x)=g_1$.  That is, we require that 
\begin{eqnarray*}
f_1(x)&=&-(\alpha_\H-\alpha_\L)\bar{\alpha}(x) - \beta(\alpha_\H-\alpha_\L)\log(\bar{\alpha}(x))\\
&=& - \alpha_\H   (\alpha_\L+\beta) \log \left(\alpha_\H \right) 
               + \alpha_\L   (\alpha_\H+\beta) \log \left(\alpha_\L \right)\\&&    
               - \alpha_\H (-\alpha_\L   -\beta) - \alpha_\L(\alpha_\H  +\beta)\\
&=&g_1.
\end{eqnarray*}
If we introduce the function $\psi(x)=\bar{\alpha}(x)+\beta(1+\log(\bar{\alpha}(x)))$,
and a constant
\begin{equation}
\mathcal{G}=\frac{ \alpha_\H   (\alpha_\L+\beta) \log \left(\alpha_\H \right) 
               - \alpha_\L   (\alpha_\H+\beta) \log \left(\alpha_\L \right)    
                }{\alpha_\H-\alpha_\L},
\end{equation}
 then we have the equivalent requirement on $x$:
\begin{equation}\label{eq:phi-of-x}
\psi(x)= \mathcal{G}.
\end{equation}
Since $\bar{\alpha}(x)=\alpha_\H x+\alpha_\L (1-x)$, we have
$$\frac{d\psi(x)}{dx}=(\alpha_\H-\alpha_\L)\left(1+\frac{\beta}{\bar{\alpha}}\right) > 0,$$
and so $\psi$ is monotonically increasing on $(0,1)$, and has a smooth inverse $\psi^{-1}$.  The range of $\psi$ is $\psi(0)=\alpha_\L+\beta(1+\log\alpha_\L)< \psi(x) < \psi(1)=\alpha_\H+\beta(1+\log\alpha_\H)$, so Equation (\ref{eq:phi-of-x}) has a unique solution, provided $\mathcal{G}$ lies in this range.  To check, we need to verify that
\begin{eqnarray*}
&\alpha_\L+\beta(1+\log\alpha_\L)&\\
&<\frac{ \alpha_\H   (\alpha_\L+\beta) \log \left(\alpha_\H \right) 
               - \alpha_\L   (\alpha_\H+\beta) \log \left(\alpha_\L \right)    
                }{\alpha_\H-\alpha_\L}<&\\
&\alpha_\H+\beta(1+\log\alpha_\H).&
\end{eqnarray*}
Upon assuming that $\beta>0$ and $0\le\alpha_\L<\alpha_\H\le1$, and setting $y=(\alpha_\H-\alpha_\L)/\alpha_\H$, these inequalities reduce to showing that 
$$\frac{y}{1+y}<\log(1+y)<y$$
for $0<y\le 1$, which are readily verified.

This calculation shows that, for small $\epsilon>0$, there can only be one maximum for $\tilde{I}_\epsilon(x)$ in the interior of the unit interval. Moreover,  $\tilde{I}_\epsilon(x)\to 0$ for $x\to0$ and $x\to 1$, and $\tilde{I}_\epsilon(x)>0$ for $0<x<1$.  Therefore,  $\tilde{I}_\epsilon$ has to have at least one maximum, but it can have at most one critical point (by the preceding argument) so it has a unique maximum.

The corresponding value of $x$ will be the asymptotically optimal value $x_\text{opt}$, as $\epsilon\to 0^+$.

\subsection{Capacity- and mutual-information--maximizing parameter values}
\label{sec:CapMaxParameters}

\dr{In \S \ref{sec:Capacity}, Equation \ref{eqn:CIIDLH}  gives the mutual information for the discrete time BIND channel.  Here we show that the mutual information is bounded with respect to the three channel parameters $\alpha_\L, \alpha_\H$ and $\beta$, and that the capacity is an increasing
function of $\beta$ and $\alpha_\H$, and is decreasing in $\alpha_L$.}
(\dr{Consequently, for a fixed time step, the extremizing values of} these parameters \dr{all} equal  either 0 or 1, thus violating the strict ordering assumption.)

Dropping the maximization from (\ref{eqn:CIIDLH}), mutual information is written
\begin{equation}
	\label{eqn:MutualInfo}
	I(X;Y) = 
	\frac{\binent(\alpha_\H p_\H + \alpha_\L p_\L) 
		- p_\H \binent(\alpha_\H) - p_\L \binent(\alpha_\L)}
	{1 + (p_\H \alpha_\H + p_\L \alpha_\L)/\beta} .
\end{equation}
We assume that the parameters are strictly ordered (see Definition \ref{defn:StrictlyOrdered}).
With this assumption, $I(X;Y) > 0$; therefore, the same is true of the
numerator in (\ref{eqn:MutualInfo}), since the denominator is positive. 
Also note that
\begin{equation} %%%% IS THIS IN NATS??? ... no ... this is okay since log is base 2
	\frac{d}{dp} \binent(p) = \log \frac{1-p}{p} .
\end{equation}
We will use these properties below.

First consider $\beta$. By inspection of (\ref{eqn:MutualInfo}), $\beta$ only appears in the denominator,
and the denominator decreases with increasing $\beta$. Thus, $I(X;Y)$ is increasing in $\beta$,
and $\beta = 1$ is optimal.

%\subsection{Effect of $\alpha_\L$}

Now consider $\alpha_\L$.
By inspection of (\ref{eqn:MutualInfo}), the denominator is increasing in $\alpha_\L$.
We can show that the numerator is decreasing in $\alpha_\L$: we can write
\begin{align}
	\nonumber
	\lefteqn{\frac{d}{d \alpha_\L} \left( \binent(\alpha_\H p_\H + \alpha_\L p_\L) 
		- p_\H \binent(\alpha_\H) - p_\L \binent(\alpha_\L) \right)} \\
%	&= p_\L \log \frac{1 - \alpha_\H p_\H - \alpha_\L p_\L}{\alpha_\H p_\H + \alpha_\L p_\L}
%		- p_\L \log \frac{1 - \alpha_\L}{\alpha_\L} \\
	&= p_\L \log \frac{(1 - \alpha_\H p_\H - \alpha_\L p_\L) \alpha_\L}
		{(\alpha_\H p_\H + \alpha_\L p_\L)(1-\alpha_\L)} \\
	\label{eqn:Derivative}
	&= p_\L \log \frac{\alpha_\L - \alpha_\L(\alpha_\H p_\H + \alpha_\L p_\L)}
		{\alpha_\H p_\H + \alpha_\L p_\L - \alpha_\L(\alpha_\H p_\H + \alpha_\L p_\L)} \\
	&\leq 0 ,
\end{align}
where the final inequality follows since
$\alpha_\L \leq \alpha_\H p_\H + \alpha_\L p_\L$ (since $\alpha_\L \leq \alpha_\H$).
Thus, $I(X;Y)$ is decreasing in $\alpha_\L$, and $\alpha_\L = 0$ is optimal.

Finally, consider $\alpha_\H$: this case is slightly trickier than $\alpha_\L$, since 
both the numerator and denominator of (\ref{eqn:MutualInfo}) are increasing.
For simplicity, we start by substituting $\beta = 1$ and $\alpha_\L = 0$: we have
\begin{equation}
	I(X;Y) = \frac{\binent(\alpha_\H p_\H) 
		- p_\H \binent(\alpha_\H)}
	{1 + p_\H \alpha_\H} .
\end{equation}
To show that this quantity is increasing with $\alpha_\H$,
the first derivative with respect to $\alpha_\H$ is
\begin{align}
	\nonumber\lefteqn{\frac{d}{d\alpha_\H} I(X;Y) 
		= }&\\
\label{eqn:DerivativeAH}
		&\frac{(1+p_\H \alpha_\H) 
			p_\H \log\frac{(1-\alpha_\H p_\H)}{(1-\alpha_\H)p_\H}
			- p_\H (\binent(\alpha_\H p_\H) 
		- p_\H \binent(\alpha_\H))}{(1+\alpha_\H p_\H)^2} .
\end{align}
The goal is to determine whether $dI(X;Y)/d\alpha_\H$ is positive. 
It is useful to write
\begin{align}
	\nonumber\lefteqn{\binent(\alpha_\H p_\H) 
		- p_\H \binent(\alpha_\H)
	=}&\\
	&\log \frac{(1-\alpha_\H)^{p_\H}}{1-\alpha_\H p_\H}
		+ \alpha_\H p_\H \log \frac{(1-\alpha_\H p_\H)}{(1-\alpha_\H)p_\H} .
\end{align}
Thus, 
(\ref{eqn:DerivativeAH}) becomes
\begin{align}
	\nonumber\lefteqn{ \frac{d}{d\alpha_\H} I(X;Y)}&\\
	\nonumber =& \frac{p_\H}{(1+\alpha_\H p_\H)^2} \Bigg[(1+p_\H \alpha_\H) 
			\log\frac{(1-\alpha_\H p_\H)}{(1-\alpha_\H)p_\H}\\
	& - %\frac{p_\H}{(1+\alpha_\H p_\H)^2} \left( 
			\log \frac{(1-\alpha_\H)^{p_\H}}{1-\alpha_\H p_\H}
		- \alpha_\H p_\H \log \frac{(1-\alpha_\H p_\H)}{(1-\alpha_\H)p_\H}
		%\right)
		\Bigg] \\
	\label{eqn:DerivativeAH2}
	=& \frac{p_\H}{(1+\alpha_\H p_\H)^2}
		\log\frac{(1-\alpha_\H p_\H)^2}{(1-\alpha_\H)^{p_\H+1}p_\H} .
\end{align}
By inspection of (\ref{eqn:DerivativeAH2}), the derivative is positive when
\begin{equation}
	\label{eqn:DerivativeAH3}
	\frac{(1-\alpha_\H p_\H)^2}{(1-\alpha_\H)^{p_\H+1}p_\H} \geq 1 .
\end{equation}
Inequality (\ref{eqn:DerivativeAH3}) is satisfied for $\alpha_\H = 0$ (as $1/p_\H \geq 1$); 
to show that
it is satisfied for all strictly ordered $\alpha_\H$, we show that the left side of (\ref{eqn:DerivativeAH3}) is
increasing for $\alpha_\H \geq 0$. After some manipulation, we have
\begin{align}
	\nonumber \lefteqn{\frac{d}{d\alpha_\H} 
		\frac{(1-\alpha_\H p_\H)^2}{(1-\alpha_\H)^{p_\H+1}p_\H}} & \\
%	&= \frac{p_\H(1-\alpha_\H)^{p_\H}(1-\alpha_\H p_\H) }
%		{((1-\alpha_\H)^{p_\H+1}p_\H)^2}
%		((1+p_\H)(1 - \alpha_\H p_\H) - 2 p_\H (1-\alpha_\H)) .
%\end{align}
%
%The fractional term is positive by inspection. The remaining term gives
%%
%\begin{align}
%	\nonumber \lefteqn{(1+p_\H)(1 - \alpha_\H p_\H) - 2 p_\H (1-\alpha_\H)} & \\
%	&= 1 - \alpha_\H p_\H + p_\H- \alpha_\H p_\H^2 - 2p_\H + 2 \alpha_\H p_\H \\
%	&= 1 + \alpha_\H p_\H - p_\H - 	\alpha_\H p_\H^2 \\
%	&= (1-p_\H)(1+\alpha_\H p_\H)
%\end{align}
		&= \frac{p_\H(1-\alpha_\H)^{p_\H}(1-\alpha_\H p_\H) (1-p_\H)(1+\alpha_\H p_\H)}
		{((1-\alpha_\H)^{p_\H+1}p_\H)^2} ,
\end{align}
which is positive for all strictly ordered parameters. Thus, 
$I(X;Y)$ is increasing in $\alpha_\H$, and
$\alpha_\H =1$ is optimal.

The proceeding analysis is true for any valid setting of $p_\H$ and $p_\L$. Therefore, it applies to capacity as well as mutual information. 

\dr{This optimization calculation applies to the discrete time model for any (fixed) time step.  Within the framework of the continuous time BIND channel (\S \ref{sec:continuoustime}), there is no \textit{a priori} upper limit on the reaction rate constants $k_+$ and $k_-$. The calculation in this Appendix thus shows that, \textit{ceteris paribus}, a ligand-receptor system would have a higher capacity, the faster its binding rate $k_+$ and  its unbinding rate $k_-$, provided it could toggle the ligand concentration arbitrarily close to zero when sending the ``low" input signal.  Thermodynamic and other physical limitations prevent channels from obtaining arbitrarily large binding and unbinding rates, and reducing the signal concentration strictly to zero is generally not possible in biological systems.  The practical limits in specific signaling systems provide appealing topics for future investigation.}

\section*{Acknowledgments}

The authors are grateful for early support for this project from Terrence J.~Sejnowski and the Howard Hughes Medical Institute. We also thank Toby Berger, Tom Bartol, Hillel Chiel, Patrick Fitzsimmons, Peter Kotelenez, Marshall Leitman, Andries Lenstra, Vladimir I.~Rotar, Robin Snyder, and Elizabeth Wilmer for helpful discussions.  PJT thanks the Oberlin College Library for research support.

\bibliographystyle{IEEEtran}
%\bibliography{infotheory,Dicty,stoch_chem,neuroscience,MCell,PJT,signaling,math}

\begin{IEEEbiography}[{\includegraphics[width=1in,height=1.25in,clip,keepaspectratio]{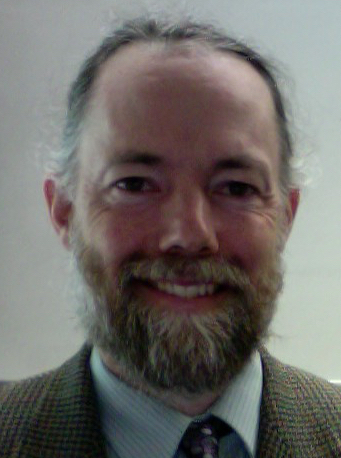}}]{Peter J.~Thomas}
studies communication and control in complex adaptive biological systems.
He obtained an M.A.~in the Conceptual Foundations of Science and 
a Ph.D.~in Mathematics from The University of Chicago in 2000. Following postdoctoral
work in the Computational Neurobiology Laboratory at The Salk Institute for
Biological Studies, he taught mathematics, neuroscience, and computational biology first at Oberlin College and now at Case Western Reserve University, where
he is Associate Professor of Mathematics, Applied Mathematics, and Statistics.
He holds secondary appointments in Biology, in Cognitive Science, and in Electrical Engineering and Computer Science. 
\end{IEEEbiography}

\begin{IEEEbiography}[{\includegraphics[width=1in,height=1.25in,clip,keepaspectratio]{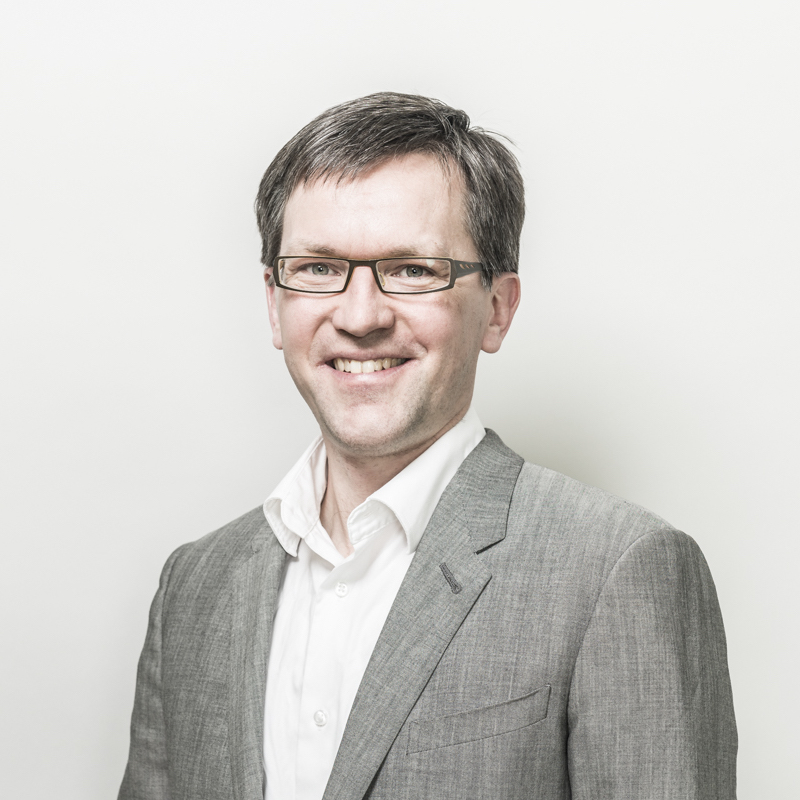}}]
{Andrew W. Eckford} received the B.Eng. degree from the Royal Military College of Canada, in 1996, and the M.A.Sc. and Ph.D. degrees from the University of Toronto, in 1999 and 2004, respectively, all in electrical engineering. He is an Associate Professor in the Department of Electrical Engineering and Computer Science at York University, Toronto, Ontario. He held postdoctoral fellowships at the University of Notre Dame and the University of Toronto, prior to taking up a faculty position at York in 2006. His research interests include the application of information theory to nonconventional channels and systems, especially the use of molecular and biological means to communicate. Dr. Eckford's research has been covered in media including The Economist, The Wall Street Journal, and IEEE Spectrum. He is also a co-author of the textbook Molecular Communication, published by Cambridge University Press, and was a finalist for the 2014 Bell Labs Prize.
\end{IEEEbiography}

\end{document}